\documentclass[a4paper,reqno]{amsart}
\usepackage{amsmath,amsthm, amssymb}
\usepackage{mathrsfs}
\usepackage{mathtools}
\usepackage{color}
\usepackage{enumerate}
\usepackage{hyperref}

\mathtoolsset{showonlyrefs=true}
		
\newcommand{\N}{\mathbb{N}}
\newcommand{\R}{{\mathbb{R}}}
\newcommand{\C}{{\mathbb{C}}}
\newcommand{\Z}{{\mathbb{Z}}}
\newcommand{\dd}{{{\rm d}}}
\newcommand{\ii}{{\rm i}}
\renewcommand{\P}{{\mathcal{P}}}
\newcommand{\T}{{\mathcal{T}}}
\newcommand{\PT}{{\mathcal{PT}}}

\newcommand{\cC}{{\mathcal{C}}}
\newcommand{\cS}{{\mathcal{S}}}

\renewcommand{\H}{{\mathcal{H}}}
\newcommand{\X}{{\mathcal{Y}}}
\newcommand{\HsQ}{{\H^s_Q}}
\newcommand{\spd}{\sigma_{\rm disc}}
\newcommand{\spe}{\sigma_{\rm ess}}

\newcommand{\s}{\varkappa}
\newcommand{\Dom}{{\operatorname{Dom}}}
\newcommand{\Ker}{{\operatorname{Ker}}}
\newcommand{\Ran}{{\operatorname{Ran}}}
\newcommand{\lspan}{{\operatorname{span}}}

\renewcommand{\Re}{\operatorname{Re}}
\renewcommand{\Im}{\operatorname{Im}}
\newcommand{\cf}{\emph{cf.}}
\newcommand{\ie}{{\emph{i.e.}}}
\newcommand{\eg}{{\emph{e.g.}}}
\newcommand{\dist}{\mathrm{dist}}

\newcommand{\eps}{\varepsilon}

\newcommand{\chiso}{\chi(\sigma_1)}
\newcommand{\chist}{\chi(\sigma_2)}
\newcommand{\ov}{\overline}

\begin{document}

\theoremstyle{plain}
\newtheorem{define}{Definition}
\newtheorem{theorem}{Theorem}[section]
\newtheorem{lemma}[theorem]{Lemma}
\newtheorem{criterion}[theorem]{Criterion}
\newtheorem{proposition}[theorem]{Proposition}
\newtheorem{corollary}[theorem]{Corollary}
\renewcommand{\proofname}{Proof}

\theoremstyle{definition}
\newtheorem{example}[theorem]{Example}
\newtheorem{remark}[theorem]{Remark}
\newtheorem{ass}{Assumption}
\renewcommand\theass{{\upshape{(\Roman{ass})}}}
%


\title[Bifurcation with antilinear symmetry]{Bifurcation of nonlinear eigenvalues in problems with antilinear symmetry}

\author{Tom\'a\v{s} Dohnal}
\address[Tom\'a\v{s} Dohnal]{
Fachbereich Mathematik, Technical University Dortmund,
Vogelpothsweg 87, 
44221 Dortmund, Germany}
\email{tomas.dohnal@math.tu-dortmund.de}

\author{Petr Siegl}
\address[Petr Siegl]{Mathematisches Institut, Universit\"at Bern, Sidlerstrasse 5, 3012 Bern, Switzerland \& On leave from Nuclear Physics Institute ASCR, 25068 \v Re\v z, Czech Republic}
\email{petr.siegl@math.unibe.ch}



\date{April 29, 2015}

\begin{abstract}
Many physical systems can be described by nonlinear eigenvalues and bifurcation problems with a linear part that is non-selfadjoint \eg~due to the presence of loss and gain. The balance of these effects is reflected in an antilinear symmetry, like \eg~the $\PT$-symmetry, of the problem. Under this condition we show that the nonlinear eigenvalues bifurcating from real linear eigenvalues remain real and the corresponding  nonlinear eigenfunctions remain symmetric. The abstract results are applied in a number of physical models of Bose-Einstein condensation, nonlinear optics and superconductivity, and further numerical analysis is performed.
\end{abstract}

\thanks{
P.S. thanks G. Wunner and H. Cartarius for drawing his attention to this topic. The research of P.S. is supported by the \emph{Swiss National Foundation}, SNF Ambizione grant No. PZ00P2\_154786. The research of T.D. is partly supported by the \emph{German Research Foundation}, DFG grant No.  DO1467/3-1.
}

\maketitle

\section{Introduction}\label{S:intro}

We consider the nonlinear eigenvalue problem 
\begin{equation}\label{nl.ev}
A \psi - \varepsilon f(\psi) = \mu \psi, 
\end{equation}
and analyze the bifurcation in $\eps$ from a simple eigenvalue $\mu_0$ at $\eps =0$ in a suitable Hilbert space for a rather general class of densely defined, closed (possibly non-selfadjoint) operators $A$ and locally Lipschitz continuous nonlinearities $f$, \cf~Assumption \ref{ass:Af} below for details. For a homogeneous nonlinearity $f$, we also consider the additional condition $\|\psi\|=1$. The main contribution of our paper is to the problem of bifurcation from real eigenvalues under an antilinear symmetry of $A$ and $f$. We show that under this condition the nonlinear eigenvalue remains real and the eigenfunction remains symmetric. This confirms a number of existing numerical computations of specific examples of such a bifurcation problem with an antilinear symmetry, see the references below. Besides presenting the abstract bifurcation results, we explain in detail how these
apply to physically relevant examples by checking the assumptions and giving concrete choices of the working space.

The question of real nonlinear eigenvalues in non-selfadjoint problems with symmetries has gained on physical relevance in the recent years due to the intensive research on nonlinear systems under the parity and time-reversal ($\PT$) symmetry mainly in Bose-Einstein condensates (BECs) \cite{Klaiman-2008-101,Kartashov-2014-107}, nonlinear optics \cite{Guo-2009-103}, see also \cite{Ruter-2010-6} for an experimental breakthrough, or superconductivity \cite{Rubinstein-2007-99}. 
In these specific physical problems, the presence of real nonlinear eigenvalues typically means the existence of stationary solutions of the form $e^{- \ii \mu t} \psi(x)$ with $\mu \in \R$ also if the system is subject to balanced gain and loss (modeled by a non-selfadjoint linear part).

The interest in antilinear symmetries was initiated by an observation in \cite{Bender-1998-80} where Schr\"odinger operators with $\PT$-symmetric complex potentials in the context of quantum-mechanics-like linear problems were numerically shown to have real eigenvalues in a certain parameter region. 

In the context of BECs, where the (nonlinear) Gross-Pitaevskii equation models the dynamics of the condensate, a complex potential describes the injection and removal of particles and a balance of these two processes is reflected in the $\PT$-symmetry of the system. Numerical results on the bifurcation of nonlinear eigenvalues in particular in one dimensional models can be found \eg~in \cite{Cartarius-2012-45,Dast-2013-46,Fortanier-2014-89}. 

In optics under the paraxial approximation, the system can be modeled by the nonlinear Schr\"odinger equation (NLS) with a potential corresponding to the refractive index, which is complex if the amplification and damping of the light wave are present, a balance is again reflected in the $\PT$-symmetry. Numerical and formal results on the NLS for the bifurcation from linear eigenvalues under $\PT$-symmetry include one dimensional \cite{Musslimani-2008-100,Wang-2011-19,Zezyulin-2012-85} or two dimensional \cite{Yang-2014-39}.

Superconducting wires driven with electric currents represent another example of a physical application of a $\PT$-symmetric nonlinear eigenvalue problem, \cf~\cite{Rubinstein-2007-99,Rubinstein-2010-195}.  
The non-selfadjointness appears due to the dependence of the electric potential on the external current.

As explained in detail in Section \ref{sec:appl}, our results cover all of the above physical models as they are particular cases of \eqref{nl.ev} with a linear operator $A$, a nonlinearity $f$ and an antilinear symmetry $\cC$ compliant with the assumptions of our analysis.

Our approach to the bifurcation problem \eqref{nl.ev} is based on the Lyapunov-Schmidt reduction and a fixed point iteration. We decompose the Hilbert space to the one dimensional $\ker(A-\mu_0)$ and its complement using the spectral projection corresponding to the eigenvalue $\mu_0$ and for $\eps>0$ we seek solutions $(\mu,\psi)$ near $(\mu_0,\psi_0)$, where $\psi_0$ is the eigenfunction of $A$ corresponding to $\mu_0$. On the complement of $\ker(A-\mu_0)$ the operator $A-\mu_0$ is invertible and a fixed point iteration can be used to obtain a small correction of the eigenfunction for $\eps$ small enough.  The scalar equation on $\ker(A-\mu_0)$ is solved likewise by a fixed point iteration and it produces a small correction of $\mu_0$. In this way we obtain an expansion of $\mu$ and $\psi$ up to second order in $\eps$.

The problem of bifurcation of nonlinear eigenvalues is of course classical and has been solved, \eg~in \cite{Crandall-1971-8} for simple eigenvalues in real Banach spaces and in \cite{Ize-1976-7} for possibly complex Banach spaces (as relevant in our problem) and for eigenvalues of odd algebraic multiplicity or for geometrically simple eigenvalues, see \cite[Thm.I.3.2]{Ize-1976-7}. We choose to prove our results in a Hilbert space 
and independently of \cite{Ize-1976-7} in order to provide an explicit expansion of $\mu$ and $\psi$ and because the Lyapunov-Schmidt reduction and the fixed point equations are used also in the second part on the preservation of the realness of $\mu$ under an antilinear symmetry condition. Moreover, we avoid the technical condition $\|f(\psi)\|=O(\|\psi\|^2)$ as $\psi\to 0$ of \cite{Ize-1976-7}.

For the problem under the assumption of an antilinear symmetry $\cC$ of $A$ and $f$, \cf~Assumption \ref{ass:A.C}, we check that the fixed point iteration preserves the symmetry of the iterates for $\psi$ and the realness of the iterates for $\mu$. As a result, if $\mu_0\in \R$ and if $\psi_0$ has the antilinear symmetry, then the nonlinear eigenpair $(\mu,\psi)$ satisfies these conditions for $\eps$ small enough too.

To our knowledge the only existing mathematically rigorous papers on similar bifurcation problems under an antilinear symmetry are \cite{Kevrekidis-2013-12} and \cite{Rubinstein-2010-195}. In \cite{Kevrekidis-2013-12} the concrete example of the discrete NLS with the $\PT$-symmetry is considered. The proof is based on the Lyapunov-Schmidt reduction and the implicit function theorem. In \cite{Rubinstein-2010-195} a one dimensional $\PT$-symmetric nonlinear parabolic problem for superconducting wires is studied using the center manifold analysis. As a special case, stationary localized solutions are found.

The structure of the paper is as follows. Section \ref{sec:main.ass} presents the assumptions on the operator $A$ and the nonlinearity $f$ in \eqref{nl.ev} needed for the general bifurcation problem and provides a number of examples of $A$ and $f$ satisfying these conditions, whereby we concentrate mainly on Schr\"odinger operators $A$ but discuss also a first order Dirac type operator. Section \ref{sec:main.ass}  also explains the choice of our function space in which the fixed point iteration is carried out. In Section \ref{sec:NL_ev} we prove the bifurcation result and the expansion of the eigenvalue and the eigenfunction. For homogeneous nonlinearities, we rescale $\eps$ and $\psi$ such that a solution with $\|\psi\|=1$ is found. The problem under symmetry assumptions is discussed in Section \ref{sec:sym}. Both antilinear and linear symmetries are discussed, where the former one is shown to lead to the preservation of the realness of $\mu$. Section \ref{sec:appl} explains applications of our results to concrete physical problems from literature. Finally, in Section \ref{sec:num} we present numerical computations of nonlinear eigenvalues of \eqref{nl.ev} with $A=-\Delta +V$ and $f(\psi)=|\psi|^2\psi$ in $L^2(\R^2)$ and with $\PT$-symmetric as well as partially $\PT$-symmetric potentials $V$. The effects of a linear symmetry are also observed.

%
%
\section{Basic assumption and examples of operators and nonlinearities} 
\label{sec:main.ass}

The following basic assumption comprises a condition on a compatibility of the linear part $A$ with the nonlinearity $f$ and a spectral condition on $A$.

\begin{ass}\label{ass:Af} 
	\noindent
	Let $A$ be a densely defined, closed operator with a non-empty resolvent set in a Hilbert space $(\H, \langle \cdot, \cdot\rangle)$ with the induced norm $\|\cdot\|$, let $f$ be a mapping in $\H$ and let $(\X, \|\cdot\|_{\X})$ be a Banach space. Suppose that the following conditions are satisfied:
	\begin{enumerate}[(a)]
		%
		\item\label{ass:Af.norm}
		$\X$ is a subspace of $\H$, for some $n\in \N$ is $\Dom(A^n) \subset \X \subset \Dom(A^{n-1})$, and there are $k_1 ,k_2>0$ such that, for all $\phi \in \Dom(A^n)$,
		\begin{equation}\label{ass.f.1.gr}
		\|\phi\|_{n-1}:= \sum_{k=0}^{n-1}\| A^{k} \phi \|  \leq   k_1 \|\phi\|_\X 
		\leq k_2 
		\sum_{k=0}^{n} \|A^k \phi \| =: k_2 \|\phi\|_n,
		\end{equation}
		\item\label{ass:Af.mu}
		$\mu_0 \in \C$ is an isolated simple (\ie~with the algebraic multiplicity one) eigenvalue of $A$. Moreover, suppose that the normalizations of $\psi_0 \in \Dom(A)$, $\psi_0^* \in \Dom(A^*)$, 
		\begin{equation}\label{psi0.psi0*.def}
		A \psi_0 = \mu_0 \psi_0, \qquad A^* \psi_0^* = \ov{\mu_0} \psi_0^*,
		\end{equation}
		\ie~of the eigenvectors of $A$, $A^*$ corresponding to $\mu_0$, $\ov{\mu_0}$, respectively, are chosen such that
		\begin{equation}\label{psi0.norm}
		\|\psi_0\| =1, \qquad \langle \psi_0, \psi_0^*  \rangle =1,
		\end{equation}
		\item\label{ass:Af.Lip} 
		the mapping $f : \X \to \Dom(A^{n-1})$ is Lipschitz in a neighborhood of the eigenvector $\psi_0$, more precisely: there exist $r_L >0$ and $L >0$ such that, for all $\phi,\psi \in \{ \eta \in \X \, : \, \|\eta - \psi_0 \|_\X < r_L\}$,
		\begin{equation}\label{E:Lip}
		\| f(\phi) - f(\psi) \|_{n-1} \leq k_1 L  \|\phi - \psi\|_\X.
		\end{equation}
	\end{enumerate}
\end{ass}

\begin{remark}[Remarks on Assumption \ref{ass:Af}]
The space $\X$ is our working space in which we perform fixed point iterations. A natural choice for $\X$ is $(\Dom(A),\|\cdot\|_1)$, \ie~the domain of $A$ equipped with its graph norm. Nonetheless, it may be convenient to work also with a different $\X$, \eg~with the form-domain and the norm induced by the quadratic form of $A$ since these can be much better accessible than $(\Dom(A),\|\cdot\|_1)$ itself, \cf~Section \ref{subsec:Schr.op} for examples. Obviously, if the Lipschitz continuity \eqref{E:Lip} is established with $\|\cdot\|_{\X}$, it holds also with $\|\cdot\|_n$. 
A motivation for considering $n>1$ is given in Examples \ref{ex.A.Schr.n} and \ref{ex.f.pol}, see also Remark \ref{R:n}. The condition $\rho(A) \neq \emptyset$ guarantees that also $\Dom (A^n)$, $n>1$, is dense in $\H$, therefore also $\X$ is dense in $\H$. 

Recall that if $\mu_0$ is a simple isolated eigenvalue of $A$, then $\ov{\mu_0}$ is a simple isolated eigenvalue of $A^*$, \cf~\cite[Chap.III.6.5-6]{Kato-1966}; moreover, it can be easily verified that the normalization \eqref{psi0.norm} can be achieved. In detail, $\langle \psi_0,\psi_0^* \rangle =0$ implies that $\psi_0 \in \Ker(A-\mu_0) \cap \Ran(A-\mu_0)$ since $\Ker(A^*-\mu_0)^\perp = \Ran(A-\mu_0)$. The eigenvalue $\mu_0$ is simple, so $\Ker(A-\mu_0)^2 = \Ker(A-\mu_0)$ in particular. From $(A-\mu_0)\phi=\psi_0$ for some $\phi\in\Dom(A)$ we get $(A-\mu_0)^2\phi=0$, thus $\phi = c\psi_0$ and hence $\psi_0=0$, which is a contradiction. 

The spectral (Riesz) projection $P_0$ on $\Ker(A-\mu_0)$, defined as a contour integral for sufficiently small $\delta>0$, \cf~\cite[Chap.III.6]{Kato-1966}, and the complementary projection $Q_0:=I - P_0$ can be written, using \eqref{psi0.norm}, as 
\begin{equation}\label{P0.def}
P_0 := - \frac 1 {2\pi \ii} \int_{\partial B_\delta(\mu_0)} (A-z)^{-1} \dd z = 
\langle \cdot, \psi_0^* \rangle \psi_0,
\quad Q_0 = I - \langle \cdot, \psi_0^* \rangle \psi_0.
\end{equation}
\end{remark}

We analyze several groups of operators and nonlinearities below and show that they satisfy Assumption \ref{ass:Af}. The selection is inspired by various physical models from literature, \cf~Section \ref{sec:appl}, where we apply our results to problems possessing typically additional symmetries, \cf~Section \ref{sec:sym}.


\subsection{Schr\"odinger operators and the space $\X$ in Ass.~\ref{ass:Af}.(\ref{ass:Af.norm})}
\label{subsec:Schr.op}

Schr\"odinger operators are naturally associated with the following spaces 
\begin{equation}\label{H.X.Schr}
\begin{aligned}
(\X,\|\cdot\|_\X) & = (\H^s_Q(\Omega), \| \cdot \|_\HsQ) 
:= ((H^s(\Omega) \cap \Dom(Q)), \|\cdot\|_{H^s} + \|Q\cdot\|),
\end{aligned}
\end{equation}	
where $\Omega$ is a domain in $\R^d$, $s >0 $, $Q \in L^2_{\rm loc}(\Omega)$ and $\Dom(Q):= \{\psi \in L^2(\Omega) \, : \, Q \psi \in L^2(\Omega) \}$. It is not difficult to verify that this $\X$ is a Banach space.

\begin{example}[Schr\"odinger operators with complex potentials, $n=1$]\label{ex.A.Schr}
Let $\H = L^2(\R^d)$, $V_1 \in W^{1,\infty}_{\rm loc}(\R^d)$ and  $V_2 \in L^2_{\rm loc}(\R^d)$. 
Let $V_1$ and $V_2$ satisfy further
\begin{enumerate}[(i)]
\item $\Re V_1 \geq 0$ and $|\nabla V_1| \leq C_1 |V_1| + C_2$ with some $C_1, C_2 >0$,
\item there exist $\alpha \in [0,1),$ $\beta \geq 0$ such that, for all $\psi \in \H^2_{V_1}(\R^d)$, 
\end{enumerate}
\begin{equation}\label{W.rb}
\|V_2 \psi \| \leq \alpha (\|\Delta \psi\| + \|V_1 \psi\|  ) + \beta \|\psi\|.
\end{equation}

Then the operator 
\begin{equation}
A := -\Delta + V_1 + V_2,
\qquad 
\Dom(A) := \H^2_{V_1}(\R^d)
\end{equation}
and the space $\X = \H^2_{V_1}(\R^d)$ satisfy Assumption \ref{ass:Af}.\eqref{ass:Af.norm} with $n=1$.
The main step needed to justify \eqref{ass.f.1.gr} is the fact that, for all $\psi \in \Dom(A)$,
\begin{equation}\label{ex.A.Schr.gr}
c_1 \| \psi \|_{\H^2_{V_1}}^2 \leq  \| A \psi \|^2 + \|\psi\|^2 \leq c_2 \| \psi \|_{\H^2_{V_1}}^2,
\end{equation}
where $c_1$, $c_2>0$ are independent of $\psi$, \cf~for instance \cite{Boegli-2014} for details; the standard proof for $V_1(x)=x^2$ and $d=1$ can be found \eg~in \cite[Ex.7.2.4]{BEH}. 
\end{example}

\begin{example}[Schr\"odinger operators with singular potentials, $n=1$]\label{ex.A.sing}
Let $\H = L^2((-r,r))$ with $r \in (0, + \infty]$, $V_1 \in L^1_{\rm loc}((-r,r))$ and $v_2$ be a sesquilinear form. 
%
Let $V_1$ and $v_2$ further satisfy
\begin{enumerate}[(i)]
\item $\Re V_1 > 0 $ and $|\Im V_1| < \tan \theta \Re V_1$ with $\theta \in [0, \pi/2)$,
\item $\H^1_{\sqrt{\Re V_1}}((-r,r)) \subset \Dom(v_2)$ and there exist $\alpha \in [0,1),$ $\beta \geq 0$ such that, for all $\psi \in \H^1_{\sqrt{\Re V_1}}((-r,r))$,
\begin{equation}\label{w.rb}
|v_2[\psi]| \leq \alpha (\| \psi' \|^2 + \|\sqrt{\Re V_1}\psi\|^2  ) + \beta \|\psi\|^2.
\end{equation}
\end{enumerate}	

Then the m-sectorial operator $A$ associated (via the first representation theorem \cite[Thm.VI.2.1]{Kato-1966}) with the closed sectorial form  
\begin{equation}\label{a.form.H1}
a[\psi] := \| \psi'\|^2 +\int_{-r}^r V_1(x)|\psi(x)|^2 \, \dd x + v_2[\psi], 
\qquad 
\Dom(a) := \H^1_{\sqrt{\Re V_1}}((-r,r))
\end{equation}
and the space $\X = \H^1_{\sqrt{\Re V_1}}((-r,r))$
satisfy Assumption \ref{ass:Af}.\eqref{ass:Af.norm} with $n=1$. 
To show \eqref{ass.f.1.gr}, recall that, for all $\psi \in \Dom(A) \subset \Dom(a)$, 
 $ |a[\psi]| = |\langle A \psi, \psi \rangle | \leq (\|A\psi\|^2 + \|\psi\|^2)/2$ and the norms $\| \cdot \|_{\H^1_{\sqrt{\Re V_1}}}$ and $\sqrt{|a[\,\cdot\,]| + c \|\cdot\|^2}$ with a sufficiently large $c \geq 0 $ are equivalent, \cf~\cite[Chap.VI.]{Kato-1966}.
 
The domain of $a$ and the space $\X$ can be selected also, for instance, as 
\begin{equation}
\H^{1,\rm D}_{\sqrt{\Re V_1}}:=H^1_0((-r,r)) \cap \Dom(\sqrt{\Re V_1})
\end{equation}
\ie~Dirichlet boundary conditions are imposed at $\pm r$, and the analogues of all claims above remain true.
\end{example}

\begin{example}[Schr\"odinger operators with bounded and regular potentials, $n>1$]\label{ex.A.Schr.n}

We present an example of Schr\"odinger operators satisfying Assumption \ref{ass:Af}.\eqref{ass:Af.norm} for $n>1$. The motivation for $n>1$ comes from the condition $s>d/2$ for $H^s(\R^d)$ needed for polynomial nonlinearities in Example \ref{ex.f.pol} below, guaranteeing that the polynomial nonlinearities satisfy Assumption \ref{ass:Af}.\eqref{ass:Af.Lip}. 

Let $\H = L^2(\R^d)$ and let $V \in W^{2(m-1),\infty}(\R^d)$ for some $m\in \N$ be a possibly complex potential. Then the operator
\begin{equation} \label{E:A.Schr.n}
A := - \Delta + V, \quad \Dom(A) := H^2(\R^d),
\end{equation}
and the space and $\X = \H_0^{2n}(\R^d)= H^{2n}(\R^d)$ satisfy Assumption \ref{ass:Af}.\eqref{ass:Af.norm} with any $n \leq m$. 
The claim holds since $\Dom(A^n) = H^{2n}(\R^d)$ and, for all $\psi \in H^{2n}(\R^d)$,
\begin{equation}
c_1 \| \psi \|_{H^{2n}}^2 \leq  \| A^n \psi \|^2 + \|\psi\|^2 \leq c_2 \| \psi \|_{H^{2n}}^2,
\end{equation}
where $c_1$, $c_2>0$.  The second inequality above follows from  $V \in W^{2(m-1),\infty}(\R^d)$ since $A^n\psi$ consists of terms $\pm \Delta^{i_1}V^{j_1}\Delta^{i_2}V^{j_2}\dots \Delta^{i_n}V^{j_n}\psi$ with 
$i_k,j_k\in \{0,1\}$ for all $k=1,\dots,n$ and such that the highest derivative acting on $V$ is of order $2(n-1)$. Each term can be thus estimated by $c\|\psi\|_{H^{2n}}$. Remaining estimates follow from the equivalence of norms on $H^{2n}(\R^d)$, see \eg~\cite[Lem.3.7.2]{Davies-1995}.
	
\end{example}

\begin{example}[Schr\"odinger operators with quasi-periodic boundary conditions]\label{ex.A.Schr.qp}
All examples \ref{ex.A.Schr}, \ref{ex.A.sing}, and \ref{ex.A.Schr.n} can be combined with quasi-periodic boundary conditions on a bounded domain $\Omega \subset \R^d$. Taking, for simplicity $\Omega = (-r,r)^d, r\in (0,\infty)$, the quasi-periodicity vectors are $k\in (-\pi,\pi]^d$ and we can choose  
\begin{equation}
\begin{aligned}
\X= \H_{V_1}^{2,k}(\Omega) := & \left\{\psi \in \H^2_{V_1}(\Omega): \psi(x+2re_j) = e^{\ii k_j} \psi(x), \nabla \psi(x+2re_j) = e^{\ii k_j} \nabla \psi(x)\right.\\
& \quad \left.  \text{ for all } j=1,\dots,d \text{ and all } x\in \partial \Omega \text{ s.t. }x+2re_j\in \partial \Omega\right\}\\
\end{aligned}
\end{equation}
in Example \ref{ex.A.Schr},
\begin{equation}
\X=\H^{1,k}_{\sqrt{\Re V_1}}((-r,r)):=\{\psi \in \H^1_{\sqrt{\Re V_1}}((-r,r)): \psi(r) = e^{\ii k} \psi(-r) \} 
\end{equation}
with $k \in (-\pi,\pi]$ in Example \ref{ex.A.sing} $(d=1)$, and
\begin{equation}
\begin{aligned}
\X=&\H_0^{2n,k}(\Omega):=  \left\{\psi\in H^{2n}(\Omega): D^m\psi(x+2re_j) = e^{\ii k_j} D^m\psi(x) \text{ for all } j=1,\dots,d, \right.\\
& \quad \left. \text{all }m \in \N^d \text{ s.t. } 0\leq |m|\leq 2n-1, \text{ and for all } x\in \partial \Omega \text{ s.t. }x+2re_j\in \partial \Omega\right\}\\
\end{aligned}
\end{equation}
in Example \ref{ex.A.Schr.n}. Here $e_j$ is the $j$-th Euclidean vector and $D^m=\partial_{x_1}^{m_1}\dots \partial_{x_d}^{m_d}$.

Assumption \ref{ass:Af}.\eqref{ass:Af.norm} is still satisfied with these choices.
\end{example}

\begin{example}[Discrete Schr\"odinger operator]\label{ex.A.DSchr}

The difference operator on $\H=\Dom(A)=\C^{2N}$ ($N\in \N$) defined by
\begin{equation}\label{A:E.DS}
(A\psi)_n := \begin{cases}\psi_{n+1}+\psi_{n-1} + \ii \gamma(-1)^n\psi_n \quad & \text{for } 2\leq n \leq 2N-1,\\
\psi_{n+1}- \ii \gamma \psi_n & \text{for } n=1,\\
-\psi_{n-1} + \ii \gamma \psi_n & \text{for } n=2N,
\end{cases}
\end{equation}
appears in the modeling of one dimensional optical lattices \cite{Kevrekidis-2013-12}. Choosing $\X = \H$, assumption \ref{ass:Af}.\eqref{ass:Af.norm} is obviously satisfied in this finite dimensional case (\eg~with the Euclidean norm).

\end{example}

\begin{example}[First order Dirac type operator]\label{ex.A.CME}

An example of a physically interesting operator other than a Schr\"odinger one is 
\begin{equation}\label{A:E.Dirac}
A:= \begin{pmatrix}-\ii \partial_x -V(x) & -\kappa(x)\\
-\kappa(x) & \ii \partial_x -V(x)
\end{pmatrix}
\end{equation}
with $V,\kappa \in L^\infty(\R)$. This operator occurs in a model for optical waves in fibers with a Bragg grating and a localized defect \cite{Goodman-2002-19}.
Assumption \ref{ass:Af}.\eqref{ass:Af.norm}  is satisfied under the natural choice $\H=L^2(\R)\times L^2(\R), \X=\Dom(A)=H^1(\R)\times H^1(\R)$, where $\|\psi\|_\X := \|\psi_1\|_{H^1}+\|\psi_2\|_{H^1}$. 
\end{example}

\subsection{The spectral condition (Ass. \ref{ass:Af}.(\ref{ass:Af.mu}))}
\label{subsec:spec.cond}

To satisfy Assumption \ref{ass:Af}.\eqref{ass:Af.mu}, a detailed spectral analysis of a given linear operator $A$ must be performed. Here we recall some perturbation results on isolated eigenvalues that can be often used to justify the presence of simple isolated eigenvalues for more complicated, typically non-self-adjoint, differential operators with complex coefficients. In Section \ref{sec:sym}, Remark \ref{rem:sim+real}, we further explain the stability of realness of simple eigenvalues if $A$ possesses a certain symmetry. Finally, we recall here basic spectral properties of Schr\"odinger operators, particularly the results on the essential spectrum and the simplicity of the ground state. 

\subsubsection{Holomorphic families of operators}
\label{subsubsec:hol.fam}
Standard results on the spectrum of a holomorphic family of operators $A(\gamma)$, $\gamma \in \C$, \cf~\cite[Chap.VII]{Kato-1966}, yield that if the spectrum of $A(0)$ is separated into two parts, then this remains true also for $A(\gamma)$ with $|\gamma|$ sufficiently small, \cf~\cite[Thm.VII.1.7]{Kato-1966}. Moreover, isolated eigenvalues depend analytically on $\gamma$ and their multiplicities are preserved, \cf~\cite[Thm.VII.1.8]{Kato-1966}. 
Criteria for the holomorphicity of an operator family can be found in \cite[Chap.VII]{Kato-1966}. A sufficient condition for the operators or quadratic forms (and hence for the operators associated with $a(\gamma)$ by the first representation theorem \cite[Thm.VI.2.1]{Kato-1966}) of the type 
\begin{equation}
A(\gamma) = A_0 + \gamma B, 
\quad
a(\gamma) = a_0 + \gamma b, \quad \gamma \in \C,
\end{equation}
where $A_0$ is a densely defined closable operator and $a_0$ is a densely defined closable sectorial form, is the relative boundedness of $B$, $b$ with respect to $A_0$, $\Re a_0$, respectively, \ie~
\begin{equation}
\begin{aligned}
A(\gamma): \ &\Dom(A_0) \subset \Dom(B),  &\|B \psi \| &\leq \alpha \|A_0\psi\| + \beta \|\psi\|, \ \psi \in \Dom(A_0),
\\
a(\gamma): \ &\Dom(a_0) \subset \Dom(b),  &|b[\psi]|^2  &\leq \alpha \Re a_0[\psi] + \beta \|\psi\|^2, \ \psi \in \Dom(a_0),
\end{aligned}
\end{equation}
with some $\alpha, \beta \geq 0$, \cf~\cite[Thm.VII.2.6, Thm.VII.4.8]{Kato-1966}.

\subsubsection{Spectra of Schr\"odinger operators}
\label{subsubsec:Schr.sp}

We consider Schr\"odinger operators in the setting of Examples \ref{ex.A.Schr}--\ref{ex.A.Schr.n}; many of the following spectral properties are valid in a much greater generality, \cf~\cite{EE} for instance.

Let $A$ be the operator from Example \ref{ex.A.Schr} and let $\lim_{|x| \to \infty}|V_1(x)| = \infty$.  Then the resolvent of $A$ is compact, hence $\sigma(A) = \spd(A)$. Moreover, the selfadjoint operator $A_0:=-\Delta + Q$, with $Q \in L^2_{\rm loc}(\R^d)$, $Q \geq 0$ and $\lim_{|x| \to \infty} Q(x) =+\infty$, has the simple ground state, \ie~the lowest eigenvalue, \cf~\cite[Thm.XIII.47]{Reed4}. 
For $d=1$, a Wronskian argument can be used to conclude the simplicity of all eigenvalues for single-well potentials like $A_0=-\partial_x^2 + |x|^\beta$, $\beta>0$. 
For the singular Schr\"odinger operators from Example \ref{ex.A.sing}, the resolvent is compact if $r<\infty$ (also if Dirichlet or quasi-periodic boundary conditions are considered) or if $r=\infty$ and  $\lim_{|x| \to \infty}|V_1(x)| = \infty$. The simplicity of the ground state in the selfadjoint case can be in some situations concluded from \cite[Thm.XIII.48]{Reed4}.

Let $A$ be the Schr\"odinger operator from Example \ref{ex.A.Schr} with $V=V_1+V_2 \in L^{\infty}(\R^d)$ and $\lim_{|x| \to  \infty}(V_1 + V_2)(x) = 0$, then $\spe(A)=[0,+\infty)$, \cf~\cite[Cor.X.4.2, Ex.X.4.3]{EE} and \eg~\cite[Ex.XIII.4.6]{Reed4}. (Note that there are several different definitions of essential spectrum for non-selfadjoint operators, \cf~\cite[Chap.IX]{EE}, nevertheless, all coincide for this special situation.) Discrete eigenvalues may appear outside essential spectrum. Particularly in $d=1,2$, the selfadjoint operator $A_0:=-\Delta + \varepsilon \, Q$ with $\int Q \, \dd x<0$ and $Q$ decaying sufficiently fast, \cf~\cite{Simon-1976-97,Klaus-1977-108} for precise assumptions on $Q$, possesses a unique negative simple eigenvalue for all sufficiently small $\varepsilon>0$;  some non-self-adjoint extensions can be found in \cite{Novak-2013-DP,Novak-2014}. 
For the singular Schr\"odinger operators from Example \ref{ex.A.sing} with $r=\infty$,  $\spe(A)=[0,+\infty)$ if $\lim_{|x| \to  \infty} V_1(x)=0$ and $v_2=0$ or $v_2$ are forms corresponding to $\delta$ potentials discussed in Section \ref{subsec:BEC}.

Let $A$ be the Schr\"odinger operator from Example \ref{ex.A.Schr} with $V = V_1+V_2 \in L^{\infty}((-r,r)^d)$ and the $k$-quasi-periodic boundary conditions, $k\in (-\pi,\pi]^d$, so, as in Example \ref{ex.A.Schr.qp}, $\Dom(A)=\H_{0}^{2,k}((-r,r)^d)$. These operators  arise naturally from a periodic problem in $L^2(\R^d)$, where $V$ is extended periodically onto $\R^d$ \cf~\cite[Chap.XIII.16]{Reed4} and the concepts of the ``band structure'' and of the Bloch eigenvalue problem. As mentioned above, the resolvent of $A$ is compact. Moreover, if $d=1$ and $V$ is real, all eigenvalues of $A$ are simple if $k \notin \{0,\pi\}$, \cf~\cite[Thm.XIII.89]{Reed4}.

\subsubsection{Spectrum of the discrete Schr\"odinger operator \upshape{(\ref{A:E.DS})}}
\label{subsubsec:DSchr.sp}

It is a straightforward calculation, see \cite{Kevrekidis-2013-12}, to show that $A$ in \eqref{A:E.DS} has the $2N$ eigenvalues
$$\pm \left(4\cos^2\left(\tfrac{\pi j}{1+2N}\right)-\gamma^2\right)^{1/2}, \quad 1\leq j \leq N,$$
which are simple and real for $\gamma\in \left(-2\cos\left(\tfrac{\pi N}{1+2N}\right),2\cos\left(\tfrac{\pi N}{1+2N}\right)\right)$.

\subsubsection{Spectrum of the Dirac type operator \upshape{(\ref{A:E.Dirac})}}
\label{subsubsec:Dirac.sp}

If $\kappa(x)\to \kappa_\infty>0$ and $V(x)\to 0$ as $|x|\to \infty$, then the essential spectrum of $A$ in \eqref{A:E.Dirac} is $(-\infty,-\kappa_\infty]\cup [\kappa_\infty,\infty)$. To see this, notice the perturbation result \cite[Prop.6.6]{Cuenin-2014-15} and the unitary equivalence of $A$ in \eqref{A:E.Dirac} and $H$ in \cite{Cuenin-2014-15} for $|x|\to \infty$. 
For real $\kappa$ and $V$ the operator is self-adjoint.
Bounds on eigenvalues of non-self-adjoint perturbations outside of the essential spectrum are proved in \cite{Cuenin-2014-15}.  Special choices of real $\kappa$ and $V$ with simple eigenvalues in $(-\kappa_\infty,\kappa_\infty)$ have been found in \cite[Sec.4.B.]{Goodman-2002-19} using a connection to the spectral problem of the inverse scattering theory for the nonlinear Schr\"odinger equation.
For example, for $\kappa(x)=(\mu_0^2+k^2\tanh^2(kx))^{1/2}$ and $V(x)=\tfrac{k^2\mu_0}{2}\kappa^{-2}(x)\text{sech}^2(k x)$ with $\mu_0\in (-\kappa_\infty,\kappa_\infty)$ and $k\in \R\setminus \{0\}$, the operator $A$ has a simple eigenvalue at $\mu_0$. 

\subsection{Nonlinearities (Ass.~\ref{ass:Af}.(\ref{ass:Af.Lip}))}

We present several nonlinearities $f$ satisfying Assumption \ref{ass:Af}.\eqref{ass:Af.Lip}. The considered spaces $\X$ are those arising for Schr\"odinger operators in Section \ref{subsec:Schr.op}. In fact, we show that the Lipschitz continuity \eqref{E:Lip} holds locally for all $\eta_0 \in \X$, thus also for the eigenvector $\psi_0$.

\begin{example}[Polynomial nonlinearity]
\label{ex.f.pol}
Let $\H=L^2(\Omega)$ and $\X$ be as in \eqref{H.X.Schr} with $s>d/2, s \in \N$, and $\Omega=\R^d$.
Let $N\in \N$ and
\begin{equation}\label{ex.f.pol.def}
f_{\rm pol}(\psi) := \sum_{p,q=0}^N a_{pq} \psi^p \ov{\psi}{}^q \ \text{with} \  a_{pq}\in C_b^s(\R^d,\C) \ \text{for all} \ p,q,
\end{equation}
where $C_b^s$ is the space of functions with continuous and bounded derivatives up to order $s$. Without any loss of generality, we set $a_{00}=a_{01}=a_{10}=0$. A classical example is the cubic nonlinearity $f_{\rm c}(\psi):=|\psi|^2\psi$, \ie~$a_{21}=1,a_{pq}=0$ otherwise.
We show below that $f_{\rm pol}$ in \eqref{ex.f.pol.def} satisfies Assumption \ref{ass:Af}.\eqref{ass:Af.Lip} with any $\eta_0 \in \X=\HsQ(\R^d)$ and $r_L >0$.

First note that, for $s>d/2$, the norm $\|\cdot\|_{H^s}$ satisfies the so-called algebra property:
there exists $C_{\rm a}>0$ such that, for all $\phi, \psi \in H^{s}(\R^d)$, 
\begin{equation}\label{Hs.alg}
\|\phi \psi \|_{H^s} \leq C_{\rm a} \|\phi \|_{H^s} \| \psi \|_{H^s},
\end{equation}
\cf~\cite[Thm.4.39]{Adams-2003} or \cite[Lem.4.2]{Dohnal-2009-238}. Moreover, the Sobolev embedding of $H^s(\Omega)$ in $L^{\infty}(\Omega)$ holds, \cf~\cite[Thm.4.12]{Adams-2003}, \ie~there exists $C_{\rm e}>0$ such that, for all $\phi \in H^{s}(\R^d)$,
\begin{equation}\label{Sob.emb.inf}
\|\phi \|_{L^\infty} \leq C_{\rm e} \|\phi \|_{H^s}. 
\end{equation}
Thus, for all $\phi, \psi \in \HsQ(\R^d)$,
\begin{equation}\label{n1.alg}
\begin{aligned}
\|\phi \psi\|_\HsQ & = \|\phi \psi\|_{H^s} + \|Q \phi \psi\|
\\
& \leq C_{\rm a} \|\phi\|_{H^s}\|\psi\|_{H^s} + \sqrt{\|\phi\|_{L^{\infty}} \|\psi\|_{L^{\infty}} \|Q \phi \| \, \|Q \psi\|} 
\\ & \leq \max\left\{C_{\rm a},\frac{C_{\rm e}}{\sqrt{2}}\right\}\|\phi\|_\HsQ \|\psi\|_\HsQ,
\end{aligned}
\end{equation}
hence the norm $\|\cdot\|_{\HsQ}$ satisfies the algebra property as well.

Clearly, it suffices to check \eqref{E:Lip} for a single term $\psi^p \ov{\psi}{}^q$. From
\begin{equation}\label{ab.dec}
\begin{aligned}
\phi^p\ov{\phi}{}^q-\psi^p\ov{\psi}{}^q & = \frac 12 (\phi^p-\psi^p)(\ov{\phi}{}^q+\ov{\psi}{}^q)+ \frac 12 (\phi^p+\psi^p)(\ov{\phi}{}^q-\ov{\psi}{}^q), \\
\phi^m - \psi^m & = (\phi-\psi) \sum_{k=0}^{m-1} \phi^k \psi^{m-1-k}
\end{aligned}
\end{equation}
and using \eqref{n1.alg}, we obtain
\begin{equation}
\begin{aligned}
\|\phi^p\ov{\phi}{}^q-\psi^p\ov{\psi}{}^q\|_\HsQ 
& \leq C_1
\left(
\|\phi^p - \psi^p\|_\HsQ \|\ov \phi{}^q + \ov \psi{}^q\|_\HsQ +  \|\phi^p + \psi^p\|_\HsQ \|\ov \phi{}^q - \ov \psi{}^q\|_\HsQ
\right)
\\ & \leq C_2 \left(
(\|\phi\|_\HsQ^q + \|\psi\|_\HsQ^q) \sum_{k=0}^{p-1} \|\phi\|_\HsQ^k \|\psi\|_\HsQ^{p-1-k}
\right.
\\ & \qquad \quad +
\left.
(\|\phi\|_\HsQ^p + \|\psi\|_\HsQ^p) \sum_{k=0}^{q-1} \|\phi\|_\HsQ^k \|\psi\|_\HsQ^{q-1-k}
\right) 
\|\phi - \psi\|_\HsQ,
\end{aligned}
\end{equation}
where $C_1,C_2 >0$ are independent of $\phi$ and $\psi$. Finally, for all $\xi
\in \{ \eta \in \HsQ(\R^d) \, : \, \|\eta - \eta_0 \|_\HsQ < r_L\}$ with any $\eta_0 \in \HsQ(\R^d)$ and $r_L >0$, we have $\|\xi\|_\HsQ \leq \|\eta_0\|_\HsQ + r_L$, thus
\begin{equation}\label{ex.f.pol.n1}
\begin{aligned}
\|\phi^p\ov{\phi}{}^q-\psi^p\ov{\psi}{}^q\|_\HsQ 
& \leq 
C_{\eta_0,r_L,p,q} \|\phi - \psi\|_\HsQ
\end{aligned}
\end{equation}
if $\phi,\psi\in \{ \eta \in \HsQ(\R^d) \, : \, \|\eta - \eta_0 \|_\HsQ < r_L\}$. Also note that for $a\in C^s_b(\R^d)$ and $\psi\in \HsQ(\R^d)$
$$\|a\psi\|_{\HsQ} \leq \|a\|_{C^s}\|\psi\|_{H^s}+\|a\|_{C^0}\|Q\psi\|\leq \|a\|_{C^s} \|\psi\|_{\HsQ}.$$
As a result
\begin{equation}\label{pol.nl.Lip}
\|f_{\rm pol}(\phi)-f_{\rm pol}(\psi)\|_\HsQ \leq 
\left(\sum_{p,q=0}^N \|a_{pq}\|_{C^s(\R^d)}C_{\eta_0,r_L,p,q} \right)\|\phi - \psi\|_\HsQ, 
\end{equation}
which implies validity of \eqref{E:Lip} in Assumption \ref{ass:Af}.\eqref{ass:Af.Lip}. 

Note that this Lipschitz continuity can be directly extended to vector valued polynomial nonlinearities $f(\psi)=(f_1(\psi), \dots, f_m(\psi))^T$ with $m\in \N$, $\psi\in (\H_Q^s(\R^d))^m$ and with $f_j(\psi)=\sum_{p_1,q_1,\dots,p_m,q_m=0}^Na_{p_1q_1\dots p_mq_m}^{(j)}\psi_1^{p_1}\overline{\psi_1}^{q_1}\dots\psi_m^{p_m}\overline{\psi_m}^{q_m}$. In the vector case the $\H^s_Q$-norm in \eqref{pol.nl.Lip} must be replaced, e.g., by $\|\psi\|_{(\H_Q^s)^m}:=\sum_{j=1}^m\|\psi_j\|_{\H_Q^s}$.
\end{example}

\begin{remark}[Motivation for $n>1$ in  Ass.~\ref{ass:Af}.\eqref{ass:Af.norm}]\label{R:n}
Examples \ref{ex.A.Schr.n} and \ref{ex.f.pol} constitute the primary motivation for $n>1$ in the choice of a space $\X$ with $\Dom(A^n) \subset \X \subset \Dom(A^{n-1})$. The condition $\N\ni s>d/2$ in Example \ref{ex.f.pol} implies that for $d \geq 4$ we need $s\geq 3$. The natural $H^2$-space of Example \ref{ex.A.Schr} is, therefore, not suitable as our working space. On the other hand, the space $H^{2n}$ of Example \ref{ex.A.Schr.n} with $n > d/4$ is sufficiently small. The reason for choosing $\Dom(A^n) \subset \X \subset \Dom(A^{n-1})$ is the use of a fixed point argument in Theorem \ref{thm.NL-fixpt}, in particular to guarantee  $G(B_{r_2 \eps^2}) \subset Q_0 \X$, \cf~the reasoning between \eqref{Brg2.def} and \eqref{Gchi.est}.

In summary, if we choose
\begin{equation}\label{E:dsnm}
\frac{d}{4}<\frac{s}{2}\leq n \leq m \ \text{with} \ s,n,m\in \N,
\end{equation}
then $A$ in \eqref{E:A.Schr.n} with $V \in W^{2(m-1),\infty}(\R^d)$, $f=f_{\rm pol}$, and $\X=H^{2n}(\R^d)$ satisfy Assumption  \ref{ass:Af}.\eqref{ass:Af.norm} and \eqref{ass:Af.Lip}. Of course, also $\X=H^{2n,k}(\Omega)$ of Example \ref{ex.A.Schr.qp} with the same $A$ and $f$ is admissible. For $d=1,2$ inequality \eqref{E:dsnm} holds also with $s=2, n=1$, such that $\X=\H_{V_1}^2(\R^d)$ of Example \ref{ex.A.Schr} or the space $\H_{V_1}^{2,k}(\Omega)$ of  Example \ref{ex.A.Schr.qp} can be used with $f=f_{\rm pol}$. For $d=1$ we can use even $s=1,n=1$ and, hence, the spaces $\H^{1}_{\sqrt{\Re V_1}}((-r,r)),\H^{1,D}_{\sqrt{\Re V_1}}((-r,r))$ of Example \ref{ex.A.sing} and $\H^{1,k}_{\sqrt{\Re V_1}}((-r,r))$ of Example \ref{ex.A.Schr.qp} are admissible.
\end{remark}
%

\begin{example}[The monopolar and dipolar interaction from \cite{Lahaye-2009-72,Dast-2013-46}] 
	\label{ex.f.mon.dip}
Let $\H=L^2(\Omega)$ and $\X$ be as in \eqref{H.X.Schr} with $s=2$ and $\Omega=\R^3$.
We investigate nonlinearities having formally the form
\begin{equation}\label{f.form}
(f_\iota(\psi))(x) := \psi(x) \int_{\R^3} K_\iota(y) |\psi(x-y)|^2 \dd y, \quad \iota \in \{{\rm m}, {\rm d} \},
\end{equation}  
with kernels $K_{\rm m}(x) = \frac{1}{|x|}$ and $K_{\rm d}(x) = \frac{1}{|x|^3} \left(1-3 \frac{(x.\alpha)^2}{|x|^2} \right)$ with fixed $\alpha \in \R^3$, $\|\alpha\|_{\R^3}=1$. 
Notice that choosing $K(x)= \delta^3(x)$, we can recover also the so called contact interaction, which, however, coincides with the cubic nonlinearity from Example \ref{ex.f.pol}. 
Here we focus on the nonlinearities $f_{\rm m}$ and $f_{\rm d}$, the so-called monopolar and dipolar interaction, respectively, see Section \ref{subsec:BEC} for more details.

First we show that, for any $\eta \in H^2(\R^3) \cap L^1(\R^3)$, 
\begin{equation}
\left \| \int_{\R^3} \frac{\eta(x-y) \; \dd y}{|y|} \right\|_{L^\infty} \leq C \|\eta\|_{H^2} + \|\eta\|_{L^1} 
\end{equation} 
with $C$ independent of $\eta$. Indeed, the Sobolev embedding \eqref{Sob.emb.inf}, applied in the last step, yields
\begin{equation}
\begin{aligned}
\int_{\R^3} \frac{|\eta(x-y)| \; \dd y}{|y|} 
& =
\int_{B_1(0)} \frac{|\eta(x-y)| \; \dd y}{|y|} +  \int_{\R^3 \setminus B_1(0)} \frac{|\eta(x-y)| \; \dd y}{|y|}
\\
& \leq 
\|\eta\|_{L^\infty} \int_{B_1(0)} \frac{\dd y}{|y|} + \|\eta\|_{L^1} 
\leq 
C \|\eta\|_{H^2} + \|\eta\|_{L^1}. 
\end{aligned}
\end{equation} 
Using this, the algebra property of $\|\cdot\|_{H^2}$, \cf~\eqref{Hs.alg}, the special case of the first formula in \eqref{ab.dec} and Cauchy-Schwartz inequality, we obtain
\begin{equation}\label{fm.est}
\begin{aligned}
\|f_{\rm m}(\phi) - f_{\rm m}(\psi) \| 
& \leq 
\| \phi - \psi \| \left \| \int_{\R^3} \frac{|\phi(x-y)|^2 \; \dd y}{|y|} \right\|_{L^\infty}
\\
& \quad + \|\psi\| \left \| \int_{\R^3} \frac{\big| |\phi(x-y)|^2-|\psi(x-y)|^2 \big| \; \dd y}{|y|} \right\|_{L^\infty}
\\
& \leq
C_1 \|\phi\|_{H^2}^2 \| \phi - \psi \|_{\H^2_Q} +  C_2 \|\psi\| \|\phi-\psi\|_{H^2} \|\phi+\psi\|_{H^2}  
\\
& \leq
C_3(\|\phi\|_{H^2}^2 + \|\psi\|_{H^2}^2)  \| \phi - \psi \|_{\H^2_Q},
\end{aligned}
\end{equation}
thus \eqref{E:Lip} is satisfied with $n=1$ for any $\eta_0 \in \H^2_Q(\R^3)$ and $r_L>0$. In particular,  Assumption \ref{ass:Af}.\eqref{ass:Af.Lip} holds.

The nonlinearity $f_{\rm d}$ is more complicated and it is even not immediately clear why it is  well-defined. Nonetheless, the appropriate framework of singular integrals, particularly  \cite[Thm.II.3]{Stein-1970}, yields that, for any $\eta \in L^2(\R^3)$,
\begin{equation}
\left\|
\int_{\R^3} K_{\rm d}(y)\eta(x-y) \, \dd y 
\right\| \leq C_{\rm d} \|\eta\|,
\end{equation}
with $C_{\rm d}$ independent of $\eta$. Using the Sobolev embedding \eqref{Sob.emb.inf}, the algebra property \eqref{Hs.alg} of $H^2$-norm and the special case of the first formula in \eqref{ab.dec}, we obtain
\begin{equation}
\|f_{\rm d}(\phi) - f_{\rm d}(\psi) \| 
\leq 
C_4(\|\phi\|_{H^2}^2 + \|\psi\|_{H^2}^2)  \| \phi - \psi \|_{\H^2_Q},
\end{equation}
thus \eqref{E:Lip} is satisfied with $n=1$ for any $\eta_0 \in \H^2_Q(\R^3)$, $r_L>0$, and Assumption \ref{ass:Af}.\eqref{ass:Af.Lip} holds in particular.
\end{example}

\begin{example}[Nonlocal nonlinearity from  \cite{Rubinstein-2010-195}] 
\label{ex.f.Rubinstein}

Let $\H=L^2((-1,1))$ and $\X= \H^{1, \rm D}_0((-1,1)) = H^1_0((-1,1))$.  We analyze the non-local nonlinearity 
\begin{equation}\label{E:heat_Rubinstein}
(f_{\rm N}(\psi))(x) := \psi \int_0^x \Im \left( \psi(s)\ov{\psi}_x(s)\right) \dd s.
\end{equation}

The validity of \eqref{E:Lip} is easily checked since $f_{\rm N}$ satisfies
\begin{equation}
\begin{aligned}
&\left\| 
f_{\rm N}(\phi) - f_{\rm N}(\psi)
\right\|
\\
& \ = 
\frac{1}{2} 
\left\|
(\phi-\psi)   \int_0^x \Im\left( \phi\overline{\phi}_x+\psi\overline{\psi}_x\right) \dd s  
+ (\phi+\psi)  \int_0^x \Im\left( \phi\overline{\phi}_x-\psi\overline{\psi}_x\right) \dd s 
\right\|
\\
& \ \leq
\frac{1}{2} 
\left(
\left\| \phi-\psi\right\|_{L^\infty}  
\left\| \int_0^x \Im\left( \phi\overline{\phi}_x+\psi\overline{\psi}_x\right) \dd s \right \| 
\right.
\\ & \quad \qquad    
\left. + \left\| \phi + \psi\right\|_{L^\infty}   
\left\|\int_0^x \Im\left( \phi\overline{\phi}_x-\psi\overline{\psi}_x\right) \dd s 
\right\|
\right).
\end{aligned}
\end{equation}
For any $\eta_0 \in H^1_0((-1,1))$, $r_L>0$ and $\phi, \psi \in \{ \eta \in H^1_0((-1,1)) \, : \, \|\eta - \eta_0 \|_{H^1} < r_L\}$, we have 
\begin{equation}
\begin{aligned}
\left\|
\int_0^x \Im\left( \phi\overline{\phi}_x+\psi\overline{\psi}_x\right) \dd s  
\right\| 
&\leq 
\left\|
\|\phi\| \|\phi_x\|+\|\psi\|\|\psi_x\|  
\right\| 
\leq 
\frac{1}{\sqrt{2}}\left(
\|\phi\|_{H^1}^2 +\|\psi\|_{H^1}^2
\right) \\
&\leq 
\sqrt{2}(r_L+\|\eta_0\|_{H^1})^2
\end{aligned}
\end{equation}
and
\begin{equation}
\begin{aligned}
\left\|
\int_0^x \Im\left( \phi\overline{\phi}_x-\psi\overline{\psi}_x\right) ds  
\right\| 
& \leq 
\frac 12 \left\|
\|\phi-\psi\| \|\phi_x+\psi_x\|+\|\phi+\psi\|\|\phi_x-\psi_x\|  
\right\| 
\\
&\leq 
\sqrt{2} \|\phi+\psi\|_{H^1} \|\phi-\psi\|_{H^1}
\\
&
\leq 
2\sqrt{2}(r_L+\|\eta_0\|_{H^1}) \|\phi-\psi\|_{H^1}.
\end{aligned}
\end{equation}
Thus, by the embedding of $H^1$ in $L^{\infty}((-1,1))$ as in Example \ref{ex.f.pol}, we obtain \eqref{E:Lip}. In summary, Assumption \ref{ass:Af}.\eqref{ass:Af.Lip} holds with $n=1$. 
\end{example}

\section{Nonlinear eigenvalue problem}\label{sec:NL_ev}

First, we prove the local existence and uniqueness of the solutions of the nonlinear eigenvalue problem \eqref{nl.ev} under Assumption \ref{ass:Af}. Next, we focus on homogeneous nonlinearities, for which \eqref{nl.ev} together with condition $\|\psi\|=1$ can be solved.  Finally, the influence of possible non-selfadjointness (more precisely non-normality) of $A$ on constants appearing in our estimates is discussed.

\subsection{Local existence and uniqueness}
\label{subsec:ex.un}

\begin{theorem} \label{thm.NL-fixpt}
Let $A$ and $f$ satisfy Assumption \ref{ass:Af}. 
Then every solution of the nonlinear eigenvalue problem 
\begin{equation}\label{A.f.eq}
\begin{aligned}
(A - \mu) \psi - \varepsilon f(\psi) =0, \quad \langle\psi,\psi_0^*\rangle=1
\end{aligned}
\end{equation}
can be written as
\begin{equation}\label{muf.ans}
\mu =\mu_0+\varepsilon \nu + \varepsilon^2\sigma, \qquad \psi = \psi_0+ \varepsilon \phi +\chi,
\end{equation}
with $\nu,\sigma\in \C$ and $\phi,\chi\in Q_0\text{Dom}(A)$ and where 
\begin{equation}\label{nu.def}
\nu = - \langle f(\psi_0), \psi_0^* \rangle,
\end{equation}
$\phi$ is the unique (in $Q_0\text{Dom}(A)$) solution of
\begin{equation}\label{E:phi_thm}
(A - \mu_0) \phi = \nu \psi_0 +f(\psi_0), 
\end{equation}
and $(\sigma,\chi)$ solves the nonlinear system
%
\begin{align}
0 & = \varepsilon \sigma +  \langle f(\psi) - f(\psi_0), \psi_0^* \rangle, \label{sigma_eq}
\\
Q_0(A-\mu_0) Q_0 \chi & = \varepsilon \left[ (\nu + \varepsilon \sigma)(\chi + \varepsilon \phi) +  Q_0 (f(\psi) - f(\psi_0)) \right] =: R(\chi). \qquad
\label{chi_eq}
\end{align}
%

Moreover, there exists $\varepsilon_0 >0$ such that, for all $\varepsilon \in (-\varepsilon_0,\varepsilon_0)$, the nonlinear eigenvalue problem \eqref{sigma_eq}--\eqref{chi_eq} has a unique small solution, namely with 
$$	|\sigma| \leq r_1, \qquad \|\chi\|_\X \leq r_2 \varepsilon^2,$$
where $r_1,r_2 = O(1)$ ($\varepsilon\to 0$) satisfy \eqref{r.sel}, \eqref{s.sel} respectively. 
\end{theorem}

\begin{proof}
Without any loss of generality, for a given $(\mu_0,\psi_0)$ we can write the solution $(\mu,\psi)$ as in \eqref{muf.ans}. Although at this point the representation is not unique, we already know that $\varepsilon\phi+\chi = Q_0(\varepsilon\phi+\chi)$ due to the constraint $\langle\psi,\psi_0^*\rangle=1$ and the normalization $\langle\psi_0,\psi_0^*\rangle=1$.

First, we apply the projection $P_0$ to \eqref{A.f.eq} and obtain
$$\nu+\varepsilon\sigma = -\langle f(\psi),\psi_0^*\rangle = -\langle f(\psi_0),\psi_0^*\rangle  -\langle f(\psi)-f(\psi_0),\psi_0^*\rangle.$$
Defining $\nu := -\langle f(\psi_0),\psi_0^*\rangle$, we get equation \eqref{sigma_eq}.

Second, we apply $Q_0$ to \eqref{A.f.eq}, resulting in
\[
\begin{aligned}
Q_0(A-\mu_0)Q_0(\varepsilon \phi+\chi) =& (\varepsilon\nu+\varepsilon^2\sigma)(\varepsilon\phi+\chi)+\varepsilon Q_0f(\psi)\\
=& \varepsilon\left[(\nu+\varepsilon\sigma)(\varepsilon\phi+\chi) + Q_0(f(\psi)-f(\psi_0))+\nu\psi_0+f(\psi_0)\right]. 
\end{aligned}
\]
Next, we define $\phi$ to be the unique solution $\phi\in Q_0\text{Dom}(A)$ of \eqref{E:phi_thm}. This solution exists because $\nu \psi_0+f(\psi_0)\in \Ker(A^*-\overline{\mu_0})^\perp = \lspan \{\psi_0^*\}^\perp$ and because $\mu_0$ is not in the essential spectrum $\sigma_{{\rm e} 5}(A)$, it is not in $\sigma_{{\rm e} 3}(A)$ either, \cf~\cite[Chap.IX]{EE}, and hence $A-\mu_0$ is Fredholm. The equation above thus becomes problem \eqref{chi_eq}.

The rest of the proof deals with the existence of a unique solution $(\sigma,\chi)$, with $\chi$ small and $\sigma$ bounded, of \eqref{sigma_eq}--\eqref{chi_eq}.
Note that $Q_0(A-\mu_0) Q_0$ is boundedly invertible in $Q_0 \H$, hence equation \eqref{chi_eq} can be rewritten as
\begin{equation}\label{chi.eq}
\chi = (Q_0(A-\mu_0) Q_0 )^{-1} R(\chi) =: G(\chi).
\end{equation}

In the first step, for $\varepsilon$ in a small neighborhood of $0$, we use the fixed point argument to conclude the existence of a solution $\chi$ of \eqref{chi.eq} with $\|\chi\|_\X=O(\varepsilon^2)$.  
We search for a fixed point 
\begin{equation}\label{Brg2.def}
\chi \in B_{r_2 \varepsilon^2}:=\{\eta \in Q_0 \X \, : \, \|\eta\|_\X \leq r_2 \varepsilon^2\} 
\end{equation}
with some $r_2>0$ independent of $\varepsilon$. Its existence is guaranteed if we can show 
\begin{enumerate}[(i)]
\item \label{fp.i}
$\chi \in B_{r_2 \varepsilon^2} \ \Rightarrow \ G(\chi)\in B_{r_2 \varepsilon^2}$,
\item \label{fp.ii}
there exists $\rho \in (0,1)$ such that $\|G(\chi_1)-G(\chi_2)\|_\X \leq \rho \|\chi_1-\chi_2\|_\X$ for all $\chi_1,\chi_2 \in B_{r_2 \varepsilon^2}$.
\end{enumerate}

Note that $\psi_0$, being an eigenfunction, satisfies $\psi_0\in\Dom(A^m)$ for any $m\in \N$. Thus the right hand side of \eqref{E:phi_thm} lies in $\Dom(A^{n-1})$, hence $\phi \in Q_0\Dom(A^n)\subset Q_0\X$. Therefore if $\chi \in Q_0\X$, then $Q_0f(\psi)\in Q_0\Dom(A^{n-1})$ and $R(\chi) \in  Q_0 \Dom(A^{n-1})$. It is straightforward to check that then $G(\chi) \in Q_0 \Dom(A^n)$.
The properties of the norm $\|\cdot\|_\X$, \cf~Assumption \ref{ass:Af}.\eqref{ass:Af.norm}, yield
\begin{equation}\label{Gchi.est}
\begin{aligned}
\|G(\chi)\|_\X 
&\leq \frac{k_2}{k_1}
\| G(\chi) \|_n 
\leq 
\frac{k_2}{k_1} C_{\mu_0}
\| R(\chi)\|_{n-1}.
\end{aligned}
\end{equation}
To show the second inequality in \eqref{Gchi.est}, note that because $G(\chi)\in Q_0\Dom(A^n)$,
\[\begin{aligned}
AG(\chi)&=(A-\mu_0)Q_0G(\chi)+\mu_0G(\chi)\\
&=Q_0(A-\mu_0)Q_0G(\chi)+\mu_0G(\chi) =R(\chi)+\mu_0G(\chi).
\end{aligned}
\]
We thus have, for $k\geq 1$,
$$
A^k G(\chi) = A^{k-1}(R(\chi)+\mu_0 G(\chi))=\dots=\sum_{j=0}^{k-1}\mu_0^{k-1-j} A^j R(\chi) + \mu_0^k G(\chi).
$$
As a result, 
$\|A^k G(\chi) \| \leq \sum_{j=0}^{k-1}|\mu_0|^{k-1-j}\|A^jR(\chi)\| + |\mu_0|^k\|(Q_0(A-\mu_0)Q_0)^{-1}\| \|R(\chi)\|$ 
and the second inequality in \eqref{Gchi.est} follows, where $C_{\mu_0}>0$ is a constant depending on $|\mu_0|,n$ and $\|(Q_0(A-\mu_0)Q_0)^{-1}\|$.

To ensure \eqref{fp.i}, we take $\chi \in B_{r\varepsilon^2}$, estimate $\|R(\chi)\|_{n-1}$ and select suitable $r$ in the following. First note that $P_0$ and $Q_0$ are bounded on $\Dom (A^{m})$, $m\in \N_0$, moreover, since $A^k P_0 \psi = P_0 A^k \psi$ and $A^k Q_0 \psi = Q_0 A^k \psi$ for all $\psi \in \Dom(A^m)$ and $k=0,1,\dots,m$, we have 
\begin{equation}
\begin{aligned}
\|P_0\|_m &:= \sup_{0 \neq \psi \in \Dom (A^m)} \frac{\|P_0 \psi\|_m}{\|\psi\|_m} 
= 
\sup_{0 \neq \psi \in \Dom (A^m)} \frac{\sum_{k=0}^m  \|P_0 A^k \psi\|}{\sum_{k=0}^m \|A^k  \psi\|}  
\leq
\|P_0\|,
\end{aligned}
\end{equation} 
and similarly $\|Q_0\|_m \leq \|Q_0\|$. Next, 
\begin{equation}
\begin{aligned}
\|R(\chi)\|_{n-1}
& \leq
|\varepsilon| \big(  |\nu + \varepsilon \sigma| (\|\chi\|_{n-1} + |\varepsilon| \|\phi\|_{n-1}) + \|Q_0(f(\psi) - f(\psi_0))\|_{n-1} \big)
\\
& \leq 
|\varepsilon| k_1 \big (|\nu + \varepsilon \sigma| ( r_2 \varepsilon^2  + |\varepsilon| \|\phi\|_\X) +  \|Q_0\| L \|\varepsilon \phi + \chi\|_\X \big)
\\
& \leq 
\varepsilon^2 k_1
\big(
|\nu + \varepsilon \sigma| (r_2 |\varepsilon| + \|\phi\|_\X) + \|Q_0\| L( \|\phi\|_\X + r_2|\varepsilon|)
\big),
\end{aligned}
\end{equation}
thus we select $r_2$ such that
\begin{equation}\label{r.sel}
r_2 > k_2 C_{\mu_0} \big( |\nu| + \|Q_0\| L \big) \|\phi\|_\X. 
\end{equation}
For all sufficiently small $\varepsilon$, we satisfy firstly $\|\psi - \psi_0\|_\X \leq |\varepsilon| \|\phi\|_\X + r_2 \varepsilon^2 < r_L$, hence the Lipschitz property of $f$, \cf~Assumption \ref{ass:Af}.\eqref{ass:Af.Lip}, can be indeed used. Secondly we satisfy condition \eqref{fp.i}. 

It remains to prove \eqref{fp.ii}.  Similarly as above, we obtain (with $\psi_{1,2}:=\psi_0+\varepsilon\phi + \chi_{1,2}$)
\begin{equation}\label{R12.est}
\begin{aligned}
\|R(\chi_1) - R(\chi_2)\|_{n-1} & =
|\varepsilon| \| 
(\nu + \varepsilon \sigma) (\chi_1 - \chi_2) + Q_0 (f(\psi_1) - f(\psi_2))
\|_{n-1}
\\
& \leq 
|\varepsilon| k_1 \big( |\nu + \varepsilon \sigma| +  \|Q_0\| L \big)  \|\chi_1 - \chi_2\|_\X.
\end{aligned}
\end{equation}
Hence, for all sufficiently small $\varepsilon$, condition \eqref{fp.ii} is also satisfied.

In summary, there exists $\tilde \varepsilon_0 >0$, such that, for all $\varepsilon$, $|\varepsilon| < \tilde \varepsilon_0$, we have the function $\chi \in B_{r_2 \varepsilon^2} $ that solves \eqref{chi.eq}; note that then $\chi \in Q_0 \Dom(A)$ as well. Note also that $\chi$ and in particular $\tilde \varepsilon_0$ depend on $\sigma$. However, we consider only $|\sigma| \leq r_1$, where $r_1$ satisfies \eqref{s.sel}, and an inspection of the estimates above shows that we can find $\tilde \varepsilon_0$ independent of $\sigma$ (dependent only on $r_1$).

In order to solve the first equation in \eqref{sigma_eq}, we prove first that the solution $\chi$ is continuous in $\sigma$. More precisely,  
\begin{equation}
\begin{aligned}
& \|\chiso - \chist\|_\X  \leq 
\frac{k_2}{k_1} C_{\mu_0} \|R(\chiso) - R(\chist)\|_{n-1}
\\
& \quad = 
\frac{k_2}{k_1} C_{\mu_0} |\varepsilon| 
\| 
\nu (\chiso - \chist) + \varepsilon (\sigma_1 \chiso - \sigma_2 \chist) +\varepsilon^2 (\sigma_1 - \sigma_2) \phi 
\\
& \qquad + Q_0( f(\psi_0 + \varepsilon \phi + \chiso) - f(\psi_0 + \varepsilon \phi + \chist)) \|_{n-1}
\\
& \quad \leq 
k_2 C_{\mu_0} |\varepsilon| 
\Big(
|\nu| \|\chiso - \chist\|_\X + \varepsilon^2  |\sigma_1 - \sigma_2| \|\phi\|_\X
\\ 
& \qquad + \frac{|\varepsilon|}2 \| (\sigma_1 - \sigma_2) (\chiso + \chist) +  (\sigma_1 + \sigma_2) (\chiso - \chist)\|_\X
\\ 
& \qquad +
\|Q_0\| L \|\chiso - \chist\|_\X
\Big)
\\
& \quad \leq 
k_2 C_{\mu_0} |\varepsilon| 
\Bigg(
|\varepsilon| \Big(\frac{1}2 \|\chiso + \chist\|_\X + |\varepsilon| \|\phi\|_\X  \Big) | \sigma_1 - \sigma_2| 
\\ 
& \qquad +
\Big( |\nu| + 
\|Q_0\| L + \frac{|\varepsilon|}2 | \sigma_1 + \sigma_2| \Big) \|\chiso - \chist\|_\X
\Bigg),
\end{aligned}
\end{equation}
Hence, for $|\sigma_{1,2}| \leq r_1$,
\begin{equation}\label{chi.sigma.cont}
\begin{aligned}
& \|\chiso - \chist\|_\X  
\Big( 
1 - k_2 C_{\mu_0} |\varepsilon| 
\big(
|\nu| + |\varepsilon| r_1 + \|Q_0\| L  \big)
\Big)
\\& \quad  \leq 
k_2 C_{\mu_0} |\varepsilon|^3 \Big(r_2 |\varepsilon| + \|\phi\|_\X  \Big) | \sigma_1 - \sigma_2|.
\end{aligned}
\end{equation}

As the final step, we use the fixed point argument on 
\begin{equation}\label{sigma.eq_fp}
\sigma = - \frac 1\varepsilon \langle f(\psi) - f(\psi_0), \psi_0^* \rangle =: S(\sigma),
\end{equation}
where we search for a fixed point in $\{ \sigma \in \C \, : \, |\sigma| \leq r_1 \}$ with a suitable $r_1$ selected below. Since
\begin{equation}
|S(\sigma)|  
\leq k_1 L \|P_0\| (\|\phi\|_\X + r_2 |\varepsilon|),
\end{equation}
we choose $r_1$ such that
\begin{equation}\label{s.sel}
r_1 > k_1 L \|P_0\| \|\phi\|_\X, 
\end{equation}
hence, for sufficiently small $|\varepsilon|$, $|\sigma| \leq r_1$ implies $|S(\sigma)| \leq r_1$. Moreover, using the continuity of $\chi$ in $\sigma$, \cf~\eqref{chi.sigma.cont}, we obtain
\begin{equation}
\begin{aligned}
|S(\sigma_1) -  S(\sigma_2)| & \leq \frac{1}{|\varepsilon|} \|f(\psi(\sigma_1))-f(\psi(\sigma_2))\|\\
&\leq \frac{k_1L}{|\varepsilon|} \|\chiso - \chist\|_\X \leq C \varepsilon^2 | \sigma_1 - \sigma_2|, 
\end{aligned}
\end{equation}
hence, for sufficiently small $|\varepsilon|$, the fixed point argument yields the sought solution of \eqref{sigma.eq_fp}.
\end{proof}

\subsection{Homogeneous nonlinearity}
\label{subsec:hom.nl}

In the case of a homogeneous nonlinearity like \eg~$f(\psi)=|\psi|^{q-1}\psi$, solutions with norm one can be generated from the nonlinear eigenfunctions of Theorem \ref{thm.NL-fixpt} by a scaling.
\begin{corollary}[Nonlinear eigenfunction with norm one]
\label{cor.hom.f}
Let $A$ and $f$ satisfy Assumption \ref{ass:Af} and suppose that $f$ is a homogeneous nonlinearity, \ie~ for all $a>0$, it satisfies the scaling property $f(a\psi)=a^q f(\psi)$ with some $q\in \R$. Given the nonlinear eigenpair $(\mu(\varepsilon),\psi(\varepsilon))$ for $|\varepsilon|<\varepsilon_0$ from Theorem \ref{thm.NL-fixpt}, the pair
\begin{equation}
(\tilde{\mu}(\tilde{\varepsilon}),\tilde{\psi}(\tilde{\varepsilon})) =\left(\mu(\varepsilon),\|\psi(\varepsilon)\|^{-1}\psi(\varepsilon)\right) \text{ with } \tilde{\varepsilon}=\varepsilon\|\psi(\varepsilon)\|^{q-1}
\end{equation}
solves $(A-\tilde{\mu})\tilde{\psi}-\tilde{\varepsilon}f(\tilde{\psi})=0$ and satisfies $\|\tilde{\psi}(\tilde{\varepsilon})\|=1$. The mapping $\varepsilon\mapsto \varepsilon\|\psi(\varepsilon)\|^{q-1}$ is injective for $|\varepsilon|<\varepsilon_1$, where $\varepsilon_1\leq \varepsilon_0$ is small enough.
\end{corollary}
\begin{proof}
From the scaling property we immediately get 
$$(A-\mu(\varepsilon))\|\psi(\varepsilon)\|^{-1}\psi(\varepsilon)-\varepsilon\|\psi(\varepsilon)\|^{q-1}f(\|\psi(\varepsilon)\|^{-1}\psi(\varepsilon))=0.$$
The injectivity follows from the asymptotic equivalence $\varepsilon\|\psi(\varepsilon)\|^{q-1}=\varepsilon\|\psi_0+\varepsilon\phi+\chi\|^{q-1}\sim \varepsilon\|\psi_0\|^{q-1}=\varepsilon $ for $\varepsilon\to 0$.
\end{proof}

\begin{remark}
Similar results on the bifurcation from simple eigenvalues of Fredholm operators appear in the literature. As mentioned in the introduction, a classical reference is the paper \cite{Ize-1976-7} by Ize, where Thm.I.3.2 applies under our assumptions and the additional condition $\|f(\psi)\|=O(\|\psi\|^2)$ as $\psi\to 0$. The theorem of Ize treats the bifurcation from the zero solution $\psi\equiv 0$ in $(A-\mu)\psi=f(\psi)$ but for a homogeneous $f$ the problem can be rescaled to \eqref{nl.ev}. Our Theorem  \ref{thm.NL-fixpt} avoids the technical condition $\|f(\psi)\|=O(\|\psi\|^2)$ ($\psi\to 0$) and provides a more explicit expansion of $\mu$ and $\psi$.
\end{remark}

\subsection{Constants and non-selfadjointness}

Our fixed point argument in Theorem \ref{thm.NL-fixpt} works for $\eps$ sufficiently small and the size of remainders $\chi$ and $\sigma$ is determined by constants $r_2$ and $r_1$, respectively, \cf~ \eqref{r.sel}, \eqref{s.sel}. Notice that the size of $\eps$ is restricted at least by 
\begin{equation}
\varepsilon_0 k_2 C_{\mu_0} \big( |\nu + \varepsilon_0 \sigma| +  \|Q_0\| L \big) <1,
\end{equation}
arising from the contraction condition $\rho<1$ together with \eqref{Gchi.est} and \eqref{R12.est}. 

Provided $A$ is normal, \ie~$AA^* = A^*A$, the spectral projection $P_0$ and the complementary projection $Q_0$ are orthogonal, hence $\|P_0\|=\|Q_0\|=1$, and, in the case $n=1$, the constant $C_{\mu_0}$ is determined by spectral properties of $A$ since
\begin{equation}
C_{\mu_0}= 1 + (1+|\mu_0|) \|(Q_0 (A-\mu_0)Q_0)^{-1}\| =  1 + \frac{1+|\mu_0|}{\dist(\mu_0,\sigma(A)\setminus\{\mu_0\})},
\end{equation}
due to the standard relation $\|(A-z)^{-1}\| = \dist(z,\sigma(A))^{-1}$ valid for normal operators. 

However, in our applications, we usually encounter non-symmetric perturbations of self-adjoint operators that result both in non-selfadjointness and non-normality of the perturbed operators. Note that we lose both equalities for $\|Q_0\|$ and $C_{\mu_0}$ if $A$ is not normal, only $\geq$ is left in general. In particular, the spectral projections as well as the complementary projections may behave wildly as the size of the spectral parameter is increased even for simple looking one-dimensional Schr\"odinger operators with a complex potential and compact resolvent. For instance, considering the rotated oscillator $-\partial_x^2 + \ii x^2$, \cf~\cite{Davies-2007,Davies-2004-70}, for which all eigenvalues $\lambda_n$ are explicit, $\lambda_n = e^{\ii \pi/4} (2n+1)$, $n\in \N_0$, the norms of the corresponding spectral projections $P_n$ grow exponentially, more precisely, 
\begin{equation}
\lim_{n \to \infty} \frac{\log \|P_n\|}{n} = const;
\end{equation}
similar behavior is exhibited also by other well-studied (often $\PT$-symmetric and with real spectrum) Schr\"odinger operators, \cf~for instance  \cite{Henry-2012-350, Henry-2014-15,Mityagin-2013arx, Krejcirik-2014} and references therein. Notice that the growth of $\|P_n\|$ implies the growth of $\|Q_n\|$ since $\|Q_n\| \geq \|P_n\| -1$.  Concerning the size of $C_{\mu_0}$, the norm of the resolvent of a non-normal operator, \ie~its pseudospectrum, see \cite{Trefethen-2005}, may be dramatically larger than $\dist(z,\sigma(A))^{-1}$. While there is a collection of recent pseudospectral results for non-normal differential operators, \cf~\cite{Davies-2007,Helffer-2013-book,Trefethen-2005}, the estimates on the norm of the resolvent of $(A-\mu_0) \restriction Q_0 \Dom(A)$ acting in $Q_0 \H$ seem not to be available.

On the other hand, there exists a large collection of perturbation results, particularly for operators with compact resolvent, \cf~for instance the classical \cite{Kato-1966,DS3,Markus-1988} or more recent \cite{Shkalikov-2010-269,Wyss-2010-258,Adduci-2012-10,Adduci-2012-73,Mityagin-2013a}, guaranteeing that the eigensystem of a perturbed selfadjoint (or normal) operator contains a Riesz basis, hence $\|P_n\|$, $\|Q_n\|$ and $\|(Q_n (A-\mu_n)Q_n)^{-1}\|$ are uniformly bounded ($\|P_n\|$ are uniformly bounded already if there is only a basis). The Riesz basis property is present for example for operators in \eqref{E:A_sing2}, \eqref{BEC.models.1}, \eqref{BEC.models.2} and \eqref{A.Rub} from Section \ref{sec:appl}.


\section{The role of symmetries}
\label{sec:sym}
We show that under \emph{antilinear} symmetry assumptions on $A$ and $f$, \cf~Assumption \ref{ass:A.C} below, the nonlinear eigenvalue $\mu$ starting from a real eigenvalue $\mu_0$ remains real and a certain symmetry of the solution $\psi$ is preserved for all $\varepsilon \in (-\varepsilon_0,\varepsilon_0)$. The symmetry of solutions is preserved also in the case of linear symmetries that are studied next.

\subsection{Antilinear symmetries}
\label{subsec:anti.sym}

\begin{ass}[Antilinear symmetries of $A$ and $f$]\label{ass:A.C}
	
\noindent
Let $A$ be a densely defined and closed operator in a Hilbert space $\H$ and let $\cC$ be an \emph{antilinear}, isometric and involutive operator, \ie~for all $\phi,\psi \in \H$ and $\lambda \in \C$, $\cC(\lambda \phi+\psi)=\ov \lambda \cC\phi+\cC\psi$, $\langle \cC \phi, \cC \psi \rangle = \langle \psi, \phi \rangle$ and $\cC^2 = I$, such that
\begin{enumerate}[(a)]
	\item \label{ass:A.C.A} 
	for all $\psi \in \Dom(A)$,  
	\begin{equation}\label{ass:A.C.A.eq}
	\cC \psi \in \Dom(A) \quad \text{and} \quad A \cC \psi = \cC A \psi,
	\end{equation}
	\item \label{ass:A.C.f}
	for all $\psi \in \X$, where $\X$ is the space from Assumption \ref{ass:Af}.\eqref{ass:Af.Lip},
	\begin{equation}\label{ass:A.C.f.eq}
	\cC f(\psi) = f(\cC \psi).
	\end{equation}
\end{enumerate}
\end{ass}


The operator $\cC$ is referred to as the \emph{antilinear symmetry} of $A$ and $f$. The standard example of $\cC$ is the $\PT$ symmetry that is naturally present in various physical models as we indicate in the following example and in Section \ref{sec:appl}. 

\begin{example}[$\PT$ symmetry]
\label{ex.PT.sym}

Let $\H = L^2(\Omega)$, where $\Omega = \R^d$ or $\Omega = (-r,r)^d$ with $r \in (0,+\infty)$. Define
\begin{equation}\label{PT.def}
(\P \psi)(x):=\psi(-x), \qquad \T \psi:=\overline{\psi}.
\end{equation}
In quantum mechanics, $\P$ corresponds to the space reflection (parity) and $\T$ is the time-reversal. The antilinear $\PT$ symmetry is the composition, \ie~$\cC:=\PT$. In more dimensional domains, the so called partial $\P_i\T$ symmetries, where 
\begin{equation}\label{Pi.def}
(\P_i \psi)(x_1,\dots,x_i,\dots,x_d):=\psi(x_1,\dots,-x_i,\dots,x_d),
\end{equation}
are sometimes considered, \cf~\cite{Borisov-2008-62} or \cite{Yang-2014-39}. Obviously, also $\P_i\T$ is antilinear.

Schr\"odinger operators $-\Delta +V$ with complex potentials $V$, \cf~Examples \ref{ex.A.Schr} and \ref{ex.A.Schr.n}, are $\PT$-symmetric if $V$ is $\PT$-symmetric, \ie~$[\PT,V] = 0$, or, in other words, the real and imaginary parts of $V$ satisfy $(\Re V)(-x)= (\Re V)(x)$ and $(\Im V)(-x) = -(\Im V)(x)$. The Schr\"odinger operators from Example \ref{ex.A.sing} posses this symmetry if, for every $\phi,\psi \in \Dom(a)$,  $v_2(\phi,\PT \psi)= \ov{v_2(\PT\phi,\psi)}$ holds, see Section \ref{sec:appl} for examples of such $v_2$.

The polynomial nonlinearity $f_{\rm pol}(\psi)$ from Example \ref{ex.f.pol} 
satisfies Assumption \ref{ass:A.C}.\eqref{ass:A.C.f} with $\cC=\PT$ if and only if $a_{pq}(-x)=\overline{a_{pq}}(x)$ and the space $\X$ can be selected as in Example \ref{ex.A.Schr}, \ref{ex.A.sing}, \ref{ex.A.Schr.n} or  \ref{ex.A.Schr.qp}. For the quasi-periodic case in Example \ref{ex.A.Schr.qp} note that $\psi$ is $k$-quasi-periodic with a given vector $k\in (-\pi,\pi]^d$ if and only if $\overline{\psi(-x)}$ is $k$-quasi-periodic.
The monopolar and dipolar interactions $f_{\rm m}$ and $f_{\rm d}$ from Example \ref{ex.f.mon.dip} and the nonlocal nonlinearity in Example \ref{ex.f.Rubinstein} are also $\PT$-symmetric as it can be easily checked.
\end{example}

We recall simple facts about spectral properties of $\cC$-symmetric operators.

\begin{lemma}\label{lem:A.C}
Let $A$ satisfy Assumption \ref{ass:A.C} and let $\mu_0$ be an eigenvalue of $A$. Then
\begin{enumerate}[\upshape(i)]
\item \label{lem:A.C.i}
$\ov{\mu_0}$ is an eigenvalue of $A$ and if there is a $\cC$-symmetric eigenvector $\psi_0$ corresponding to $\mu_0$, \ie~ 
\begin{equation}\label{psi0.C}
(A-\mu_0) \psi_0 = 0, \quad \cC \psi_0 = \psi_0,
\end{equation}
then $\mu_0 \in \R$. Moreover, if $\mu_0$ is real and simple, then the corresponding eigenvector can be chosen $\cC$-symmetric.
\item if $\mu_0$ is isolated simple and real, then, in addition, the spectral (Riesz) projection $P_0$ of $A$ corresponding to $\mu_0$ commutes with $\cC$, \ie~ 
\begin{equation}\label{P0.C}
P_0 \,  \cC = \cC P_0.
\end{equation}
\end{enumerate}
\end{lemma}
\begin{proof}
\begin{enumerate}[(i)]
\item Let $\widetilde \psi_0 \in \Dom(A)$, $\widetilde \psi_0 \neq 0$ satisfy $(A-\mu_0) \widetilde \psi_0 =0$.  
Clearly, by the $\cC$-symmetry of $A$, $A \cC \widetilde \psi_0 = \cC A \widetilde \psi_0 = \ov{\mu_0} \, \cC \widetilde \psi_0$, hence $\ov{\mu_0}$ is an eigenvalue of $A$. 

Next, let $\psi_0$ be as in \eqref{psi0.C}. Then 
$$\mu_0 \psi_0 =A\cC\cC\psi_0=\cC A\psi_0 = \cC\mu_0 \psi_0=\overline{\mu_0}\psi_0,$$
thus $\mu_0 \in \R$.

Finally, let again $\widetilde \psi_0$ satisfy $(A-\mu_0)\widetilde \psi_0 = 0$ and $\mu_0$ be simple. It follows from \eqref{ass:A.C.A.eq} that $A \cC \widetilde \psi_0 = \mu_0 \cC \widetilde \psi_0$, hence both of $\widetilde \psi_0 \pm \cC \widetilde \psi_0$ are eigenvectors of $A$. However, one of them must be 0 since $\mu_0$ is simple. Note that $\widetilde \psi_0 + \cC \widetilde \psi_0$ is automatically $\cC$-symmetric and if it is zero, then we take $\ii (\widetilde \psi_0 - \cC \widetilde \psi_0)$.
\item Since $A$ commutes with $\cC$, we have $\cC(A-z)^{-1}= (A-\ov z)^{-1} \cC$. Hence, for the spectral projection $P_0$, \cf~\eqref{P0.def},
we get
\begin{align*}
\cC P_0 & = - \cC \frac 1 {2\pi } \int_0^{2\pi} (A-(\mu_0 + \delta e^{\ii \varphi}))^{-1} \delta e^{\ii \varphi} \dd \varphi 
\\
& = - \frac 1 {2\pi } \int_0^{2\pi} (A-(\mu_0 + \delta e^{-\ii \varphi}))^{-1} \delta e^{-\ii \varphi} \dd \varphi \, \cC \\
& = - \frac 1 {2\pi } \int_{-2\pi}^0(A-z)^{-1} \dd z\, \cC= P_0 \, \cC.
  \qedhere
\end{align*}
\end{enumerate}
\end{proof}

\begin{remark}[Stability of real simple eigenvalues]\label{rem:sim+real}

For a $\cC$-symmetric family $A(\gamma)$, \cf~Section \ref{subsubsec:hol.fam}, isolated simple real eigenvalues of $A(0)$ remain isolated simple and real for small $|\gamma|$ since eigenvalues are analytic in $\gamma$ and always form complex conjugated pairs, \cf~Lemma \ref{lem:A.C}.\eqref{lem:A.C.i}.

Notice that many examples in literature as well as in Section \ref{sec:appl} can be in fact viewed as a holomorphic family with $A(0)=A(0)^*$, thus $\sigma(A(0)) \subset \R$. A typical behavior of real eigenvalues as $\gamma$ is increased, \ie~the non-symmetric part of the operator becomes stronger, is a tendency to merge and create a complex conjugated pair, see \eg~Figure \ref{Fig:SHO6_BD_gam} (a), (b).

Finally, we remark that there are also $\cC$-symmetric (actually with $\cC=\PT$) operators with real spectrum that are not small perturbations of a selfadjoint operator, \eg~the celebrated imaginary cubic oscillator $-\partial_x^2 + \ii x^3$ acting in $L^2(\R)$, \cf~\cite{Bender-1998-80,Shin-2002-229}.
\end{remark}

The following theorem shows that real simple eigenvalues $\mu_0$ of an operator $A$ persist to be \emph{real nonlinear eigenvalues} $\mu(\varepsilon)$ of $A-\eps f$ for all $\varepsilon$ small enough if $A$ and $f$ possess an \emph{antilinear} symmetry. Moreover, starting with a $\cC$-symmetric linear eigenfunctions $\psi_0$ of $A$ (the existence of which is guaranteed by Lemma \ref{lem:A.C}), the nonlinear eigenfunctions $\psi(\varepsilon)$ are also $\cC$-symmetric.

\begin{theorem}\label{thm.sym}
Let $A$ and $f$ satisfy Assumptions \ref{ass:Af} and \ref{ass:A.C}. Suppose in addition that $\mu_0 \in \R$ and choose the corresponding eigenvector $\psi_0$ as $\cC$-symmetric, \ie~ $\cC \psi_0 = \psi_0$.
Then, for all $\varepsilon \in (-\varepsilon_0,\varepsilon_0)$, the nonlinear eigenpair $(\mu, \psi)$ from Theorem \ref{thm.NL-fixpt} satisfies $\mu \in \R$ and $\cC \psi = \psi$. 
\end{theorem}

\begin{proof}
Recall that	$\mu =\mu_0 + \varepsilon\nu + \varepsilon^2\sigma$, $\psi=\psi_0+\varepsilon\phi+\chi$ and $P_0$ is the spectral projection of $A$ corresponding to the eigenvalue $\mu_0$. 

In the first step, we show that $\nu$ is real and $\cC \phi = \phi$. The spectral projection $P_0$ can be written as $P_0= \langle \cdot ,\psi_0^* \rangle \psi_0$, where $\psi_0^*$ is as in Theorem \ref{thm.NL-fixpt}, therefore $- \nu \psi_0 = P_0(f(\psi_0))$, \cf~\eqref{nu.def}. Using symmetries \eqref{ass:A.C.f.eq}, \eqref{psi0.C} and \eqref{P0.C}, we obtain
\begin{equation}\label{nu.real}
\begin{aligned}
\cC(- \nu \psi_0) = \cC P_0(f(\psi_0)) = P_0(f( \cC \psi_0)) = P_0(f(\psi_0))=- \nu \psi_0,
\end{aligned}
\end{equation}
thus $- \ov \nu \psi_0 = - \nu \psi_0$, hence $\nu \in \R$. Applying $\cC$ to equation \eqref{E:phi_thm}, using the symmetries of $A$, $f$ and $\psi_0$, \cf~\eqref{ass:A.C.A.eq}, \eqref{ass:A.C.f.eq} and \eqref{psi0.C}, we get
\begin{equation}
(A - \mu_0) \cC \phi = \nu \psi_0 + f(\psi_0).
\end{equation}
Because $P_0 \, \cC \phi = \cC P_0 \phi$ by \eqref{P0.C} and because the solution of \eqref{E:phi_thm} with $P_0 \phi=0$ is unique, we have $\cC \phi = \phi$.

Let us now work on $\sigma$. Since $\sigma$ is the solution of the fixed point problem $\sigma =S(\sigma)$ with $S(\sigma) = -\frac 1\varepsilon \langle f(\psi(\sigma))-f(\psi_0),\psi_0^*\rangle$, where $\psi(\sigma)=\psi_0+\varepsilon\phi +\chi(\sigma)$ and $\chi$ solves the fixed point equation $\chi = G(\chi;\sigma)$, it remains to show that the coupled fixed point problem preserves the realness of $\sigma$ and the $\cC$-symmetry of $\chi$. 

Given $\sigma\in \R$ (with $|\sigma|\leq r_1$), we prove that 
\begin{equation}\label{E.chi-sym}
\cC \chi = \chi  \quad \Longrightarrow \quad  \cC G(\chi;\sigma) = G(\chi;\sigma).
\end{equation}
As $G(\chi;\sigma) = (Q_0(A-\mu_0)Q_0)^{-1} R(\chi;\sigma),$
we first show the analogous property for $R$ and then the commutation of $(Q_0(A-\mu_0)Q_0)^{-1}$ with $\cC$.
Since $Q_0 = I - P_0$, we get from \eqref{P0.C} that $Q_0 \,  \cC = \cC \, Q_0$ as well.
Also note that for $\cC\chi=\chi$ the full solution $\psi=\psi_0+\phi+\chi$ is $\cC$-symmetric.
Hence, for $\varepsilon, \sigma \in \R$ and $\cC \chi = \chi$,
\begin{equation}
\begin{aligned}
\cC R(\chi;\sigma) &= \varepsilon \big( (\nu+\varepsilon\sigma)(\chi+\varepsilon\phi) +Q_0 \cC (f(\psi)-f(\psi_0))\big)\\
&=\varepsilon \big( (\nu+\varepsilon\sigma)(\chi+\varepsilon\phi) +Q_0 (f( \cC\psi)-f( \cC\psi_0))\big)=  R(\chi;\sigma).
\end{aligned}
\end{equation}
To prove \eqref{E.chi-sym}, it remains to show that $\cC (Q_0(A-\mu_0)Q_0)^{-1}=(Q_0(A-\mu_0)Q_0)^{-1}\cC$. To this end, take any $\varphi \in Q_0 \H$, then 
\begin{equation}\label{E.chi-sym.2}
\begin{aligned}
& \cC (Q_0(A-\mu_0)Q_0)^{-1} \varphi  
\\
& \quad = (Q_0(A-\mu_0)Q_0)^{-1} (Q_0(A-\mu_0)Q_0) \cC (Q_0(A-\mu_0)Q_0)^{-1} \varphi
\\
& \quad = (Q_0(A-\mu_0)Q_0)^{-1} \cC \varphi.
\end{aligned}
\end{equation}
%

Property \eqref{E.chi-sym} implies that the fixed point of $\chi=G(\chi)$ in $B_{r_2 \varepsilon^2}$, \cf~\eqref{Brg2.def}, lies in  $B_{r_2 \varepsilon^2} \cap \{\eta\in \H: \cC \eta = \eta\}$.

Finally, we need to show that
\begin{equation}
\cC \chi = \chi  \quad \Longrightarrow \quad S(\sigma) \in \R.
\end{equation} 
Once again, because $\cC \chi = \chi$ implies $\cC\psi=\psi$, we get by a straightforward manipulation analogous to \eqref{nu.real},
\begin{equation}
\begin{aligned}
\cC(-\varepsilon S(\sigma) \psi_0) &= \cC P_0( f(\psi)- f(\psi_0)) 
= P_0( f( \cC \psi)- f(\cC\psi_0)) 
= -\varepsilon S(\sigma) \psi_0,
\end{aligned}
\end{equation}
hence $S(\sigma) \in \R$ by the same arguments as below \eqref{nu.real}.
\end{proof}

\subsection{Linear symmetries}

The operator $A$ and the nonlinearity $f$ may possess also a \emph{linear} symmetry $\cS$. Then (for simple eigenvalues) symmetry or antisymmetry of the nonlinear 
eigenfunctions $\psi(\varepsilon)$ is preserved as well, however, the preservation of the realness of $\mu$ cannot be concluded based solely on a linear symmetry.

\begin{ass}[Linear symmetries of $A$ and $f$]\label{ass:A.S}
Let $A$ be a densely defined and closed operator in a Hilbert space $\H$ and let $\cS$ be a \emph{linear}, selfadjoint and involutive operator, \ie~for all $\phi,\psi \in \H$ and $\lambda \in \C$, $\cS(\lambda \phi+\psi)=\lambda \cS\phi+\cS\psi$, $\langle \cS \phi, \psi \rangle = \langle \phi, \cS\psi \rangle$ and $\cS^2 = I$, such that
\begin{enumerate}[(a)]
\item \label{ass.A.S} 
for all $\psi \in \Dom(A)$,  
\begin{equation}\label{A.Ssym}
\cS \psi \in \Dom(A) \quad \text{and} \quad A \cS \psi = \cS A \psi,
\end{equation}
\item \label{ass.f.S} 
for all $\psi \in \X$, where $\X$ is the subspace from Assumption \ref{ass:Af}.\eqref{ass:Af.Lip},
\begin{equation}\label{E:f_S}
\cS f(\psi) = f(\cS \psi) \quad \text{and} \quad f(\pm \psi) = \pm f(\psi).
\end{equation}
\end{enumerate}
\end{ass}

\begin{lemma}\label{lem:A.S}
Let $A$ satisfy Assumption \ref{ass:A.S} and let $\mu_0$ be a simple eigenvalue of $A$. Then
\begin{enumerate}[\upshape(i)]
\item the eigenvector $\psi_0$, $\|\psi_0\|=1$, corresponding to the eigenvalue $\mu_0$ is either $\cS$-symmetric or $\cS$-antisymmetric, \ie~ 
\begin{equation}\label{psi0.S}
\cS \psi_0 = \s \psi_0, \quad \s = 1 \text{ or } \s = -1.
\end{equation}
\item the spectral projection $P_0$ of $A$ corresponding to $\mu_0$ commutes with $\cS$, \ie~ 

\begin{equation}\label{P0.S}
P_0 \,  \cS = \cS  P_0.
\end{equation}
\end{enumerate}
\end{lemma}
\begin{proof}
The reasoning is analogous to the one in the proof of Lemma \ref{lem:A.C}.
\end{proof}

\begin{theorem}\label{thm.sym.lin}
Let $A$ and $f$ satisfy Assumptions \ref{ass:Af} and \ref{ass:A.S}. 
Then the symmetry of the eigenvector $\psi_0$, corresponding to the simple eigenvalue $\mu_0$, \ie~
$\cS \psi_0 = \s \psi_0$ with $\s=1$ or $\s = -1$, is preserved for the nonlinear eigenfunction $\psi$ from Theorem \ref{thm.NL-fixpt}, \ie~$\cS \psi = \s \psi$ for all $\varepsilon \in (-\varepsilon_0,\varepsilon_0)$ with the same $\s$ as for $\psi_0$. 
\end{theorem}

\begin{proof}
Similarly to the proof of Theorem \ref{thm.sym}, we show that $\cS \phi = \s \phi$ and search for the fixed point in a subset $B_{r_2 \varepsilon^2} \cap \{\eta\in \H: \cS \eta = \s \eta\}.$ Note that this is allowed since, similarly as in \eqref{E.chi-sym}--\eqref{E.chi-sym.2},
\begin{align*}
\cS \chi &= \s \chi  \quad \Longrightarrow \quad  \cS G(\chi;\sigma) = \s G(\chi;\sigma). 
\qedhere
\end{align*}
\end{proof}

\section{Applications}
\label{sec:appl}

\subsection{Toy model}

Let $\H = L^2((-r,r))$ with $r =\pi/2$, $\gamma \in \C$ and let $A_\gamma$ be the m-sectorial operator associated with the form
\begin{equation}
\begin{aligned}
a_\gamma[\psi] := \|\psi'\|^2 + \gamma 
\left(
|\psi(\pi/2)|^2 - |\psi(-\pi/2)|^2
\right), 
\quad
\Dom(a_\gamma) := H^1((-\pi/2,\pi/2)). 
\end{aligned}
\end{equation} 
Since $V_1=0$, we take $\X=\H^1_0((-\pi/2,\pi/2))=H^1((-\pi/2,\pi/2))$, \cf~Example \ref{ex.A.sing}. Note that the inequality, valid for every $\varepsilon >0$,
\begin{equation}\label{delta.ineq}
\|\psi\|_{L^{\infty}}^2 \leq C \|\psi\|_{H^1} \|\psi\| \leq \varepsilon \|\psi\|_{H^1}^2 + C(\varepsilon) \|\psi\|^2 , 
\end{equation}
implies that condition \eqref{w.rb} is satisfied. By standard arguments, \cf~\cite[Ex.VI.2.16]{Kato-1966}, 
\begin{equation}\label{E:A_sing2}
\begin{aligned}
A_\gamma  & = -\partial_x^2, 
\\
\Dom(A_\gamma)  &= \{\psi \in H^2((-\pi/2,\pi/2)) \, : \, \psi'(\pm \pi/2) + \gamma \psi(\pm \pi/2)  = 0\}, \end{aligned}
\end{equation}
and $(A_\gamma)^*  = A_{\ov \gamma}$, \ie~$A_\gamma$ is a non-selfadjoint (unless $\gamma \in \R$) perturbation of the Neumann Laplacian on $(-\pi/2,\pi/2)$ in boundary conditions; a collection of spectral results for $A_\gamma$ can be found \eg~in \cite{Krejcirik-2006-39,Krejcirik-2014-8}. The spectrum of $A_\gamma$ is discrete with explicit eigenvalues:
\begin{equation}\label{A.gamma.EV}
\sigma(A_{\gamma})  = \{ -\gamma^2 \} \cup \{n^2\}_{n=1}^{\infty}.
\end{equation}
	
For $\gamma = \ii \alpha$, $\alpha \in \R\setminus \Z$, the operator $A_{\ii \alpha}$ is $\PT$-symmetric, \cf~\cite{Krejcirik-2006-39} and Example \ref{ex.PT.sym}, and all its eigenvalues are real and simple.
Hence the spectral condition in Assumption \ref{ass:Af}.\eqref{ass:Af.mu} is satisfied for any $\mu_0 \in \sigma(A_{\ii \alpha})$. The eigenfunctions $\{\xi_n\}_{n=0}^\infty$, $\{\xi_n^*\}_{n=0}^\infty$ of $A_{\ii \alpha}$ and $A_{\ii \alpha}^*$, respectively, with normalization satisfying \eqref{psi0.norm} read
\begin{equation}
\begin{aligned}
\xi_0(x) &=  \frac{1}{\sqrt\pi} \, e^{- \ii \alpha x},  
\\
\xi_n(x) &= \sqrt{\frac{2}{\pi}}\frac{n }{\sqrt{n^2+\alpha^2}} \left(\cos\left(n \left(x+\frac \pi2 \right)\right) - \frac{\ii \alpha}{n} \sin\left(n \left(x+\frac \pi2 \right)\right) \right), \quad  n \geq 1,
\\
\xi_0^*(x) &= \frac{\alpha \pi}{\sin (\alpha \pi)} \, \overline{\xi_0(x)}, \qquad 
\xi_n^*(x) = \frac{n^2+\alpha^2}{n^2-\alpha^2} \, \overline{\xi_n(x)}, \quad n \geq 1.
\end{aligned}
\end{equation}

We consider the cubic nonlinearity $f_{\rm c}$ from Example \ref{ex.f.pol} that is $\PT$-symmetric as well, \cf~Example \ref{ex.PT.sym}. By Theorems \ref{thm.NL-fixpt}, \ref{thm.sym} and Corollary \ref{cor.hom.f}, we have the nonlinear eigenpair $(\mu(\varepsilon),\psi(\varepsilon))$ with $\|\psi(\varepsilon)\|=1$ and $\mu(\varepsilon) \in \R$ for any  $\mu_0 \in \sigma(A_{\ii \alpha})$ and small $\varepsilon$. 

For $\mu_0 =n^2, n\geq 1$, straightforward calculations lead to 
\begin{equation}
\begin{aligned}
\nu_n &= - \langle \xi_n|\xi_n|^2, \xi_n^* \rangle = - \frac{3}{2\pi}
\end{aligned}
\end{equation}
and also $\phi_n$ can be calculated explicitly by solving \eqref{E:phi_thm}.
	
For the special eigenvalue $\mu_0=\alpha^2$ and the corresponding eigenfunction $\psi_0=\xi_0$, the normalized nonlinear eigenpair can be found even explicitly, namely 
\begin{equation}
\mu(\varepsilon) = \alpha^2 - \frac{\varepsilon}{\pi}, \quad \psi(\varepsilon) = \psi(0) \equiv \frac{1}{\sqrt{\pi}}  e^{- \ii \alpha x}.
\end{equation}
Notice the agreement with the expansion of $\mu$ and $\psi$ from Theorem \ref{thm.NL-fixpt}. Indeed, since $f_{\rm c}(\psi_0)=\psi_0/\pi$, identity \eqref{nu.def} yields $\nu_0 = -1/\pi$. Equation \eqref{E:phi_thm} for $\phi \in Q_0 \Dom(A_{\ii \alpha})$ then reads
\begin{equation}
-\phi'' - \alpha^2 \phi = 0,
\end{equation}
hence $\phi=0$. Moreover, $\psi(\varepsilon) = \psi(0)$ implies that $\sigma=0$ is the fixed point of 
\eqref{sigma.eq_fp}.

\subsection{Bose-Einstein condensates with injection and removal of particles}
\label{subsec:BEC}

In the physics literature on Bose-Einstein condensates, non-selfadjoint perturbations of harmonic oscillators or Laplacians with $\delta$-interactions are considered; the imaginary part of the linear potential models the injection and removal of particles, \cf~\cite{Cartarius-2012-45,Dast-2013-46,Fortanier-2014-89} for instance. 

The nonlinear part of the problem corresponds to the contact (cubic) $f_{\rm c}$, monopolar $f_{\rm m}$ or dipolar $f_{\rm d}$ interaction, \cf~Example \ref{ex.f.mon.dip}; the unit vector $\alpha$ entering the dipolar interaction represents the direction of the polarization, \cf~\cite{Lahaye-2009-72,Dast-2013-46}. 
Spectral parameter $\mu$ corresponds to the chemical potential, the parameter $\varepsilon$  controls the strength of the nonlinear interaction and the intensity of particle removal and injection (the non-selfadjoint part of the linear operator) is described by the parameter $\gamma$, see \eqref{BEC.models.1}--\eqref{BEC.models.3}. The balance in the removal and injection is reflected in the $\PT$ symmetry of the system (the imaginary part of the potential is odd). 
One-dimensional examples (obtained by the separation of variables in $d=3$ models) of $\PT$-symmetric linear parts $A$ from the literature are the following
\begin{align}
&-\partial_x^2 + x^2 + v_0 e^{-\sigma x^2} + \ii \gamma x e^{-\rho x^2}, && \gamma, v_0 \in \R, \rho \geq 0,
\label{BEC.models.1} \\
&-\partial_x^2 + x^2 + \ii \gamma (\delta(x-\tau) - \delta(x+\tau)), && \gamma \in \R, \tau > 0,
\label{BEC.models.2} \\
&-\partial_x^2 - (1-\ii \gamma) \delta(x-\tau) - (1 + \ii \gamma) \delta(x+\tau), && \gamma  \in \R, \tau >0.
\label{BEC.models.3}
\end{align}
Notice that all can be viewed as holomorphic operator families $A(\gamma)$ with $A(0)=A(0)^*$, \cf~Section \ref{subsubsec:hol.fam}.

Model \eqref{BEC.models.1} corresponds to an operator $A$ from Example \ref{ex.A.Schr} with $\sigma(A) = \spd(A)$, \cf~Section \ref{subsubsec:Schr.sp}. For $v_0$ and $|\gamma|$ small, all eigenvalues of $A$ are simple and real, \cf~Remark \ref{rem:sim+real}, moreover, it follows \eg~from \cite{Mityagin-2013a} that the number of non-real eigenvalues is finite for any $v_0, \gamma \in \R$; a numerical analysis of eigenvalues for \eqref{BEC.models.1} can be found in \cite{Dast-2013-46}.

Models \eqref{BEC.models.2}--\eqref{BEC.models.3} correspond to the singular Schr\"odinger operator from Example \ref{ex.A.sing}. Both can be introduced through the closed sectorial form $a$, \cf~\eqref{a.form.H1}, namely $V_1(x)=x^2$, $v_2[\psi]=\ii \gamma (|\psi(\tau)|^2 - |\psi(-\tau)|^2$ for \eqref{BEC.models.2} and $V_1(x)=0$, $v_2[\psi]=(-1 + \ii \gamma ) |\psi(\tau)|^2 - (1 + \ii \gamma ) |\psi(-\tau)|^2$ for \eqref{BEC.models.3}. 
Note that condition \eqref{w.rb} is satisfied (for any $\gamma \in \C$) because of \eqref{delta.ineq}.
For \eqref{BEC.models.2}, the spectrum is purely discrete and real for sufficiently small $|\gamma|$, \cf~Section \ref{subsubsec:Schr.sp} and Remark \ref{rem:sim+real}, the number of non-real eigenvalues is finite for any $\gamma \in \R$, \cf~\cite{Mityagin-2015} for a detailed spectral analysis and \cite{Haag-2014-54} for a numerical analysis of eigenvalues.
The essential spectrum for \eqref{BEC.models.3} is equal to $[0,+\infty)$ and discrete eigenvalues outside of $[0,+\infty)$ may appear, \cf~Section \ref{subsubsec:Schr.sp} and Remark \ref{rem:sim+real}. 
In more detail, if
\begin{equation}\label{ex.2delta.EV}
(1+\gamma^2)\tau  > 1 
\quad \mbox{and} \quad 
(1+\gamma^2) e^{ - 2\tau } > \gamma^2,
\end{equation} 
then there are two negative eigenvalues that are obtained as solutions $\mu$ of
\begin{equation}
\label{ex.2delta.EV.eq}
e^{-4 \tau \sqrt{-\mu}} = \frac{4}{1+\gamma^2} 
\left(-\mu - \sqrt{-\mu}  + \frac{\gamma^2}{4} \right), 
\end{equation}
\cf~\cite{Kondej-2013-49,Lotoreichik-2014}. A numerical analysis of eigenvalues of \eqref{BEC.models.3} can be found in \cite{Cartarius-2012-45}. 

In summary, Theorem \ref{thm.NL-fixpt}, Corollary \ref{cor.hom.f} and Theorem \ref{thm.sym} prove the effects observed in physics literature, \ie~for $\eps \neq 0$, nonlinear eigenvalues are shifted with respect to linear ones and those $\mu$ that start from real simple linear eigenvalues $\mu_0$ are real and the $\PT$ symmetry of the nonlinear solution $\psi$ is preserved. 
Note that the latter applies for all nonlinear interactions $f_{\rm pol}$, $f_{\rm m}$ or $f_{\rm d}$ mentioned above and a collection of Schr\"odinger operators in Examples \ref{ex.A.Schr} and \ref{ex.A.sing}. 

A numerical analysis of a model with $d=2$, which is qualitatively similar to \eqref{BEC.models.1}, is performed in Example \ref{ex:num_PT}.

\subsection{Spin-orbit-coupled Bose-Einstein condensate}
\label{subsec:BEC.spin}

The spin-orbit-coupled Bose-Einstein condensate is described by the spinor $\phi \in \H:= L^2(\R) \times L^2(\R)$ obeying the equation
\begin{equation}
\ii \partial_t \phi =  A \phi + \varepsilon f(\phi)
\end{equation}
with
\begin{equation}
\label{Af.BEC.spin}
A = \begin{pmatrix}
- \partial_x^2 + V(x) + \ii \gamma & \omega + \ii \kappa \partial_x 
\\
\omega + \ii \kappa \partial_x  & - \partial_x^2 + V(x) - \ii \gamma 
\end{pmatrix},
\quad f(\phi) = 
\begin{pmatrix}
(|\phi_1|^2 + |\phi_2|^2) \phi_1
\\
(|\phi_1|^2 + |\phi_2|^2) \phi_2
\end{pmatrix}
\end{equation}
where $V$ is a trap potential satisfying the conditions in Example \ref{ex.A.Schr}, $\kappa \in \R$ is the strength of the spin-orbit coupling, $\omega \in \R$ is the strength of the linear Zeeman coupling, $\varepsilon$ characterizes the inter-atomic interactions and $\gamma >0$ accounts for the decay and the gain of the (pseudo-)spin states up and down, respectively, \cf~\cite{Kartashov-2014-107}.  The time harmonic ansatz $\phi(x,t)=e^{-\ii \mu t}\psi(x)$ leads to the eigenvalue problem \eqref{nl.ev}. 

Since the off-diagonal part of $A$ is relatively bounded with respect to the diagonal part (\ie~two copies of a Schr\"odinger operator from Example \ref{ex.A.Schr}) with the bound $0$, the operator $A$ in \eqref{Af.BEC.spin} with $\Dom(A) = \H^2_{V_1}(\R) \times \H^2_{V_1}(\R)$ has the graph norm equivalent to $\|\psi\|_{\X}:=\|\psi_1\|_{\H_{V_1}^2} + \|\psi_2\|_{\H_{V_1}^2}$, see Example \ref{ex.A.Schr}. Moreover, $A$ has compact resolvent if $|V_1(x)| \to \infty $ as $|x| \to \infty$ and it is self-adjoint for real $V$ and $\gamma = 0$. Hence, real simple eigenvalues (if any) of $A$ for $\gamma =0$ stay simple and real for $\gamma$ small, \cf~Remark \ref{rem:sim+real}. For instance, if $V(x)=x^2$, then, for $\gamma = 0$, the eigenvalues of $A$ read $2n+1 \pm \omega - \kappa^2/4$, $n \in \N_0$ and are all simple if $\omega \neq 0$, \cf~\cite{Kartashov-2014-107}. Notice that the nonlinearity $f$ from \eqref{Af.BEC.spin} satisfies Assumption \ref{ass:Af}.\eqref{ass:Af.Lip} with $\X:=\Dom(A)$, $\|\cdot\|_{\X}$ as above, \cf~the remark on the vector case at the end of Sec. \ref{ex.f.pol}.

If $V$ is $\PT$-symmetric, \ie~$V(x) = \ov{V(-x)}$, the operator $A$ and the nonlinearity $f$ in \eqref{Af.BEC.spin} possess a natural antilinear symmetry $(\mathcal C\psi)(x):=\sigma_1 \ov{\psi(-x)}$, where $\sigma_1$ is the Pauli matrix, being the composition of the parity, time and charge symmetries, \cf~\cite{Kartashov-2014-107}.
Hence our results show the existence of the stationary nonlinear modes bifurcating from the linear ones, particularly for the parabolic trap $V(x)=x^2$ investigated in \cite{Kartashov-2014-107}.

\subsection{Optics: nonlinear Schr\"odinger-type equations}
\label{subsec:optics}

Another set of physical applications of nonlinear $\PT$-symmetric problems is optics, \cf~for instance \cite{Musslimani-2008-100, Yang-2014-39}, where light propagation is modeled by nonlinear
Schr\"odinger-type equations
$$\ii \partial_z u +\Delta u -V(x)u+f(u)=0,$$
with a gauge invariant $f$, \ie~$f(e^{\ii \alpha}u)=e^{\ii \alpha}f(u)$ for all $\alpha\in \R$, and 
typically with $V \in L^\infty(\Omega)$. The variable $z$ is the propagation direction of the optical waves. The real part of the optical potential $V$ corresponds to the refractive index and $\Im V$ models the gain and loss of the medium. If the latter is balanced, in the sense that $(\Im V)(-x)=-(\Im V)(x)$, and if, in addition $(\Re V)(-x)=(\Re V)(x)$ and $f$ is $\PT$-symmetric, then the whole system is $\PT$-symmetric. For a heterogeneous material, $V$ can generally be discontinuous and, in that case, the decomposition to $V_1$ and $V_2$ complying with $V_1 \in W^{1,\infty}_{\rm loc}(\R^d)$ must be selected to fit into the setting of Example \ref{ex.A.Schr}; however notice that any bounded $V$ fits there after setting $V_2=V$. 
The physically most usual nonlinearity is $f_{\rm c}$, \ie~the cubic one, \cf~Example \ref{ex.f.pol}.

Examples of smooth $V$ in $d=2$ from \cite{Yang-2014-39} and in $d=1$ from \cite{Musslimani-2008-100} are
\begin{align}
 V(x_1,x_2) &= - (v_0+\ii \gamma x_1x_2) e^{-x_1^2} e^{-x_2^2}, & 
\gamma, v_0 \in \R, \quad  
\label{E:partial_PT_V}
\\
V(x_1,x_2) &=  -3 v_0 \left(e^{-\left(x_1-a\right)^2-\left(x_2-a\right)^2}+ e^{-\left(x_1+a\right)^2-\left(x_2-a\right)^2}\right)   \\
& \quad - 2 v_0 \left(e^{-\left(x_1-a\right)^2-\left(x_2+a\right)^2}+e^{-\left(x_1+a\right)^2-\left(x_2+a\right)^2}\right)\\
 & \quad -2\ii \gamma\left(e^{-\left(x_1-a\right)^2-\left(x_2-a\right)^2}-e^{-\left(x_1+a\right)^2-\left(x_2-a\right)^2}\right)\\
& \quad -\ii \gamma\left(e^{-\left(x_1-a\right)^2-\left(x_2+a\right)^2}-e^{-\left(x_1+a\right)^2-\left(x_2+a\right)^2}\right), 
& \gamma,v_0,a \in \R, \quad
\label{E:Gauss9_PT}
\\
V(x) &=   - \cos^2(x)-\ii \gamma \sin(2x),   
& \gamma \in \R. \quad
\label{E:per.PT}
\end{align}
Once again, like for Bose-Einstein condensates, the parameter $\gamma$ determines the strength of non-selfadjointness, \ie~the loss and gain here, and all models can be viewed as holomorphic operator families $A(\gamma)$ with $A(0)=A(0)^*$, \cf~Section \ref{subsubsec:hol.fam}.

The potential $V$ satisfies $|V(x)| \to 0$ as $|x| \to + \infty$ for both \eqref{E:partial_PT_V} and \eqref{E:Gauss9_PT}, hence the essential spectrum of corresponding $A(\gamma)$ is $[0,\infty)$, \cf~Section \ref{subsubsec:Schr.sp}. Since $\int_{\R^2} \Re V \dd x<0$ (and $\Re V$ decays sufficiently fast), there are discrete negative eigenvalues of $A(0)=-\Delta + \Re V$, which are simple for sufficiently small $v_0$, \cf~Section \ref{subsubsec:Schr.sp}. 
Neither one of the potentials \eqref{E:partial_PT_V} and \eqref{E:Gauss9_PT} is $\PT$-symmetric but both are $\P_1\T$-symmetric and \eqref{E:partial_PT_V} is also $\P_2\T$-symmetric, \cf~Example \ref{ex.PT.sym}, hence the simple real eigenvalues of $A(0)$ remain simple and real for sufficiently small $|\gamma|$, \cf~Remark \ref{rem:sim+real}. 

Regarding periodic problems, like \eg~\eqref{E:per.PT}, our results are relevant for the Bloch eigenvalue problem. In the case of $A=-\partial_x^2+V(x)$ with a $2\pi$-periodic $V$, one considers the family of operators $A$ in $L^2((-\pi,\pi))$ with $k$-quasi-periodic boundary condition, \ie~with the domain
$$\Dom(A)=\{ \psi \in H^2((-\pi,\pi)): \psi(\pi) = e^{\ii k}\psi(-\pi), \psi'(\pi) = e^{\ii k}\psi'(-\pi)\}$$
and the form domain $\H_0^{1,k}((-\pi,\pi))$, cf. Example \ref{ex.A.Schr.qp}.
Since $V$ in \eqref{E:per.PT} is $\PT$-symmetric, eigenvalues of $A$ for $k \notin \{0,\pi\}$ are simple and real for sufficiently small $|\gamma|$, \cf~Section \ref{subsubsec:Schr.sp} and Remark \ref{rem:sim+real}. Numerical analysis from \cite{Musslimani-2008-100} suggests that for \eqref{E:per.PT} this is the case if $|\gamma|<1/2$.

In summary, our results in Theorem \ref{thm.NL-fixpt}, Corollary \ref{cor.hom.f} and Theorem \ref{thm.sym} are applicable and provide for \eqref{E:partial_PT_V} and \eqref{E:Gauss9_PT} and any $f$ compatible with Assumptions \ref{ass:Af} and \ref{ass:A.C} real nonlinear eigenvalues $\mu$ of $-\Delta \psi+V\psi -\eps f(\psi)=\mu \psi$ with nonlinear solutions $\psi$ that satisfy the corresponding partial $\PT$ symmetries. For the periodic problem \eqref{E:per.PT}, we obtain nonlinear Bloch functions $\psi(x)=p(x)e^{\ii k x}$, where $p$ is $2\pi$-periodic. This complements the results of \cite{Dohnal-2014} on the bifurcation of nonlinear Bloch waves in the selfadjoint case. In the $z$-dependent nonlinear Schr\"odinger equation we obtain solutions $u(z,x)=e^{-\ii \mu z}\psi(x)$ with a real propagation constant $\mu$ despite the fact that the material exhibits loss and gain. 

A numerical analysis of the model with the potential in \eqref{E:Gauss9_PT} is performed in Example \ref{ex:num_PT1}.

\subsection{Optics: discrete nonlinear Schr\"odinger equation}
\label{subsec:DNLS}

The propagation of light in a finite one dimensional lattice of linearly coupled Kerr-nonlinear fibers is often modeled by the discrete nonlinear Schr\"odinger equation
$$\ii \partial_z u_n = u_{n+1} + u_{n-1} + \ii \gamma (-1)^n u_n +|u_n|^2u_n, \quad 1\leq n\leq 2N, \quad u_0=u_{2N+1}=0,$$
where $z$ is the propagation direction, $n\in \N$ is the lattice site and $\ii \gamma (-1)^n\in \ii\R$ describes the loss or gain at the site $n$, see \cite{Kevrekidis-2013-12}. For time harmonic solutions $u_n(t)=e^{-\ii \mu t}\phi_n$ and after the rescaling $\psi_n:=\eps^{-1/2}\phi_n$ (with $\eps>0$),  we get eigenvalue problem \eqref{nl.ev} with $A$ in \eqref{A:E.DS} and the nonlinearity $f_n(\psi)=|\psi_n|^2\psi_n$. Example \ref{ex.A.DSchr} and Section \ref{subsubsec:DSchr.sp} explain that for $|\gamma|$ small enough Assumption \ref{ass:Af} holds with $\H=\X:=\C^{2N}$. Note that the Lipschitz continuity of $f$ holds, \eg~with $\sum_{j=1}^{2N}\left||\psi_n|^2\psi_n-|\phi_n|^2\phi_n\right|^2\leq c\max_{j=1,\dots,2N}\{|\psi_n|^4,|\phi_n|^4\}\sum_{j=1}^{2N}|\psi_n-\phi_n|^2$.

Assumption \ref{ass:A.C} is satisfied with $(\cC \psi)_n = \overline{\psi_{-n}}$ (\ie~the discrete $\PT$-symmetry) due to the choice of the ``potential'' $V_n:=\ii \gamma (-1)^n$, such that $V_{-n}=\overline{V_n}$. Our results therefore recover Theorem 1 in \cite{Kevrekidis-2013-12}.

\subsection{Optics: coupled mode equations}
\label{subsec:CME}

In Kerr-nonlinear optical fibers with a Bragg grating and a localized defect the propagation of asymptotically broad wavepackets can be described by the system of ``coupled mode equations''
\begin{equation}
\begin{aligned}
\ii (\partial_t E_1+\partial_x E_1) + \kappa(x) E_2+V(x) E_2 + (|E_1|^2+2|E_2|^2)E_1&=0\\
\ii (\partial_t E_2-\partial_x E_2) + \kappa(x) E_1+V(x) E_1 + (|E_2|^2+2|E_1|^2)E_2&=0
\end{aligned}
\end{equation}
with $\kappa(x)\to \kappa_\infty>0$ and $V(x)\to 0$ as $|x|\to \infty$, see \cite{Goodman-2002-19}. The potentials $\kappa(x)-\kappa_\infty$ and $V(x)$ describe the defect of the material and are determined by the refractive index. Once again, we consider the time harmonic ansatz $E(x,t)=e^{-\ii \mu t}\phi(x)$ and after the rescaling $\psi:=\eps^{-1/2}\phi$ (with $\eps>0$), we obtain eigenvalue problem \eqref{nl.ev} with $A$ in \eqref{A:E.Dirac} and 
$$
f(\psi)=
\begin{pmatrix} 
(|\psi_1|^2+2|\psi_2|^2)\psi_1\\
(|\psi_2|^2+2|\psi_1|^2)\psi_2
\end{pmatrix}.
$$
As mentioned in Section \ref{subsubsec:Dirac.sp}, real smooth and bounded potentials $\kappa$ and $V$ exist such that $A$ has a simple isolated eigenvalue. Examples \ref{ex.A.CME} and 
\ref{ex.f.pol} (see the remark on the vector case at the end of Sec. \ref{ex.f.pol}), guarantee that Assumption 
\ref{ass:Af} is satisfied with $\H=L^2(\R)\times L^2(\R)$ and $\X=H^1(\R)\times H^1(\R)$ provided $V,\kappa \in L^\infty$ and $\kappa(x)\to \kappa_\infty>0$ and $V(x)\to 0$ as $|x|\to \infty$.

For materials with loss/gain the potentials $V$ and $\kappa$ become complex and choosing them $\PT$-symmetric, we satisfy also Assumption \ref{ass:A.C}. The existence of a real simple isolated eigenvalue $\mu_0$ of $A$ is guaranteed at least for small imaginary parts of $\kappa$ and $V$ by the analytic dependence as in Remark \ref{rem:sim+real}.
Hence, (in the language of \cite{Goodman-2002-19}), our results show that conservative nonlinear defect modes bifurcate from linear ones in the $\PT$-symmetric case.

\subsection{Superconductivity}

A model of a finite superconducting wire is discussed in \cite{Rubinstein-2007-99,Rubinstein-2010-195} and the bifurcation of nonlinear states for a nonlinear parabolic equation ($d=1$) on $\Omega=(-1,1)$ is studied. In detail, the problem
\begin{equation}\label{E:heat_Rubinstein.eq}
\begin{aligned}
w_t &= w_{xx} + \ii x Iw+\Gamma w +N[w], \quad x\in (-1,1), \qquad w(-1)=w(1)=0,\\
N[w] &= -|w|^2w + \ii w\int_0^x \Im \left( w(s,t)\overline{w}_x(s,t)\right) \dd s,
\end{aligned}
\end{equation}
where $I$ and $\Gamma$ are real parameters, is considered. In \cite[Sec.6]{Rubinstein-2010-195} the authors study the bifurcation of nonlinear (generally $t$-dependent) solutions from the zero solution at the smallest eigenvalue $\lambda_1\in \R$  of 
\begin{equation}\label{A.Rub}
A:=-\partial_x^2 -\ii xI, \quad \Dom(A) := H^2((-1,1)) \cap H^1_0((-1,1)).
\end{equation}
The potential $-\ii x I$ is $\PT$-symmetric, so the spectrum of $A$ remains real if the parameter $I$ is chosen small enough, \cf~Remark \ref{rem:sim+real}, and the number of non-real eigenvalues remains finite for any $I \in \R$. For the bifurcation problem the authors set $\Gamma = \text{Re} \lambda_1+\eps, 0<\eps \ll 1$ and use the center manifold reduction, where the center manifold is one dimensional and corresponds to the zero eigenvalue of $A-\text{Re} \lambda_1$. On the manifold they study $t$-dependent, but also stationary nonlinear solutions. The asymptotics of the latter are given by 
$$w(x) \sim \eps^{1/2}\alpha u_1(x),$$
where $u_1$ is the linear eigenfunction corresponding to $\lambda_1$ and $\alpha\in \R$ is the projection coefficient on the center subspace and solves an algebraic equation. For $I$ small enough $\lambda_1\in \R$, such that a real nonlinear eigenvalue $\Gamma$ bifurcates. The eigenfunction $w$ is $\PT$-symmetric due to the $\PT$-invariance of the center manifold.

In the formal part of \cite{Rubinstein-2010-195} the more detailed expansion
$$w(x) \sim \eps^{1/2}\alpha u_1(x) +\eps^{3/2}w_1(x),$$
is given, where  the correction $w_1$ solves
\begin{equation}\label{E:w1_eq}
(A-\lambda_1)w_1 = \alpha u_1+N[\alpha u_1].
\end{equation}
$\alpha \in \R$ can then be selected via the solvability condition of the above equation and agrees to leading order with the $\alpha$ from the center manifold approach.

To relate this work to our results, we rescale the $t$-independent solution $w(x)=\eps^{1/2}\psi(x)$ and recover from \eqref{E:heat_Rubinstein.eq} a problem of type \eqref{nl.ev}, namely
\begin{equation}
(A-\Gamma)\psi -\eps \big(f_{\rm c}(\psi) + f_{\rm N}(\psi) \big)=0,
\end{equation}
\cf~Examples \ref{ex.f.pol}, \ref{ex.f.Rubinstein}. Equation \eqref{E:w1_eq} thus corresponds to our 
\eqref{E:phi_thm}. Observe that $A$ fits into the setting of Example \ref{ex.A.sing} with $V_1(x)=\ii x I$, $v_2=0$ and $\X=\H_0^{1, \rm Dir}((-1,1)) = H_0^1((-1,1))$. The nonlinearities are discussed in Examples \ref{ex.f.pol}, \ref{ex.f.Rubinstein} and \ref{ex.PT.sym} and it shown that $H^1$ is a suitable space for the Lipschitz condition \eqref{E:Lip}. Hence, our results in Theorem \ref{thm.NL-fixpt}, Corollary \ref{cor.hom.f} and Theorem \ref{thm.sym} are applicable and provide real nonlinear eigenvalues $\mu$ with $\PT$-symmetric nonlinear solutions $\psi$.

\section{Numerical Examples}
\label{sec:num}
We analyze numerically two nonlinear problems of type \eqref{nl.ev}, both with the Schr\"odinger operator $A = -\Delta +V$ in $L^2(\R^2)$, \cf~Example \ref{ex.A.Schr}, and the cubic nonlinearity $f_{\rm c}$, \cf~Example \ref{ex.f.pol}, \ie~
\begin{equation}\label{nl.ev_schr}
(-\Delta +V) \psi - \eps |\psi|^2\psi = \mu \psi, \qquad \|\psi\|_{L^2}=1.
\end{equation}
Selected potentials $V$ posses some antilinear or linear symmetries. Clearly, the nonlinearity is highly symmetric and satisfies Assumption \ref{ass:A.C}.\eqref{ass:A.C.f} with $\cC=\PT$ as well as $\cC=\P_jT$, $j=1,2$, and also Assumption \ref{ass:A.S}.\eqref{ass.f.S} with any coordinate reflection symmetry $\cS$.

Our choice of $d=2$ rather than the numerically simpler $d=1$ allows the investigation of partial $\PT$-symmetries as well as the interplay between antlinear and linear symmetries in a single problem.

The numerics are performed using the package \texttt{pde2path}
\cite{p2p,p2p2,p2p2b} for numerical continuation and bifurcation 
in nonlinear elliptic systems of PDEs. The package uses linear finite
elements for the discretization, Newton's iteration for the
computation of nonlinear solutions and an arclength continuation of solution branches. 
In all numerical computations, the free complex phase of the solution was fixed by forcing $\Im(\psi(x_0))=0$ at a selected point $x_0$ within the computational domain. For the plots, we select $x_0=(0,0)$ for the $\PT$-symmetric case in Example \ref{ex:num_PT} and $x_0=(0,2)$ for the $\P_1\T$-symmetric Example \ref{ex:num_PT1}. In all computations, except for one case mentioned below, the numerical grid is selected symmetric with respect to both coordinate axes as well as with respect to the reflection $x\to -x$. This is crucial for recovering symmetries of eigenfunctions and realness of eigenvalues.

\begin{example}\label{ex:num_PT} 
We consider first the following imaginary perturbation of the harmonic oscillator that is compatible with Example \ref{ex.A.Schr} and inspired by the Bose-Einstein condensates models from Section \ref{subsec:BEC},
\begin{equation}\label{E:SHO6_PT}
V(x_1,x_2) = \frac{1}{2}(x_1^2+x_2^2) + \ii \gamma x_1 \frac{2}{x_1^2+x_2^2+2}.
\end{equation}
Clearly, $V(-x_1,-x_2)=\overline{V(x_1,x_2)}=V(-x_1,x_2)$. Hence, the problem has three symmetries: two antilinear symmetries, namely the $\PT$ symmetry and the $\P_{1}\T$ symmetry, and the linear $\P_2$ symmetry, \cf~Example \ref{ex.PT.sym}.

For $\gamma=0$, the eigenvalues of $A$ are known explicitly:
$$\lambda_k = \sqrt{2}k, \ k =1,2, \dots, \text{ where the multiplicity of } \lambda_k \text{ is } k.$$ 
Enumerating the eigenvalues including their multiplicity, we obtain our eigenvalues $\mu_n$ for $\eps=\gamma=0$.

For the discretization of the PDE, we take $2*80^2=12800$ isosceles right triangles of equal size generated by Matlab's command \texttt{poimesh} on the domain $x\in [-8,8]^2$ with homogeneous Dirichlet boundary conditions. The first four eigenfunctions are well localized within the selected domain.

For $\gamma=2$, the numerically obtained first four eigenvalues (for $\varepsilon=0$) are
$$\mu_1\approx 2.096, \ \mu_2\approx 2.583, \ \mu_3 \approx 3.155, \ \mu_4\approx 4.256,$$
and they are all simple.  
\begin{figure}[ht!]
\label{Fig:SHO6_BD_g}
\includegraphics[width=0.4 \textwidth]{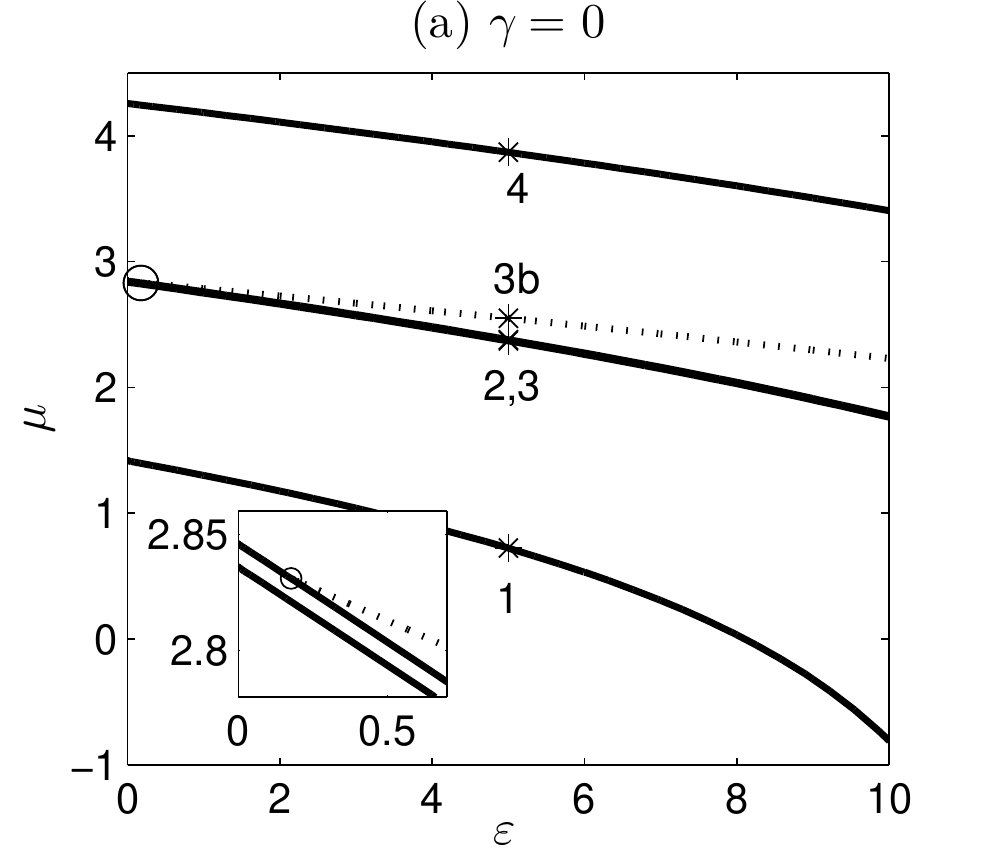}
\includegraphics[width=0.4 \textwidth]{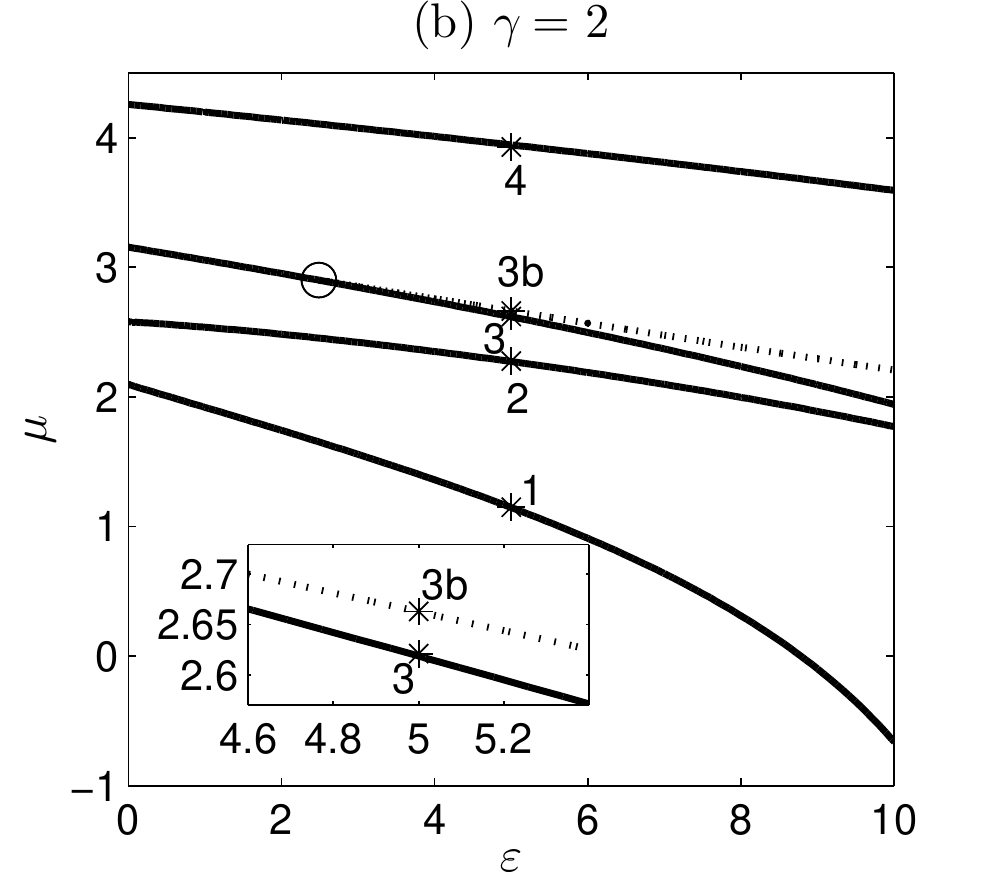}
\caption{Bifurcation diagram of \protect\eqref{nl.ev_schr}, \protect\eqref{E:SHO6_PT} in the parameter $\varepsilon$ for the first eigenvalues $\mu_1, \dots, \mu_4$ with $\gamma=0$ in (a) and with $\gamma=2$ in (b). Circles label secondary bifurcation points.}
\end{figure}

In Fig.~\ref{Fig:SHO6_BD_g}, the bifurcation diagram in $\varepsilon$ for $\mu_1,\mu_2,\mu_3$ and $\mu_4$ is plotted. The eigenvalues are continued in $\varepsilon>0$ from the linear eigenvalues at $\varepsilon=0$ for two values of $\gamma$, namely $\gamma=0$ in (a) and $\gamma=2$ in (b). In both cases, all plotted eigenvalues (including the ones bifurcating from $\mu_3$) are real. Note that for $\gamma=0$ the grid is symmetric only with respect to the reflection $x\to -x$ and not with respect to the coordinate axes. This suppresses the multiplicity of the first four linear eigenvalues such that these can be easily numerically continued in $\varepsilon$. The numerics suggest that all the four simple eigenvalues remain real for at least $\varepsilon \leq 10$. Clearly, the numerics agree with the analysis as simple eigenvalues stay real for $\varepsilon$ small. 
\begin{figure}[h!]
\includegraphics[scale=0.5]{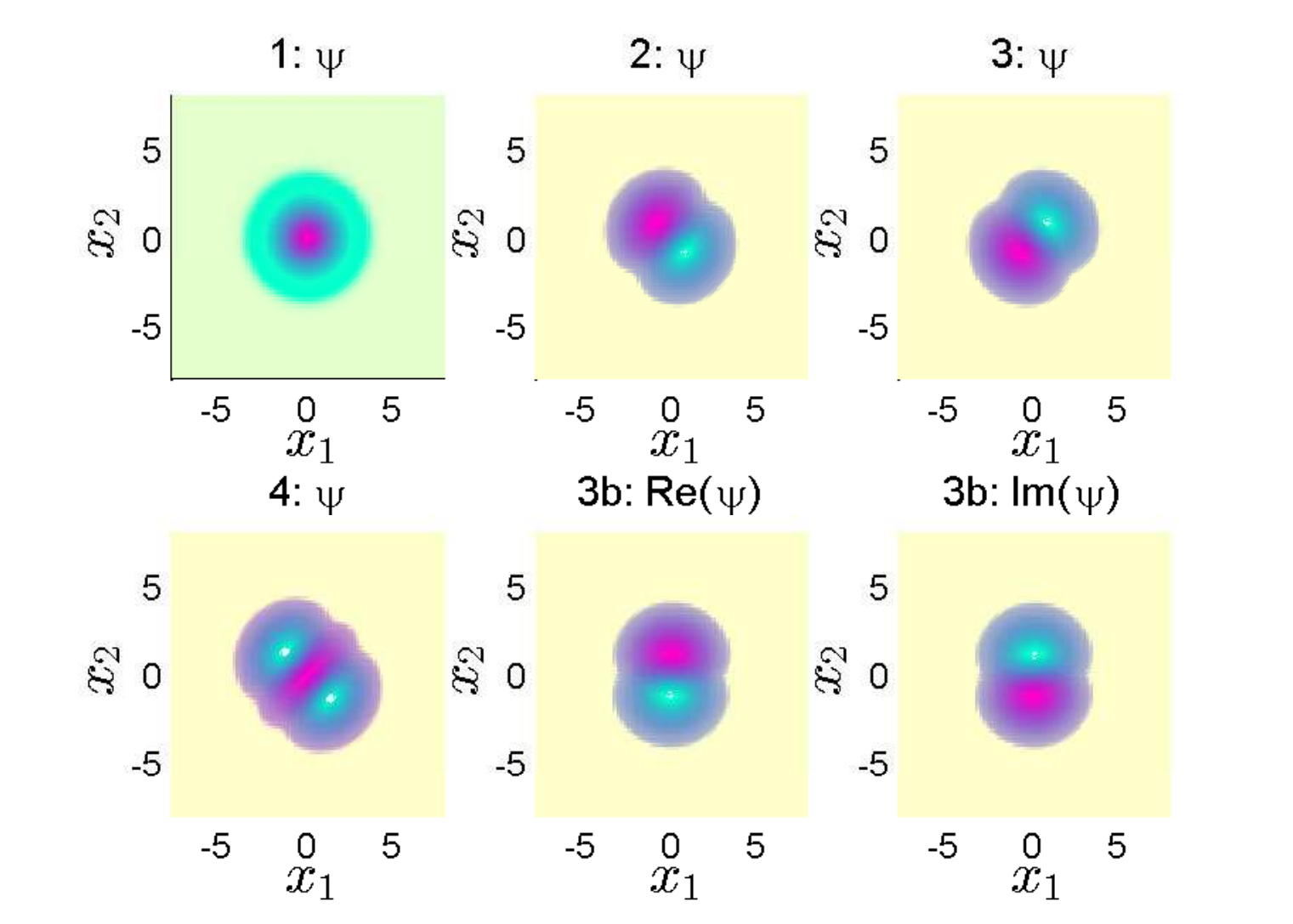}
\caption{Profiles of the nonlinear eigenfunctions of \protect\eqref{nl.ev_schr}, \protect\eqref{E:SHO6_PT} at $\gamma=0$ labeled by 1--4, 3b in Fig.~\protect\ref{Fig:SHO6_BD_g} (a). }
\label{Fig:SHO6_PT_profs_gam0}
\end{figure}
In Fig.~\ref{Fig:SHO6_PT_profs_gam0}, the eigenfunctions at the five labeled points at $\varepsilon=5$ in Fig.~\ref{Fig:SHO6_BD_g} (a) for $\gamma=0$ are plotted. Since $V$ is real, the eigenfunctions can be automatically chosen $\PT$-symmetric (using a proper rotation of the complex phase). Note also that after a proper rotation in the $(x_1,x_2)-$plane (allowed due to the rotation symmetry of $V$ at $\gamma=0$) all eigenfunctions 1--4 and 3b are symmetric or antisymmetric with respect to $x_1\to -x_1$ as well as $x_2\to -x_2$. The numerically generated profiles for the eigenfunctions 1--4 are symmetric about other axes due to the lack of coordinate symmetry of the grid, as explained above.

In Fig.~\ref{Fig:SHO6_PT_profs_gam2}, the profiles for the case $\gamma=2$ from Fig.~\ref{Fig:SHO6_BD_g} (b) appear. Eigenfunctions 1 and 2 satisfy all the three symmetries, \ie~$\PT$, $\P_1\T$ as well as the linear  $\psi(x_1,-x_2)=\pm\psi(x_1,x_2)$. Eigenfunctions 3 and 4 satisfy the linear (anti)symmetry and can be chosen either $\PT$- or $\P_1\T$-symmetric after a suitable multiplication by $\ii$.  The eigenfunction 3b on the dotted branch (bifurcating from the primary branch) is only $\P_1\T$-symmetric.
\begin{figure}
\includegraphics[scale=0.5]{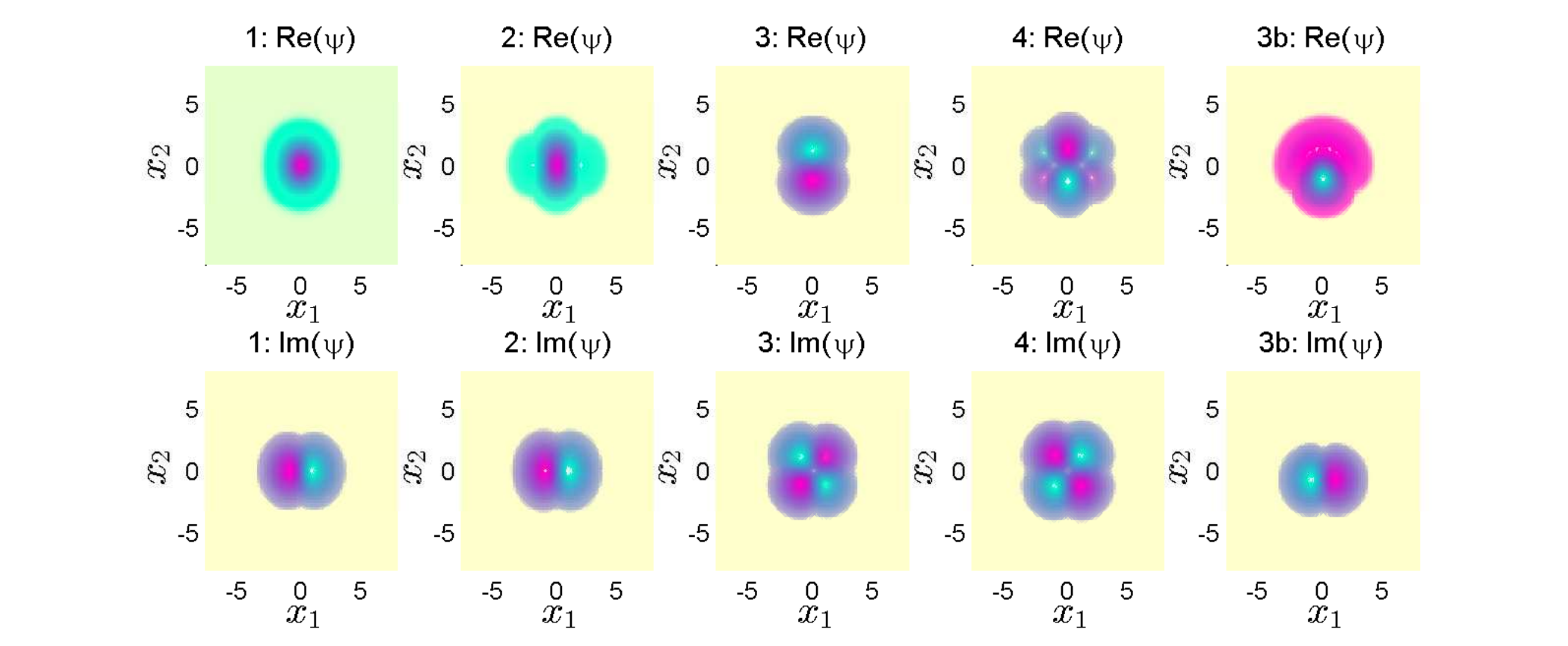}
\caption{Profiles of the nonlinear eigenfunctions of \protect\eqref{nl.ev_schr}, \protect\eqref{E:SHO6_PT} at $\gamma=2$ labeled by 1--4, 3b in Fig.~\protect\ref{Fig:SHO6_BD_g} (b). }
\label{Fig:SHO6_PT_profs_gam2}
\end{figure}

In Fig.~\ref{Fig:SHO6_BD_gam}, we perform continuation in the parameter $\gamma$ for the two values $\varepsilon=0$ and $\varepsilon=2$. The results are qualitatively similar to those in 1D from  \cite{Dast-2013-46}. When two real eigenvalues collide, they leave the real axis and become a complex conjugate pair. In addition, however, a secondary bifurcation can occur, like, \eg, from $\mu_3$ at $\varepsilon\approx 1.5$ (see the inset in Fig.~\ref{Fig:SHO6_BD_gam} (c)).

Note that for a complex conjugate pair of simple eigenvalues the corresponding eigenfunctions $\psi_1,\psi_2$ can be chosen to be related by $\psi_2(x) =(\PT\psi_1)(x)=\alpha (\P_1\T\psi_1)(x)$ with a suitable factor $\alpha\in \C,|\alpha|=1$. Hence, below we always plot only one eigenfunction for a complex conjugate pair.
\begin{figure}[h!] 
\includegraphics[width=0.4 \textwidth]{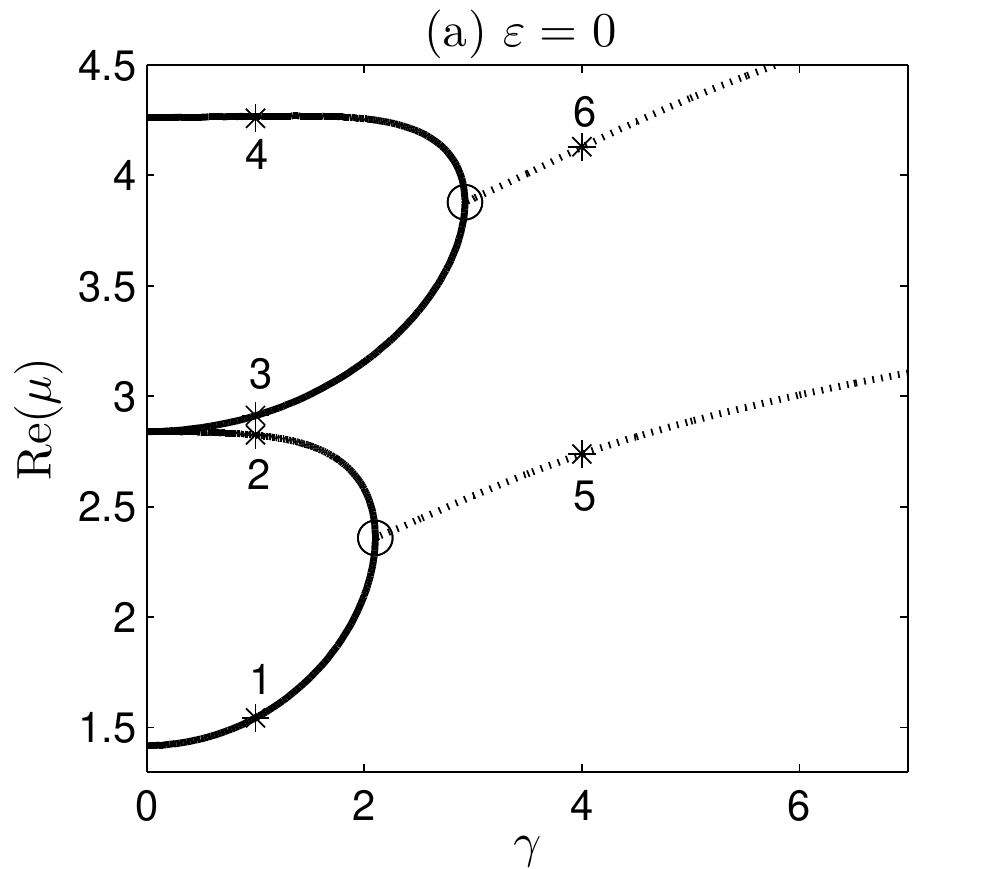}
\includegraphics[width=0.4 \textwidth]{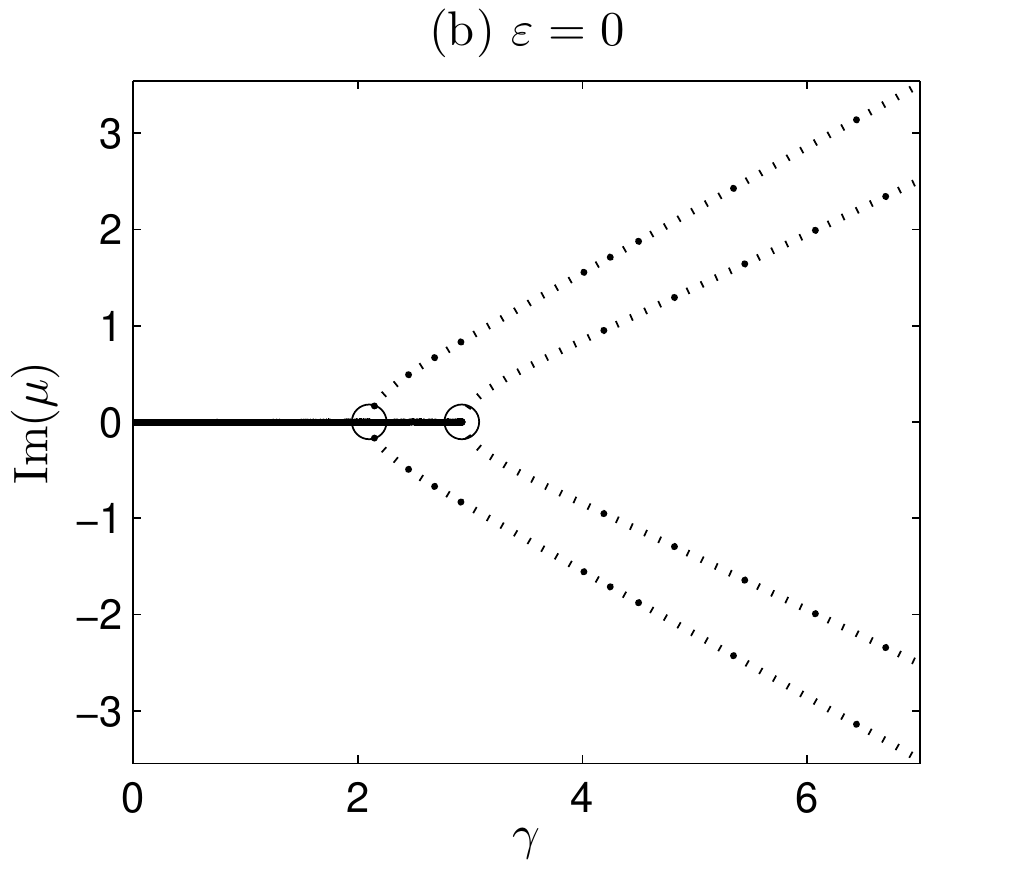}

\includegraphics[width=0.4 \textwidth]{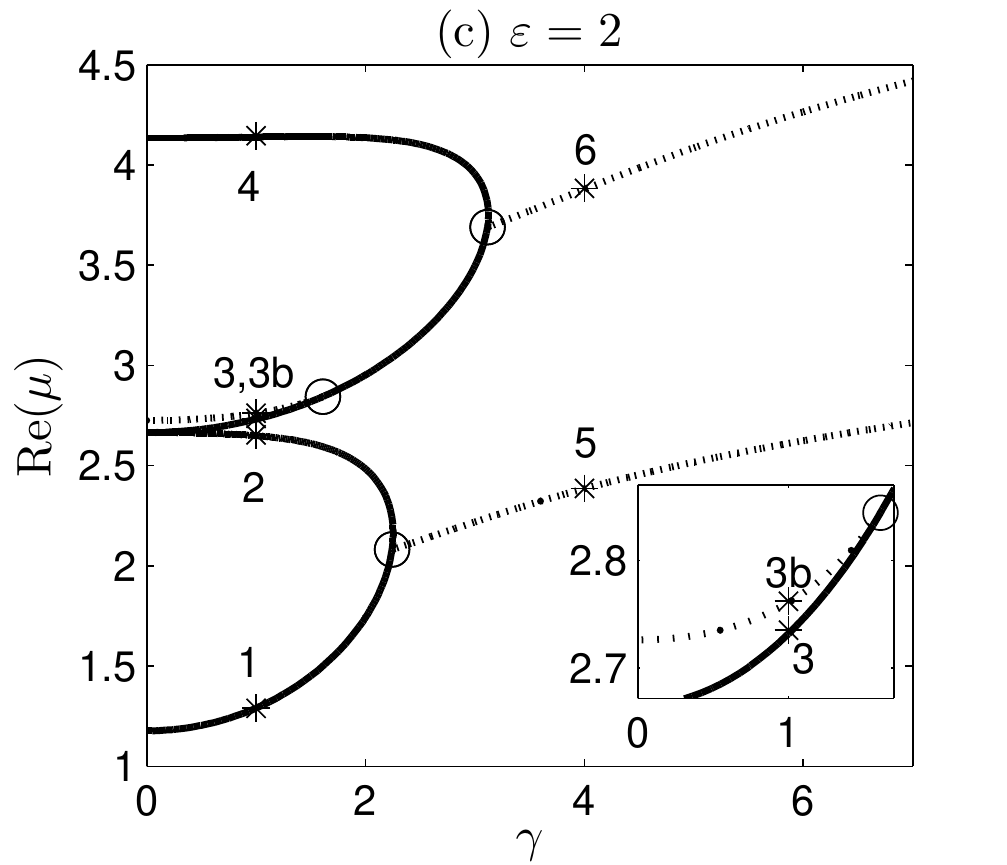}
\includegraphics[width=0.4 \textwidth]{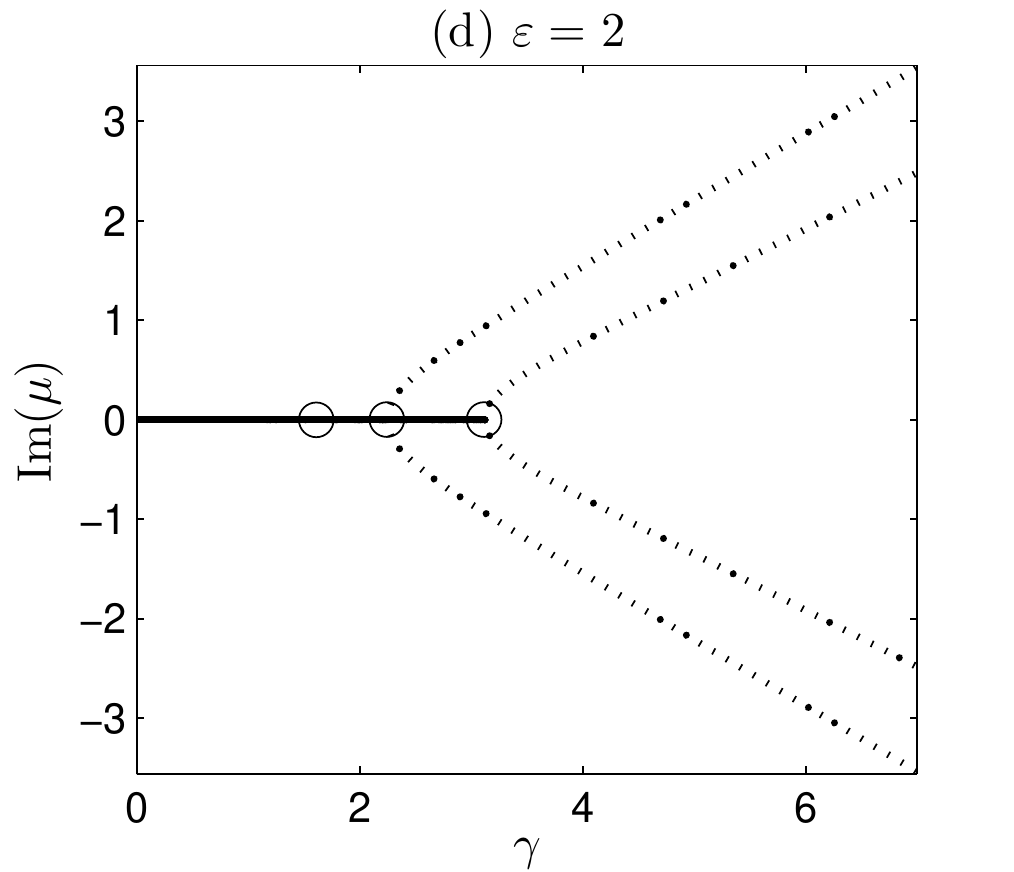}
\caption{Bifurcation diagram of \protect\eqref{nl.ev_schr}, \protect\eqref{E:SHO6_PT} in the parameter $\gamma$ for the first eigenvalues $\mu_1, \ldots, \mu_4$ with $\eps=0$ in (a) and (b) and with $\varepsilon=2$ in (c) and (d). Circles label secondary bifurcation points.}
\label{Fig:SHO6_BD_gam}
\end{figure}

The eigenfunctions for the linear case $\varepsilon=0$ at the six labeled points in Fig.~\ref{Fig:SHO6_BD_gam}(a) are shown in Fig.~\ref{Fig:SHO6_PT_profs_g0}. The symmetry properties of the eigenfunctions 1--4 are the same as for the case $\gamma=2, \varepsilon=5$ in Fig.~\ref{Fig:SHO6_PT_profs_gam2}. The eigenfunctions 5 and 6 corresponding to the complex eigenvalues are neither $\PT$- nor $\P_1\T$-symmetric but the linear symmetry $\psi(x_1,-x_2)=\pm\psi(x_1,x_2)$ is preserved. This is in agreement with Lemmas \ref{lem:A.C} and \ref{lem:A.S}.
\begin{figure}[ht!]
\includegraphics[scale=0.38]{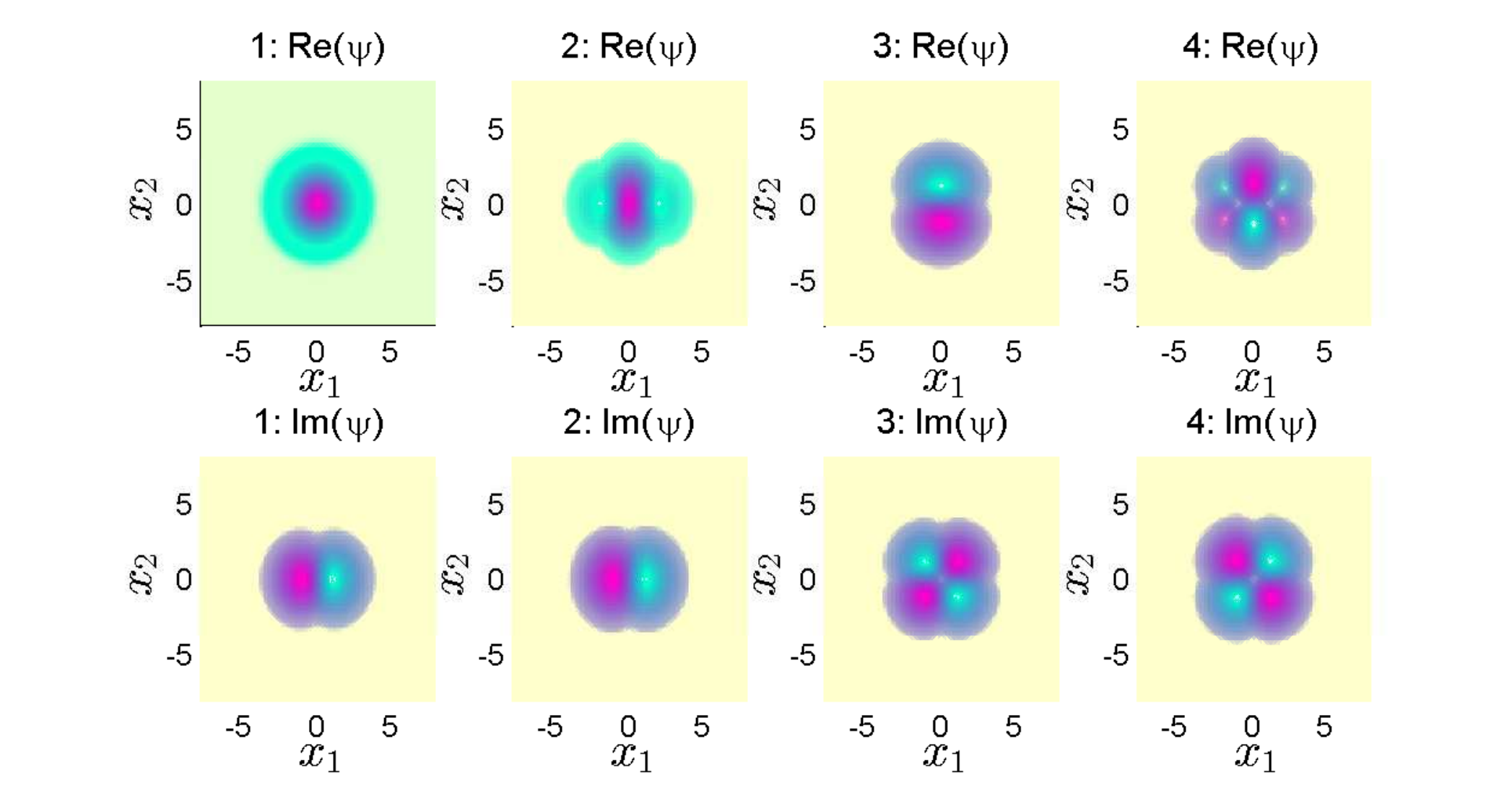}\hspace{-.9cm}
\includegraphics[scale=0.38]{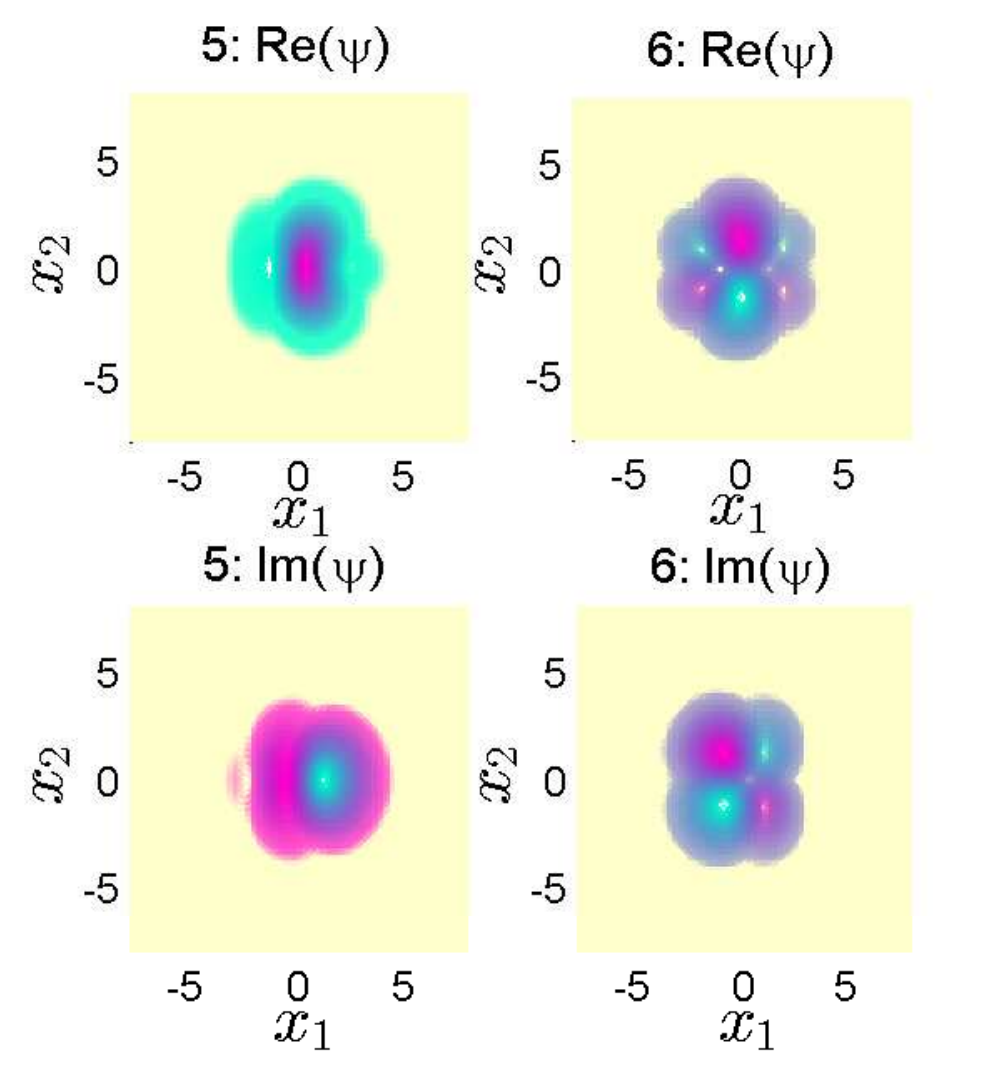}
\caption{Profiles of the nonlinear eigenfunctions of \protect\eqref{nl.ev_schr}, \protect\eqref{E:SHO6_PT}  at $\varepsilon=0, \gamma=1$ labeled by 1--6 in Fig.~\protect\ref{Fig:SHO6_BD_gam} (a). }
\label{Fig:SHO6_PT_profs_g0}
\end{figure}

For the nonlinear case $\varepsilon=2$, the eigenfunctions are in Fig.~\ref{Fig:SHO6_PT_profs_g2}. The symmetries of the eigenfunctions 1--4 are again the same as for the case $\gamma=2, \varepsilon=5 $ in Fig.~\ref{Fig:SHO6_PT_profs_gam2}. The eigenfunctions 5 ad 6 of the complex eigenvalues are, once again, only linearly symmetric. On the other hand, the eigenfunction 3b on the secondary bifurcation branch (corresponding to a real eigenvalue) has only the $\P_1\T$ symmetry. 
\begin{figure}[ht!]
\includegraphics[scale=0.43]{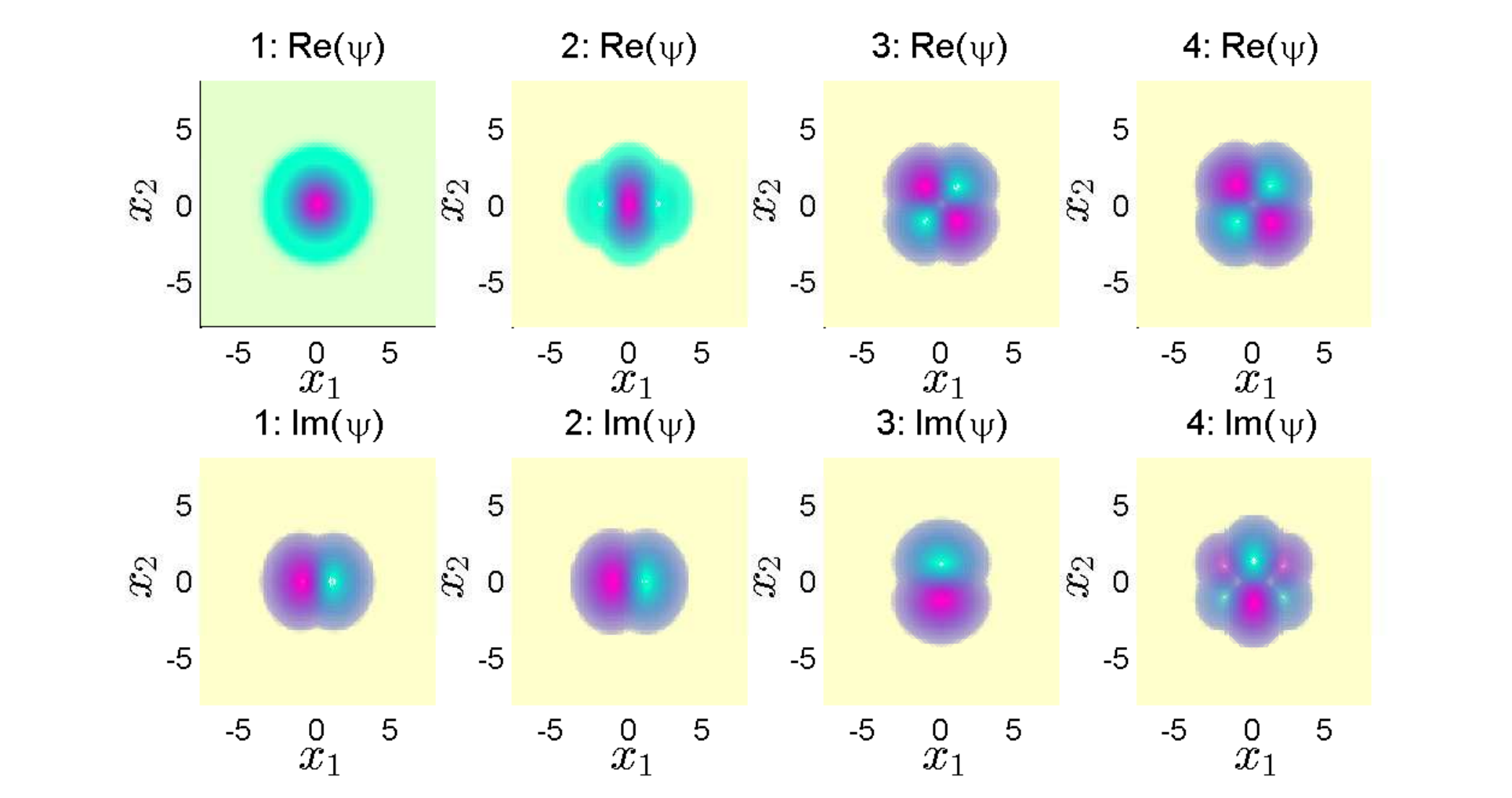}

\includegraphics[scale=0.43]{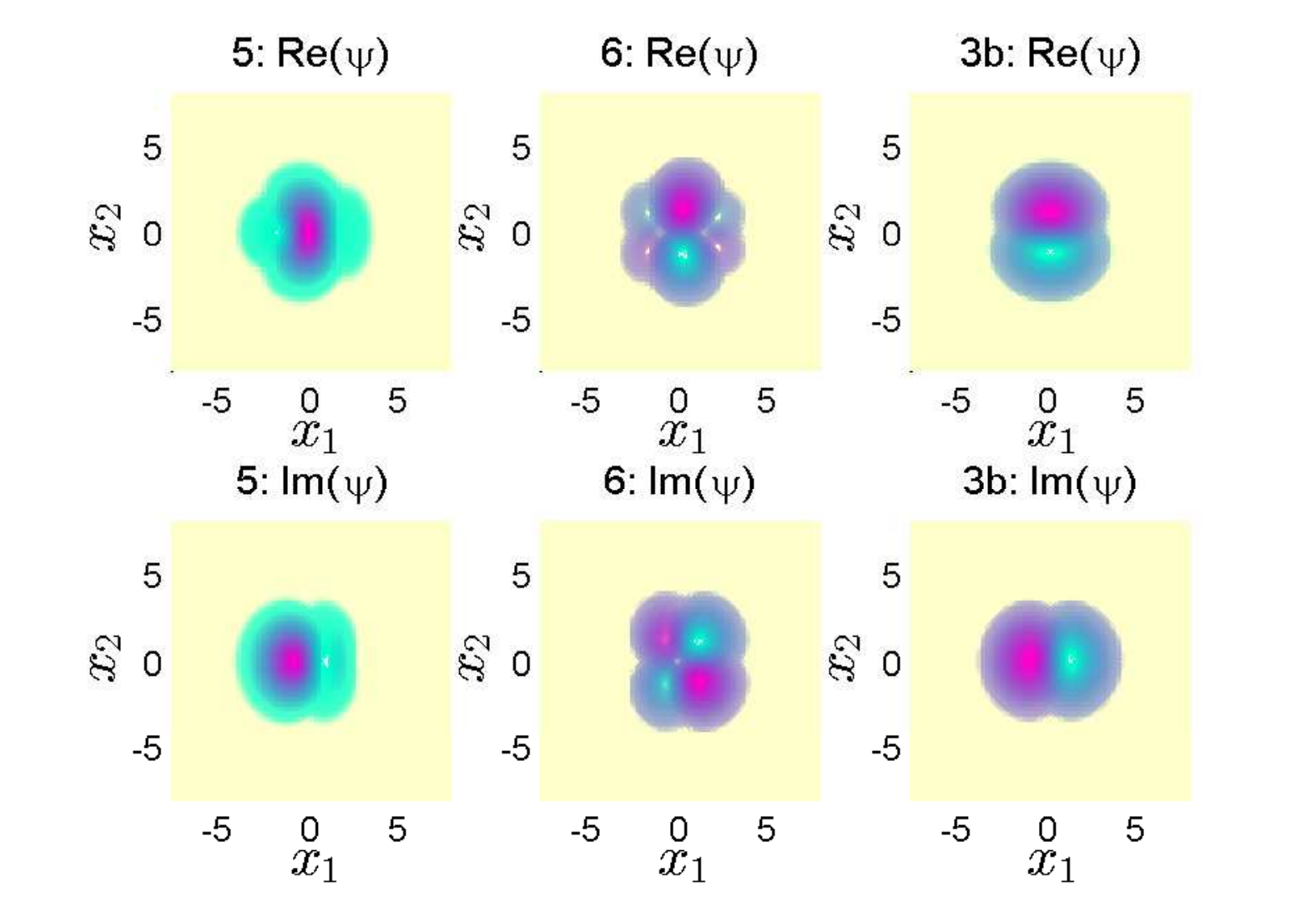}
\caption{Profiles of the nonlinear eigenfunctions of \protect\eqref{nl.ev_schr}, \protect\eqref{E:SHO6_PT}  at $\varepsilon=2, \gamma=1$ labeled by 1--6, 3b in Fig.~\protect\ref{Fig:SHO6_BD_gam} (c). }
\label{Fig:SHO6_PT_profs_g2}
\end{figure}
\end{example}
\begin{example}\label{ex:num_PT1} 
As the second example, we choose the $\P_1\T$-symmetric potential \eqref{E:Gauss9_PT} with $a=\tfrac{3}{2}$ and $v_0=1$, appearing in optics literature. Clearly, $V$ satisfies $V(-x_1,x_2)=\overline{V(x_1,x_2)}$, but no other obvious antilinear or linear symmetry involving  reflections of coordinates. Simplicity of eigenvalues of $A$ is discussed in Section \ref{subsec:optics}; the four lowest eigenvalues appear to be simple numerically for $\gamma=0$.

Our discretization mesh is given by $14400$ isosceles right triangles on the domain $x\in [-13,13]^2$ chosen such that the mesh is symmetric about the coordinate axes and with respect to the reflection $x\to -x$. We again use homogeneous Dirichlet boundary conditions.

In Fig.~\ref{Fig:Gauss9_BD_g} we plot the bifurcation diagram in the parameter $\eps$ for $\gamma=0$ and $\gamma=0.1$. Similarly to Example \ref{ex:num_PT}, for $\gamma=0$ (real potential $V$),  the nonlinear eigenvalues stay real even after secondary bifurcations, while for $\gamma=0.1$ secondary bifurcations result in complex conjugate pairs of nonlinear eigenvalues. 
\begin{figure}[ht!]
\includegraphics[width=0.4 \textwidth]{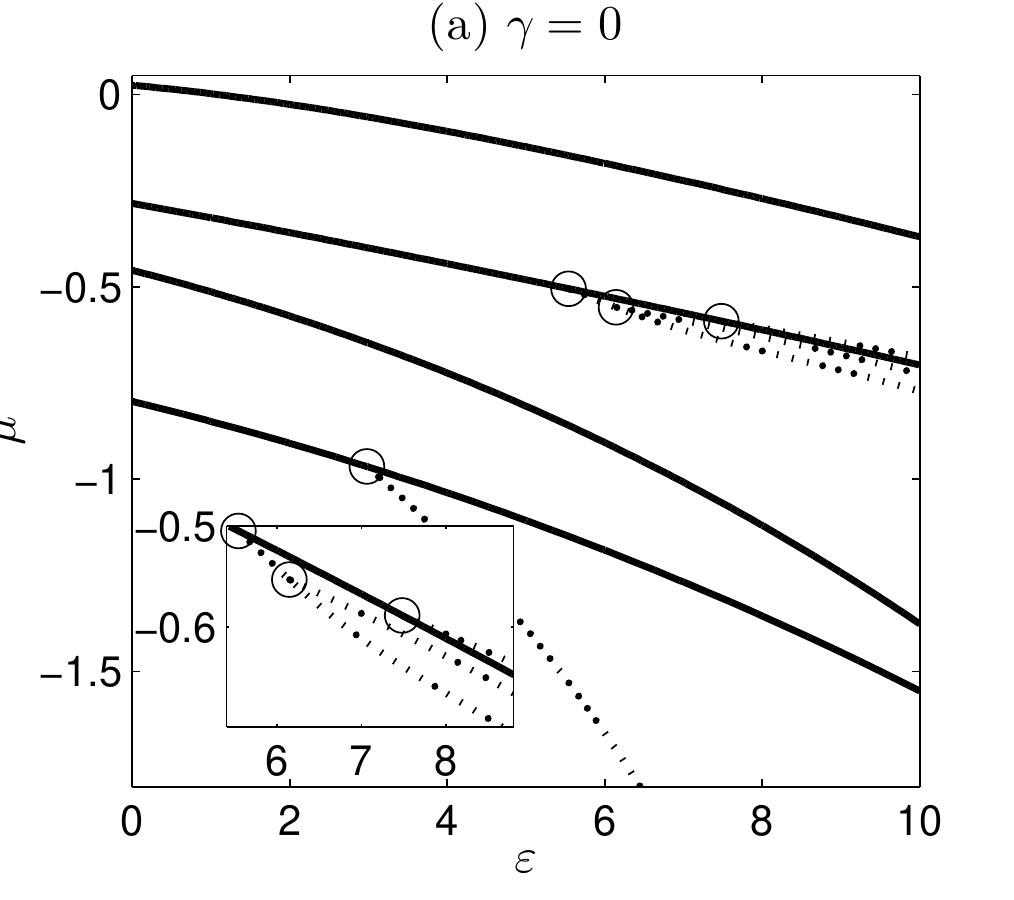}

\includegraphics[width=0.4 \textwidth]{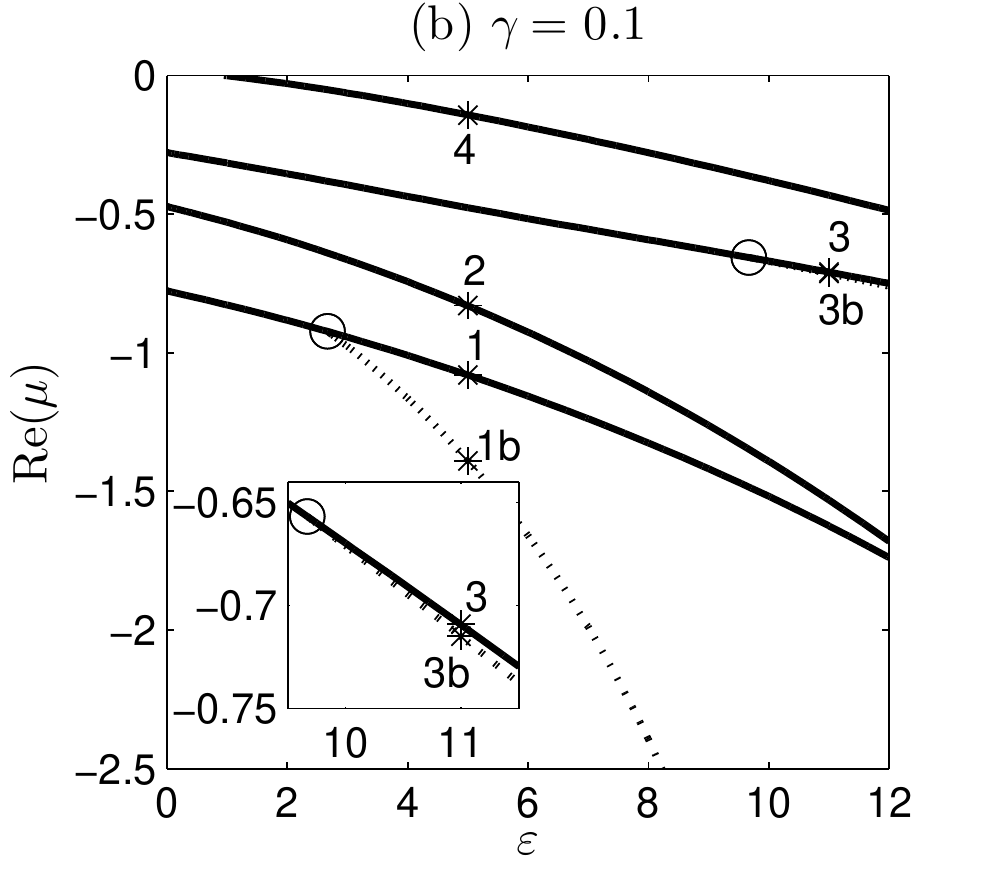}
\includegraphics[width=0.4 \textwidth]{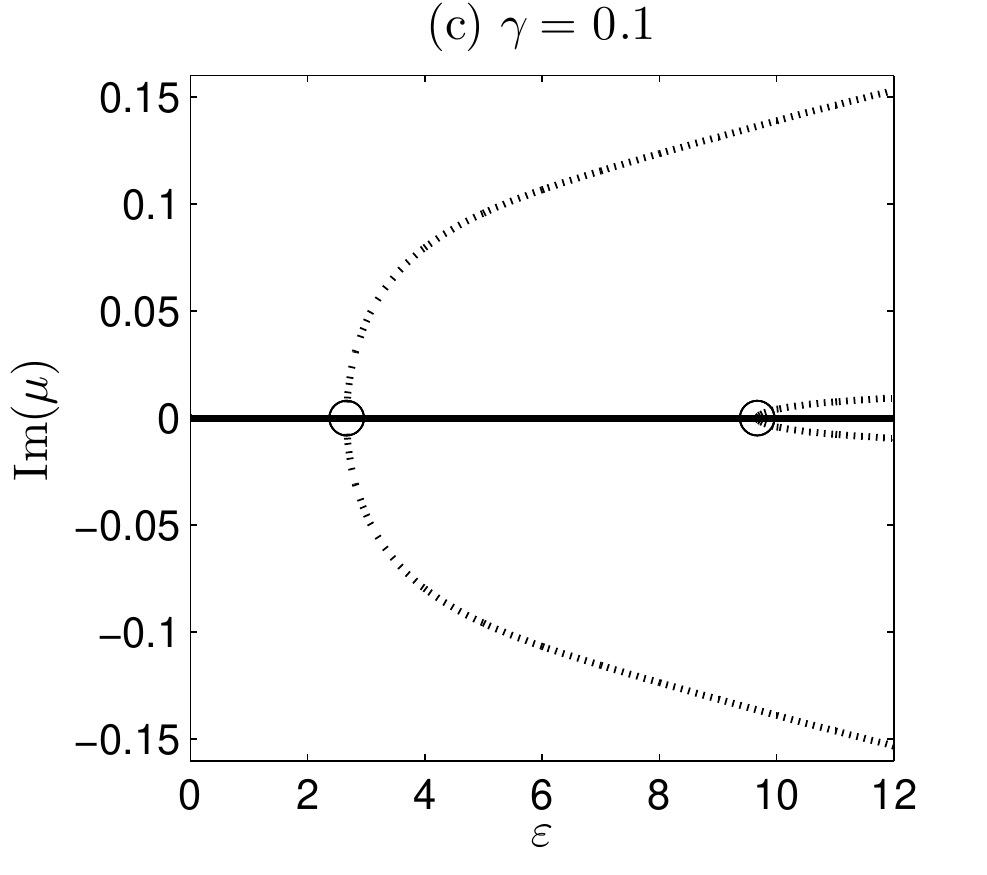}
\caption{Bifurcation diagram in the parameter $\varepsilon$ for the first eigenvalues $\mu_1, \dots, \mu_4$ of \protect\eqref{nl.ev_schr}, \protect\eqref{E:Gauss9_PT} with $\gamma=0,v_0=1$ in (a) and with $\gamma=0.1,v_0=1$ in (b) and (c). Circles label secondary bifurcation points.}
\label{Fig:Gauss9_BD_g}
\end{figure}

In Fig.~\ref{Fig:Gauss9_PT_prof_gam_p1}, we plot the six eigenfunctions labeled in Fig.~\ref{Fig:Gauss9_BD_g} (b) for the case $\gamma=0.1$. Clearly, all the four eigenfunctions on the primary branches (labels 1--4) with real eigenvalues are $\P_1\T$-symmetric. The bifurcating solutions (labels 1b and 3b) are asymmetric. Once again, for complex conjugate pairs we plot only one eigenfunction as the two can be chosen to be related by $\psi_2(x)=(\P_1\T\psi_1)(x)$.
\begin{figure}[ht!]
\includegraphics[scale=0.4]{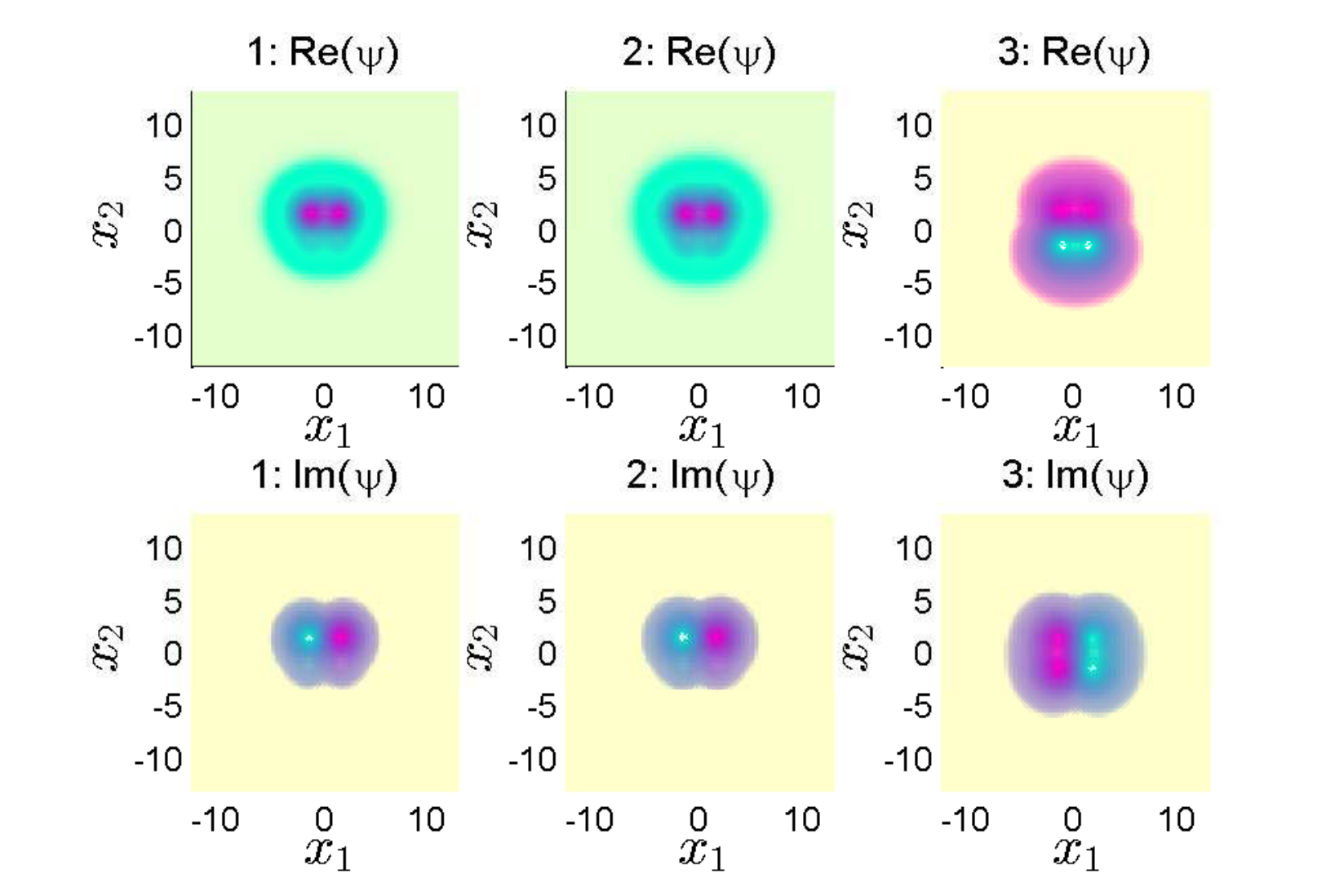}\hspace{-1cm}
\includegraphics[scale=0.4]{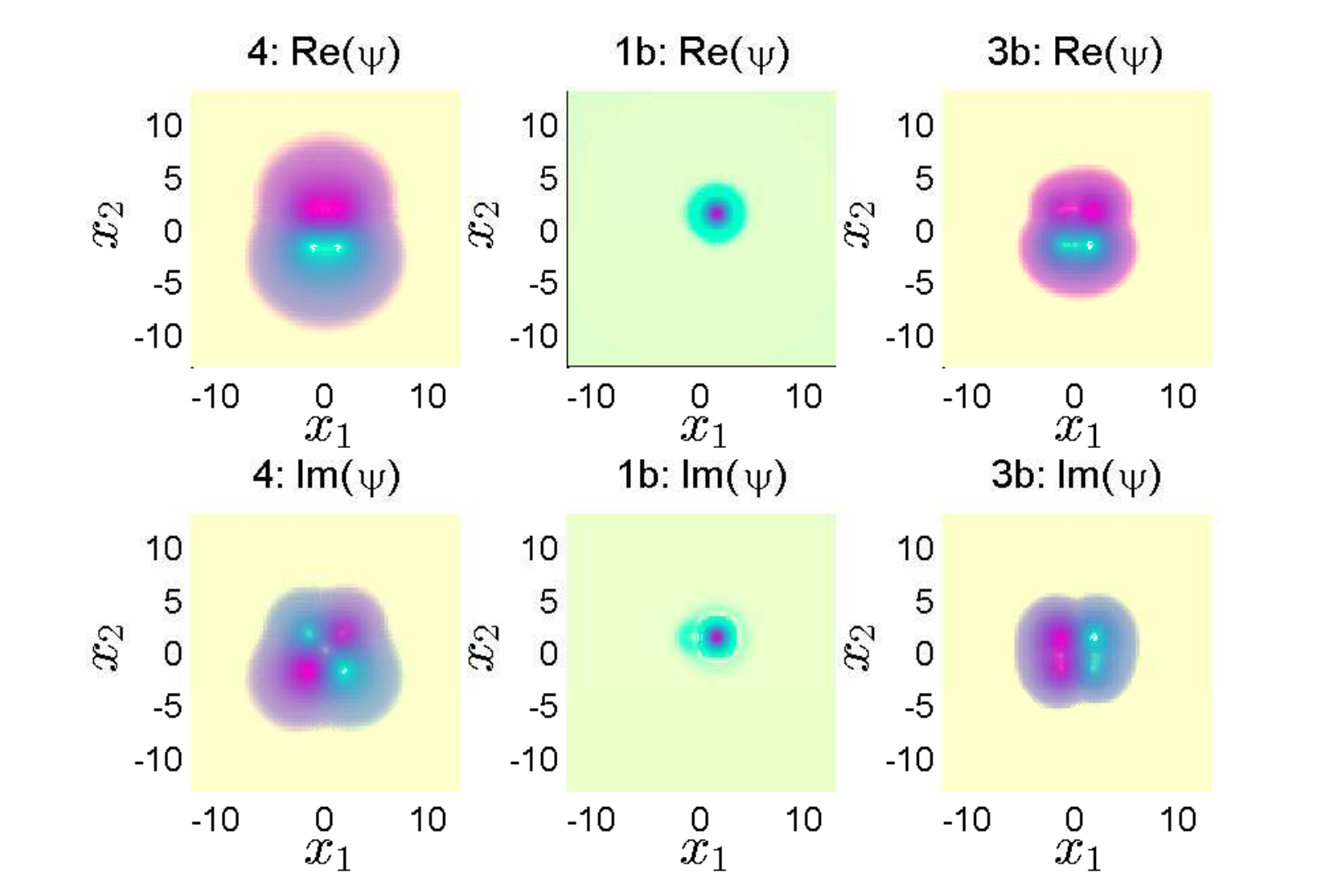}
\caption{Profiles of the nonlinear eigenfunctions of \protect\eqref{nl.ev_schr}, \protect\eqref{E:Gauss9_PT} at $\gamma=0.1,v_0=1$ labeled by 1--4, 1b and 3b in Fig.~\protect\ref{Fig:Gauss9_BD_g} (b). }
\label{Fig:Gauss9_PT_prof_gam_p1}
\end{figure}

The bifurcation diagram in the parameter $\gamma$ is plotted in Fig.~\ref{Fig:Gauss9_BD_gam} for the two values $\varepsilon=0$ and $\varepsilon=2$ and the results are, again, analogous to Example \ref{ex:num_PT}. The first collision of eigenvalues occurs at $\gamma =\gamma_*\approx 0.22$ in agreement with the value $\gamma_*=0.214$ reported in \cite{Yang-2014-39}.
\begin{figure}[ht!]
\includegraphics[width=0.4 \textwidth]{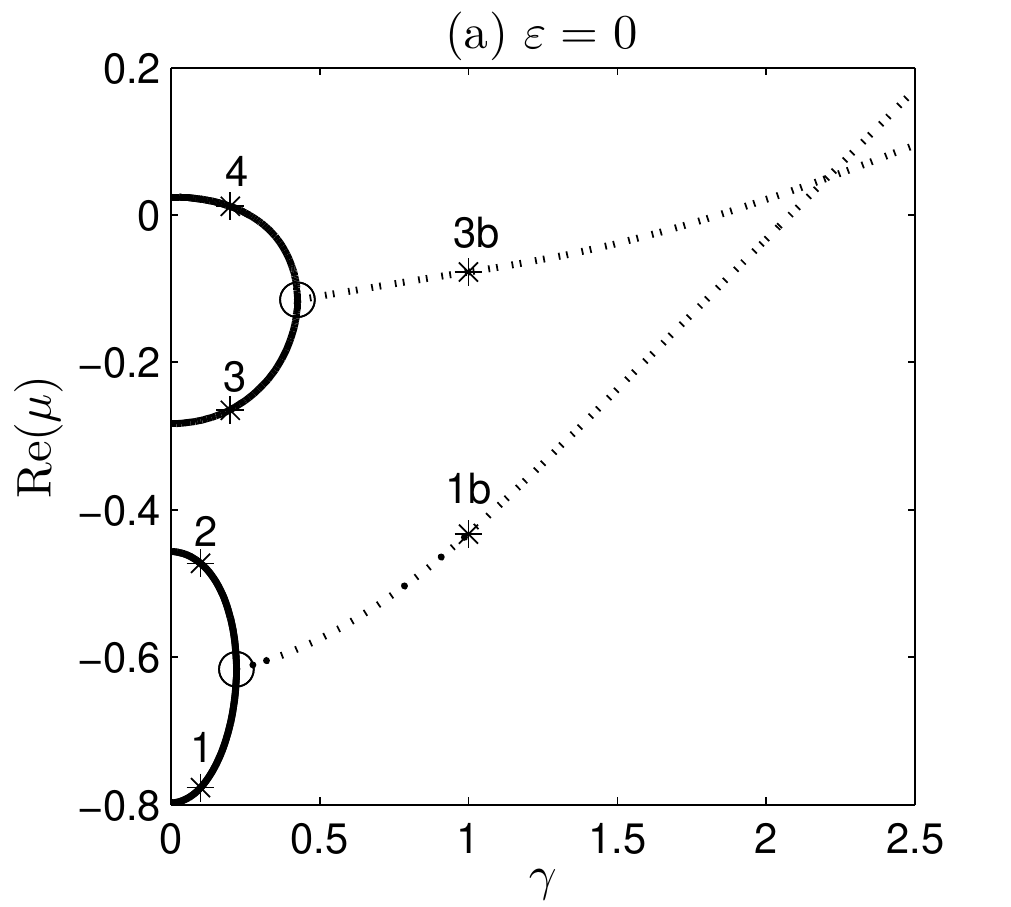}
\includegraphics[width=0.4 \textwidth]{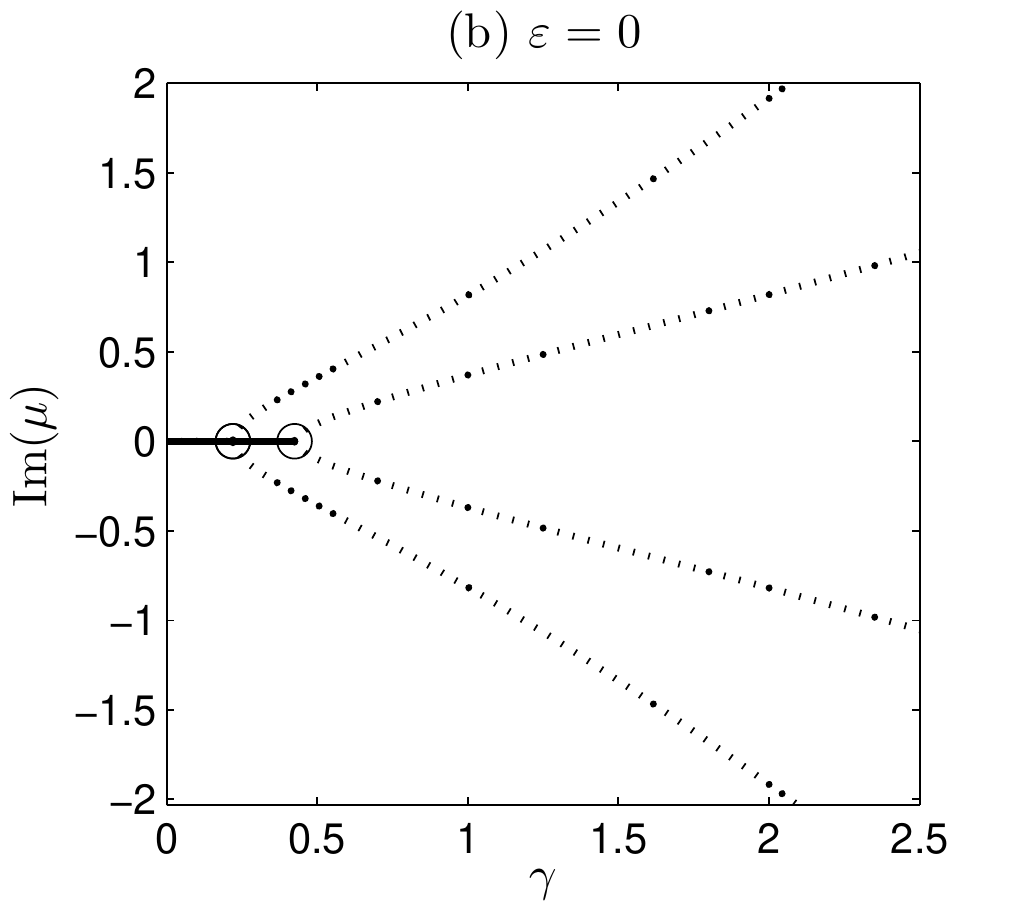}

\includegraphics[width=0.4 \textwidth]{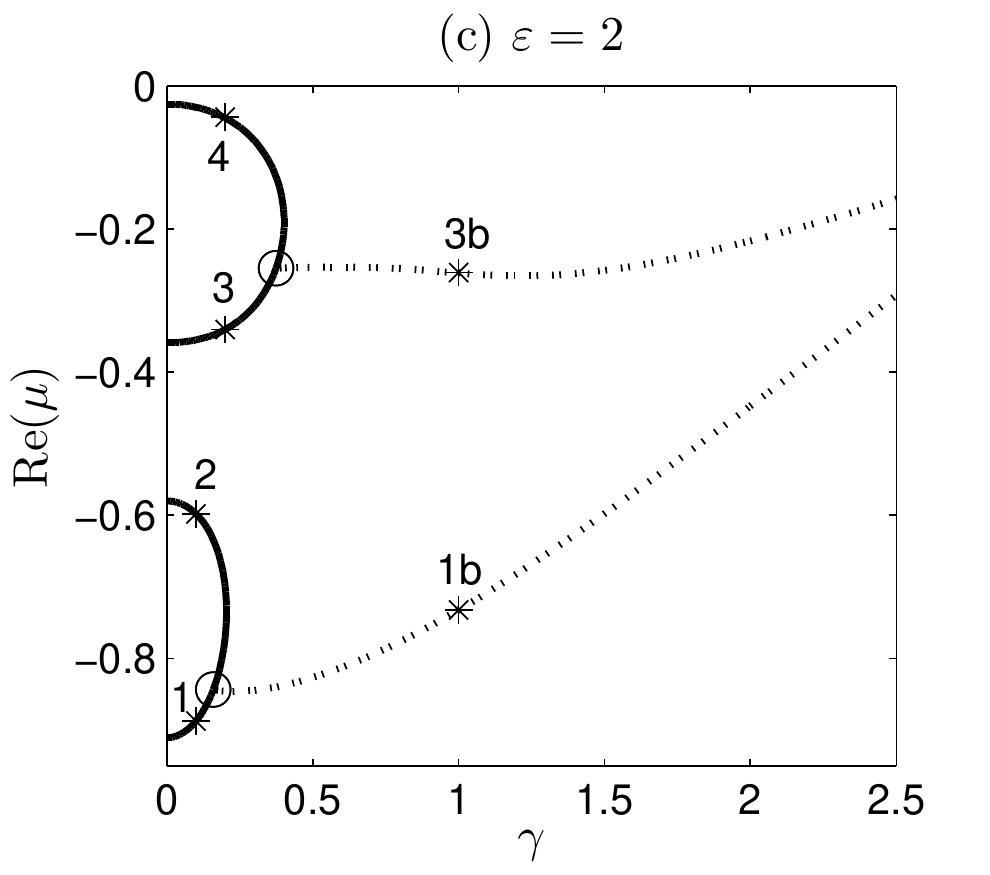}
\includegraphics[width=0.4 \textwidth]{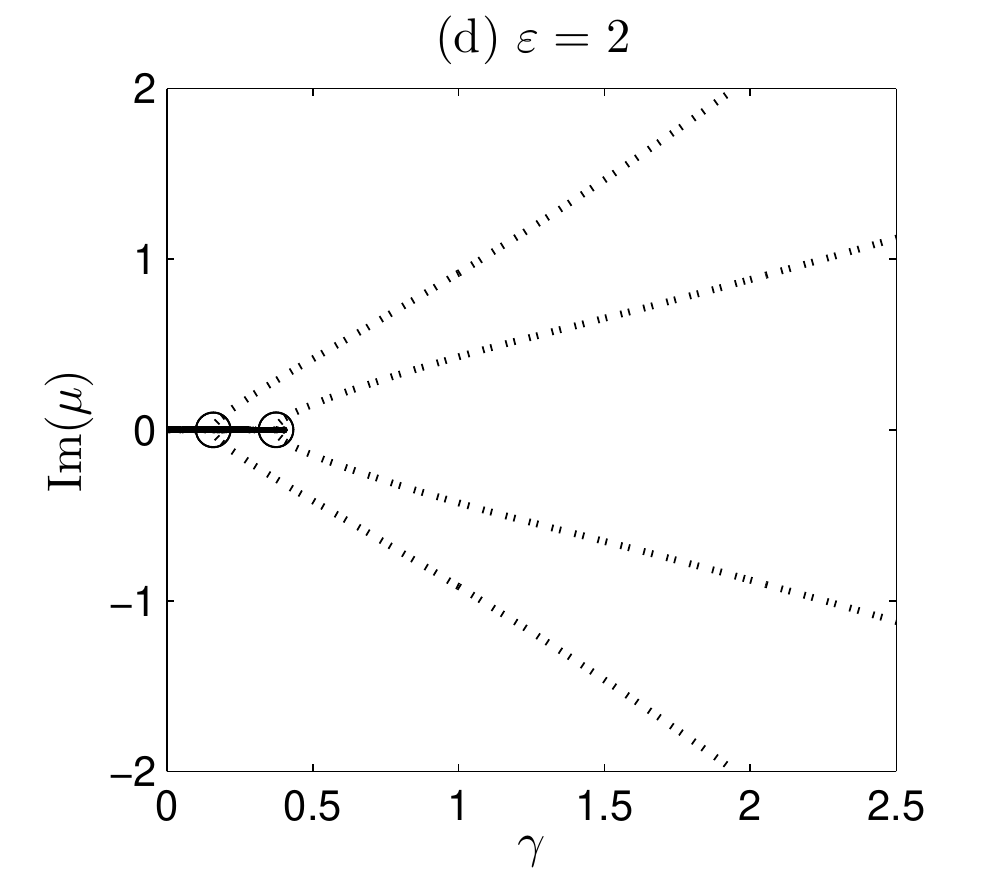}
\caption{Bifurcation diagram in the parameter $\gamma$ for the first eigenvalues $\mu_1, \dots, \mu_4$ of of \protect\eqref{nl.ev_schr}, \protect\eqref{E:Gauss9_PT} with $\varepsilon=0,v_0=1$ in (a) and (b) and with $\varepsilon=2,v_0=1$ in (c) and (d). Circles label secondary bifurcation points.}
\label{Fig:Gauss9_BD_gam}
\end{figure}

Figures~\ref{Fig:Gauss9_PT_profs_g0} and \ref{Fig:Gauss9_PT_profs_g2} show the eigenfunctions at the points labeled in the bifurcation diagram in Fig.~\ref{Fig:Gauss9_BD_gam}. The eigenfunctions at points 1--4 (before eigenvalue collision) are all $\P_1\T$-symmetric while after the collision at points 1b and 3b they are asymmetric. This is in contrast with Example \ref{ex:num_PT}, where linear symmetry was preserved in collisions. Here, no obvious linear symmetry is available.
\begin{figure}[ht!]
\includegraphics[scale=0.42]{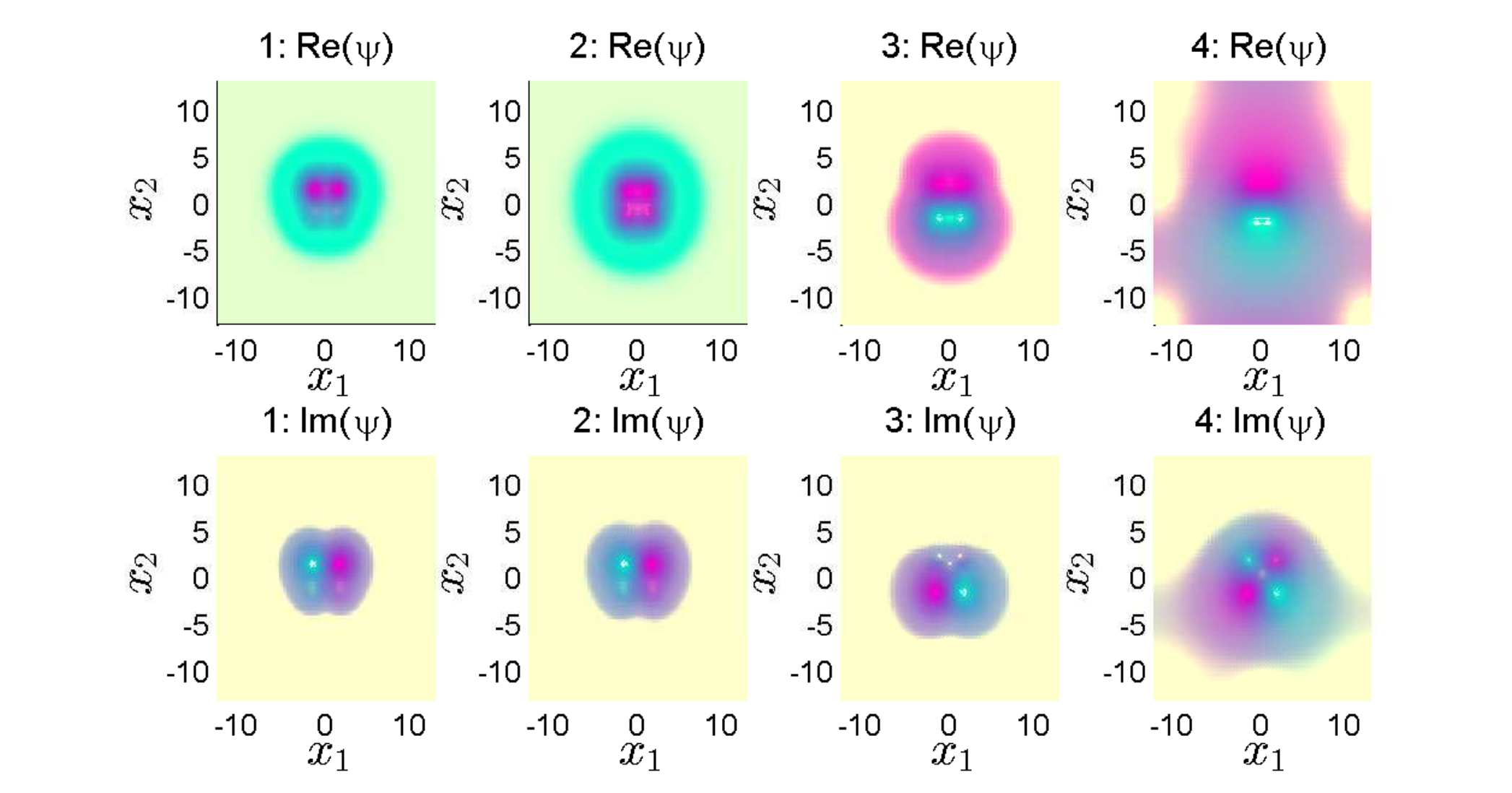}\hspace{-1cm}
\includegraphics[scale=0.42]{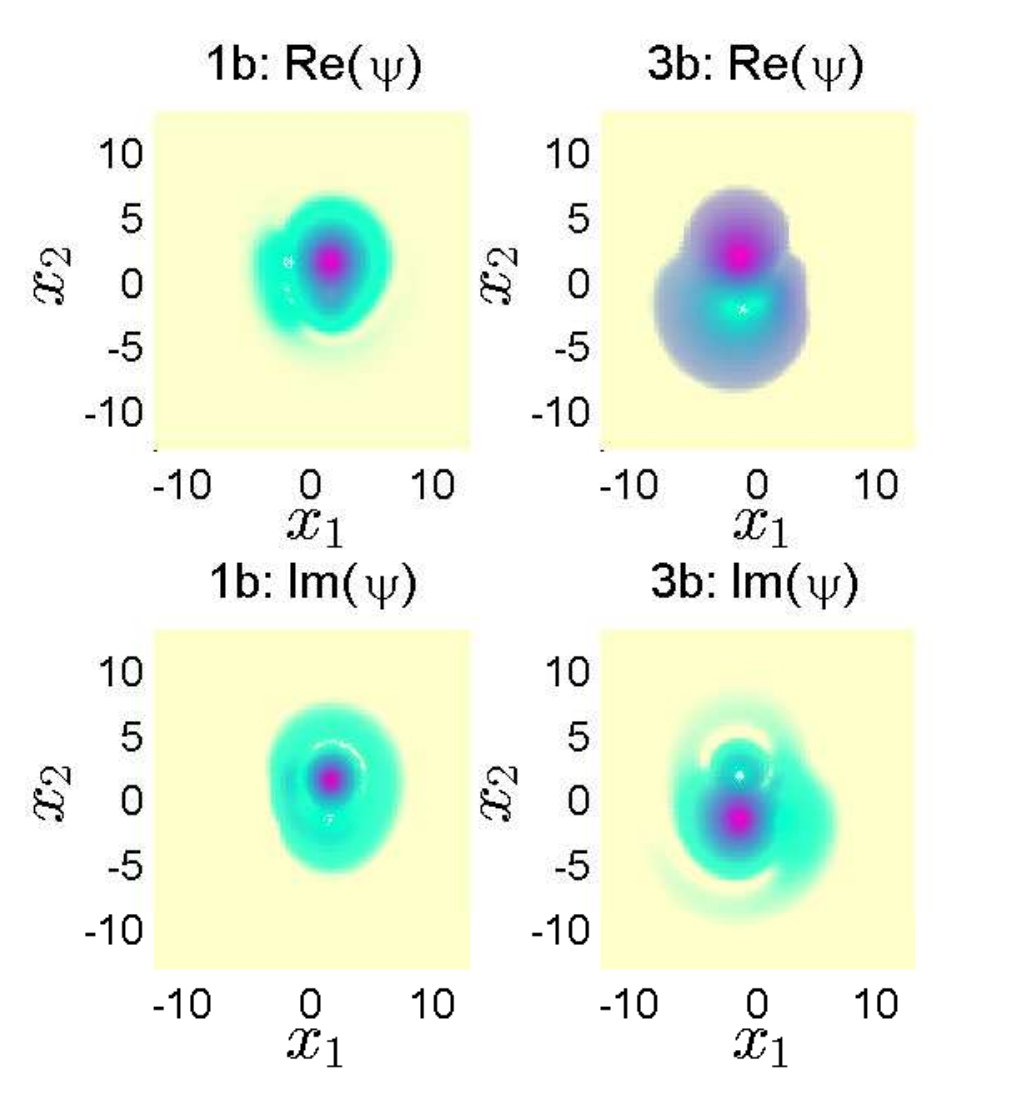}
\caption{Profiles of the nonlinear eigenfunctions of \protect\eqref{nl.ev_schr}, \protect\eqref{E:Gauss9_PT} at $\varepsilon=0,v_0=1$ labeled by 1--4 and 1b, 3b in Fig.~\protect\ref{Fig:Gauss9_BD_gam} (a). }
\label{Fig:Gauss9_PT_profs_g0}
\end{figure}
\begin{figure}[ht!]
\includegraphics[scale=0.42]{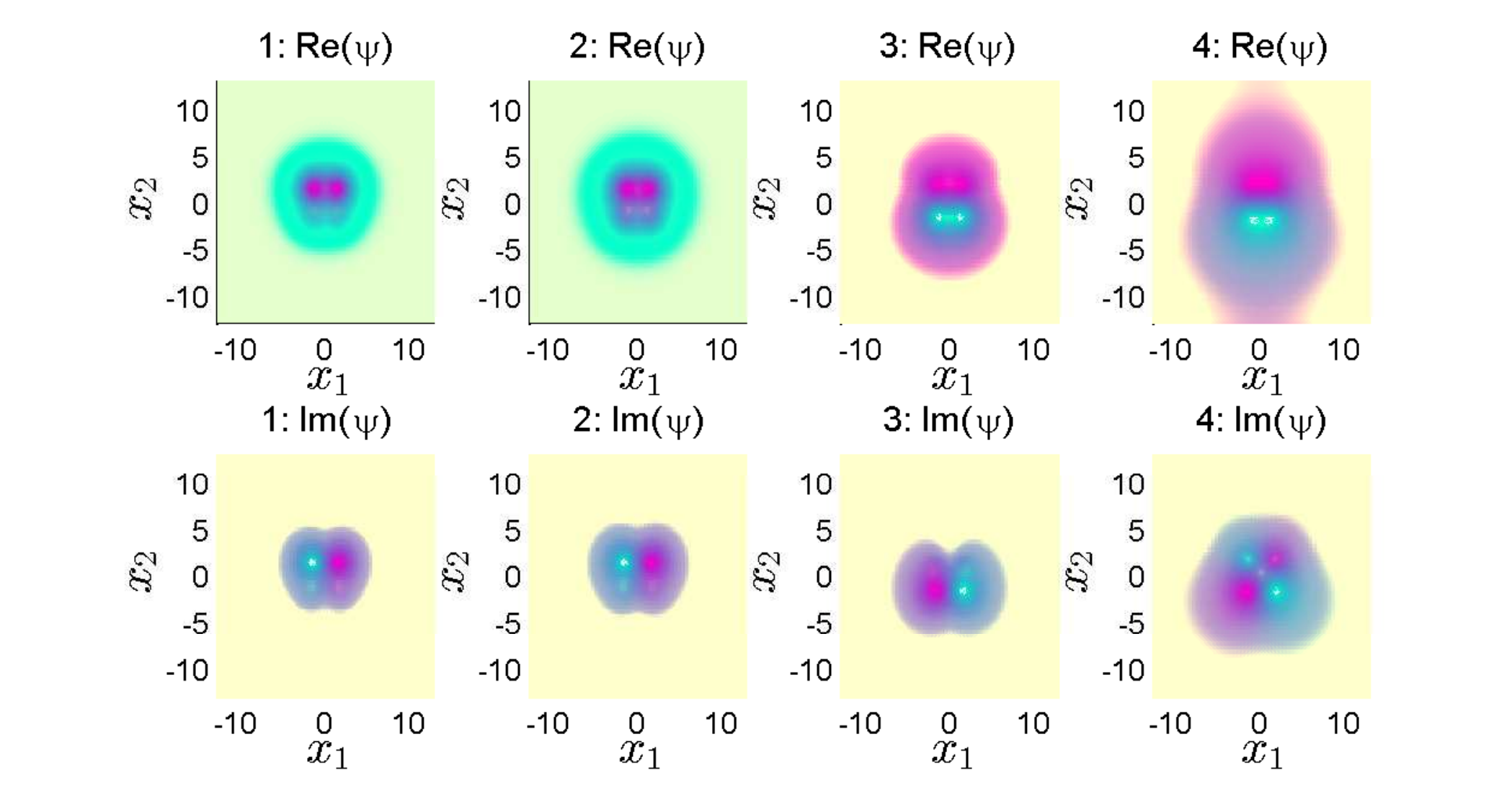}\hspace{-1cm}
\includegraphics[scale=0.42]{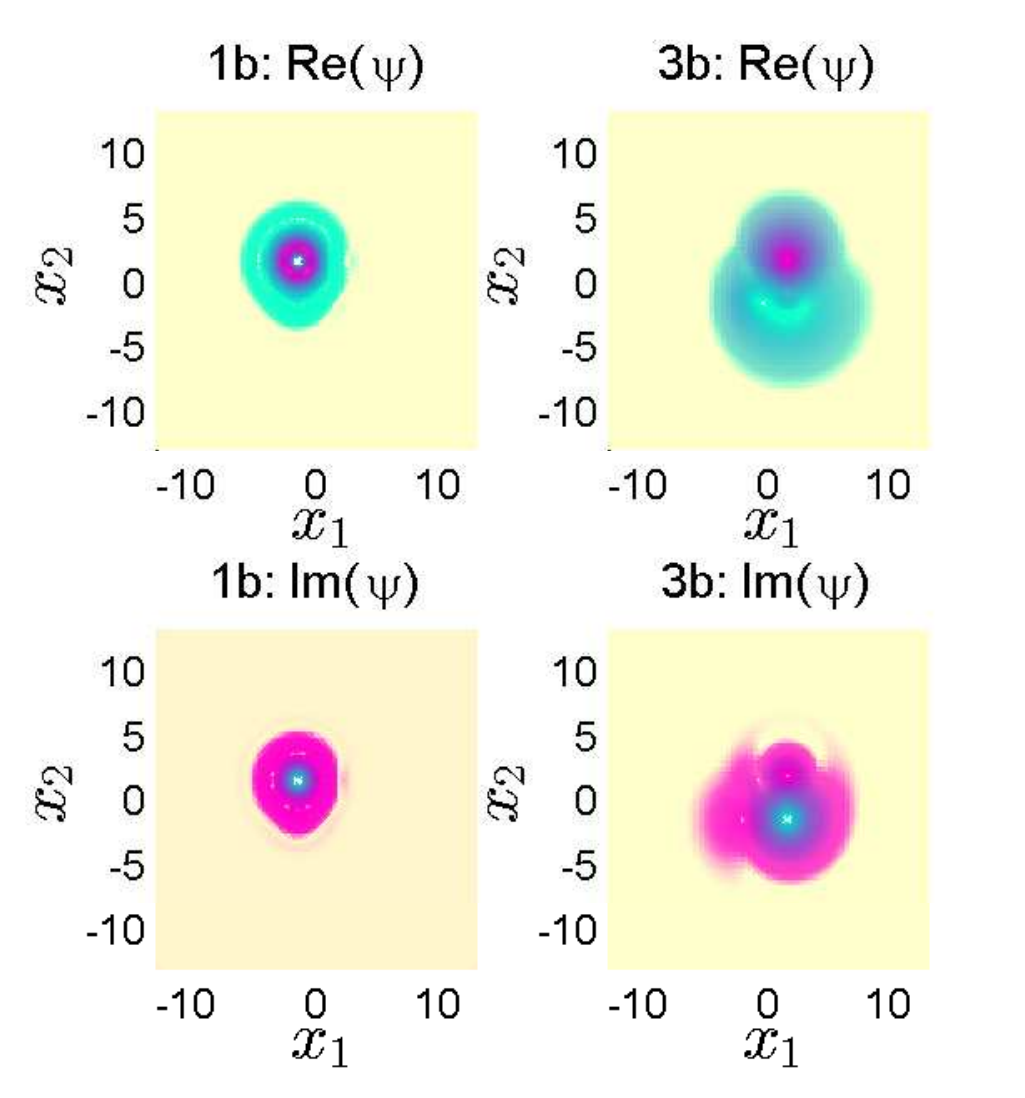}
\caption{Profiles of the nonlinear eigenfunctions of \protect\eqref{nl.ev_schr}, \protect\eqref{E:Gauss9_PT} at $\varepsilon=2,v_0=1$ labeled by 1--4, 1b and 3b in Fig.~\protect\ref{Fig:Gauss9_BD_gam} (c). }
\label{Fig:Gauss9_PT_profs_g2}
\end{figure}
\end{example}


{\footnotesize
\bibliographystyle{acm}
\bibliography{references}
}

\end{document}